\documentclass[leqno]{amsart}
\usepackage{amsmath,amsthm,amsfonts,amssymb}

\numberwithin{equation}{section}
\theoremstyle{plain}

\usepackage{graphicx}
\usepackage{float}
\usepackage[suffix=]{epstopdf}
\usepackage{hyperref}

\newtheorem{theorem}{Theorem}

\newtheorem{remark}[theorem]{Remark}

\begin{document}

\title{A reduced-form model for level-1 limit order books}

\author{Tzu-Wei Yang}
\address{School of Mathematics, University of Minnesota}
\email{yangx953@umn.edu}

\author{Lingjiong Zhu}
\address{Department of Mathematics, Florida State University}
\email{zhu@math.fsu.edu}

\date{}

\subjclass[2010]{91G99, 62P05}
\keywords{Limit order books, data analysis, reduced form models, diffusion approximations}

\begin{abstract}
One popular approach to model the limit order books dynamics
of the best bid and ask at level-1 is to use the reduced-form diffusion approximations.
It is well known that the biggest contributing factor to the price movement is the imbalance
of the best bid and ask. We investigate the data of the level-1 limit order books of a basket of stocks
and study the numerical evidence of drift, correlation, volatility and their dependence on the imbalance.
Based on the numerical discoveries, we develop a nonparametric discrete model
for the dynamics of the best bid and ask, which can be approximated by a reduced-form model 
that incorporates the empirical data of correlation and volatilities with analytical tractability
that can fit the empirical data of the probability of price movement.
\end{abstract}
\maketitle
\section{Introduction}

The traditional human traders have largely been replaced by the automatic and electronic traders
in today's financial world. The role and the controversy of those high frequency traders
have caught public's attention ever since the infamous flash crash on May 6, 2010
and the long-standing debate over the fairness of equity markets briefly became salient
with the recent publication of the book ``Flash Boys'' by Michael Lewis, who argued
that the trading has become unfair and skewed by the high-frequency trading and dampened
the opportunities of the regular investors.

In automatic and electronic order-driven trading platforms, 
orders arrive at the exchange and wait in the limit order book. There are two types of orders
in the limit order book: market orders and limit orders. Cancellation is also allowed.
One of the key research areas in limit order books has been centered around modeling the limit order book dynamics.
In this paper, we only consider the limit order book model at level-1, that is, we only study
the dynamics of the volumes at the best bid and the best ask.

The limit order book is a discrete queuing system and many of the works in the literature model
study the dynamics of the limit order book in a discrete setting directly, see e.g. Cont \textit{et al.} \cite{Cont2010}, Abergel and Jedidi \cite{ABERGEL2013}.
Another popular approach, is to study the reduced-form of the discrete model. 
In the sense of heavy traffic limits, various authors, see e.g. Cont and de Larrard \cite{Cont2012, Cont2013}, 
Avellaneda \textit{et al.} \cite{Avellaneda2011}, Guo \textit{et al.} \cite{Guo2015}
considered the diffusion limit as an approximation of the discrete model. 
The diffusion approximation is valid if the average queue sizes are much larger than the typical quantity
of shares traded and the frequency of orders per unit time is high, 
see e.g. the discussions in Avellaneda \textit{et al.} \cite{Avellaneda2011}.
With the empirical finding of approximate scale invariance, Bouchaud \textit{et al.} \cite{Gareche2013} derive a two dimensional Fokker-Planck equation describing
the statistical behavior of the queue dynamic.
In the recent work by Huang \textit{et al.} \cite{HLR}, they introduced a model which accommodates the empirical
properties of the full order book and the stylized facts of lower frequency financial data. In their model, the
order flows have state-dependent intensities.
We refer to the recent book \cite{LL} for more details.

One research area of great interests is the dynamics of the limit order books and how
it influences the stock price movement, 
see e.g. Avellaneda \textit{et al.} \cite{Avellaneda2011}, Cont \textit{et al.} \cite{Cont2010}, Huang and Kercheval \cite{Huang2012} etc.
There is strong empirical evidence to suggest the biggest factor that drives the movement for the stock price
to the next level is the \emph{imbalance} of the best bid and the best ask, see e.g. Avellaneda \textit{et al.} \cite{Avellaneda2011},
which is defined
as the ratio of the volume at the best bid and total volume at the best bid and ask:
\begin{equation}
\text{Imbalance}=\frac{\text{Volume at Best Bid}}{\text{Volume at Best Bid}+\text{Volume at Best Ask}}.
\end{equation}
In the limit order book, the stock price will move up when the best ask queue is depleted
and the price will move down when the best bid queue is depleted.
The empirical data suggests that the probability that the stock price will move up
increases as the imbalance increases. One can think the probability of price moving up
as a monotonically increasing function of the current imbalance. 
One may expect that the probability of the stock price moving up
is a monotonic function from $0$ to $1$ as the imbalance increases from $0$ to $1$.
But empirical evidence suggests otherwise. In Avellaneda \textit{et al.} \cite{Avellaneda2011}, they discovered
that even though the probability of the stock price moving up is indeed an increasing function
of the imbalance, it increases from a positive value to a value less than one.
One explanation is the hidden liquidity, that is, the sizes that are not
shown in the limit order book, see \cite{Avellaneda2011}. As it is hypothesized in \cite{Avellaneda2011}, 
there can be two explanations for hidden liquidity. First, markets are fragmented
and it can happen that once the best ask on an exchange is depleted, the price will not
necessarily go up since an ask order at that price may still be available on another market and a new
bid cannot arrive until that price is cleared on all markets. Second, there exist so-called iceberg orders,
the trading algorithms that split large orders into smaller ones that refill the best quotes
as soon as they are depleted.

Indeed, we also discover the hidden liquidity in our numerical analysis
and we use the idea of the hidden liquidity to better fit the model. 
The numerical evidence suggests that empirical probability of price moving up
depends linearly on the imbalance, with hidden liquidity at the very small and large
imbalance levels, see Figures \ref{ProbBAC}, \ref{ProbGE}, \ref{ProbGM}, \ref{ProbJPM}.
In \cite{Avellaneda2011}, correlated Brownian motions are used as a reduced-form model to describe the dynamics at
the best bid and ask queues. The linear dependence of the empirical probability of price moving up
on the imbalance level suggests that in the correlated Brownian motions model, the correlation
should be exactly $-1$. However, we carried out numerical analysis to study 
the correlation and volatilities of the best bid and ask sizes and their dependence
on the imbalance, and found out that the correlation is negative, but far away from $-1$
and it is also dependent on the level of the imbalance. Therefore, the correlated Brownian motions
might be too simplistic explaining the dynamics of the best bid and ask queues.
In this paper, we will build up a non-parametric model that can fit the data 
of the empirical correlation, empirical volatilities of the best bid and ask sizes and the empirical probability
of price moving up simultaneously. 

In this paper, we will use data to study the level-1 limit order books for a basket of stocks, to further understand the dependence
on the imbalance of the best bid and ask sizes. 
We will look at a basket of stocks and also compare the results amongst different exchanges, 
in particular, NASDAQ and NYSE because our empirical data suggest that the stocks we selected have
the largest trading volumes in these two stock exchanges. We discover that the micro structure and the
dynamics of the limit order books depend on their exchanges, in the sense, that the key statistics like correlation
between the best bid and ask, and the drift effect at the best bid and ask queues can differ across
the exchanges. The discrepancy amongst the exchanges has caught a lot of attention lately.
As noted in a recent article on Wall Street Journal \cite{Macey2015}:
\begin{quotation}
	``There is no question that U.S. equity markets are fragmented. The New York Stock Exchange's share
	of trading in its listed stocks has dropped to 32\% of its volume from 77\% a decade ago...
	This fragmentation... also creates arbitrage opportunities that did not exist when trading markets
	were unified.''
\end{quotation}
In our empirical study, we discover the evidence of discrepancy amongst different exchanges
and also across different stocks. 
This has two possible implications. First, the discrepancy can possibly be explained
by the different trading patterns of different algorithmic traders. Say we have two
high frequency traders A and B, who are trading two different baskets of stocks. 
Then the different behavior of the dynamics of the limit order books of the stocks may be due
to the different trading strategies and patterns of these two different players. 
As we will see later from our empirical studies, different stocks are concentrated in different
exchanges. Hence the different trading strategies of the traders behind different stocks
can result in the discrepancy across the different exchanges. Second, the fragmentation of the stock exchanges
may intrinsically cause the difference of the dynamics of the limit order books on different exchanges,
especially when the imbalance of the best bid and ask is either small or large, that is when the queues
at the best bid and ask are near depletion. In these cases, new orders may be directed
to a different exchange when liquidity is still available. Thus the fragmentation of the stock exchanges
and the discrepancy of the limit order books dynamics across different exchanges may create
arbitrage opportunities. As the same article on WSJ pointed out,
\begin{quotation}
	``Transparency disappears
	behind a shroud of complex order types executed on vaguely sinister dark pools, trading venues
	that sometimes are used to disadvantage long-term investors... The remedy is to create multiple
	trading venues and then limit trading in a particular security to one of them.''
\end{quotation}

The paper is organized as follows. In Section \ref{DataSection}, we carry out the data analysis
to study the empirical evidence of the dependence of the dynamics at the best bid and ask on the imbalance. 
Based on the empirical evidence, we build a nonparametric reduced-form model in Section \ref{ReducedSection}
with analytical tractability, hence extending the existing reduced-form and diffusion approximation approach
of the level-1 limit order book dynamics in the literature. The conclusion of the paper is in Section \ref{sec:conclusion}.
Finally, the technical proofs are in Appendix \ref{ProofSection} and the tables are in Appendix \ref{sec:table}.
\section{Data Analysis}\label{DataSection}

For our data analysis, we use the consolidated quotes of the NYSE-TAQ data set 
from the Wharton Research Data Services (WRDS). We look at the level-1 data, that is, the best bid price,
best ask price, best bid size and best ask size for a basket of stocks
traded on different exchanges. The time window of the data set we selected 
is the first five trading days of 2014, that is, January 2nd, 3rd, 6th, 7th and 8th. 
We only consider the trades happened between 10am and 4pm of the trading days.
We exclude the pre-market and after-market data as well as the data from the first half an hour
of each trading day, i.e., 9:30-10am since these data are usually quite noisy. 
We concentrate on the studies of the following stocks: Bank of America (BAC), General Electric (GE), 
General Motors (GM) and JP Morgan \& Chase (JPM). These blue-chip stocks have large market capitalization, 
are highly liquid and have large trading volumes. Moreover, the average prices of these stocks are reasonably
small and the bid-ask spreads are narrow (one tick size most of time), so that the data sets are not too noisy.
A sample table of the raw data we are using is given in Table \ref{SampleTable} as an illustration.   In Table \ref{SampleTable}, the units of best bid and ask 
prices are US Dollars and the units of the best bid and ask sizes are 100.

Out empirical observation from the raw data, shows that for most of time, there is only one size of the bid and ask queues changes, while the other size remains the 
same.  Such a pattern would not be an appropriate approximation of a diffusion process that we consider in Section \ref{ReducedSection} and the empirical correlation 
between the bid and ask queue sizes is almost zero.  We therefore average the consecutive data when there is only one queue size varying while the other one remains 
the same.  In this way, the data would have a better approximation to a diffusion process.

\begin{table}
	\caption{An Example of the Raw Data. Source: Wharton Research Data Services (WRDS).}
	\centering 
	\begin{tabular}{|c|c|c|c|c|c|c|c|}
		\hline 
		Ticker &   Date   &   Time   &  Bid  &  Ask  & Bid Size & Ask Size & Exchange \\ \hline\hline
		                                  GM                                   & 20140102 & 10:12:44 & 40.55 & 40.57 &    9     &    33    &    N     \\ \hline
		                                  GM                                   & 20140102 & 10:12:44 & 40.55 & 40.57 &    7     &    33    &    N     \\ \hline
		                                  GM                                   & 20140102 & 10:12:44 & 40.56 & 40.57 &    4     &    4     &    T     \\ \hline
		                                  GM                                   & 20140102 & 10:12:44 & 40.56 & 40.57 &    4     &    6     &    T     \\ \hline
		                                  GM                                   & 20140102 & 10:12:44 & 40.56 & 40.57 &    1     &    2     &    P     \\ \hline
		                                  GM                                   & 20140102 & 10:12:44 & 40.56 & 40.57 &    3     &    6     &    T     \\ \hline
	\end{tabular}
	\label{SampleTable} 
\end{table}

We found out that the largest volumes
for these stocks we studied are traded on NASDAQ (T) and NYSE (N)\footnote{That is not always the case. For example, for the first five trading days of 2014, the exchanges that trade the largest volumes of Walmart (WMT) are NYSE (N) and BATS (Z) and the exchanges that trade the largest volumes of Microsoft (MSFT) are NASDAQ (T) and BATS (Z). For the comparison purposes, we only study those stocks with top two exchanges being NASDAQ and NYSE.}, see Figure \ref{PieCharts}.
The symbols in Figure \ref{PieCharts} stand for the exchanges that the stocks are traded at
and the details are given in Table \ref{ExchangeTable}.

\begin{table}
	\caption{Primary Listed Exchange Codes}
	\centering 
	\begin{tabular}{|c|c||c|c|}
		\hline 
		Code &           Exchange           & Code &      Exchange       \\ \hline\hline
		                                 B                                   &        NASDAQ OMX BX         &  P   &    NYSE Arca SM     \\ \hline
		                                 C                                   &   National Stock Exchange    &  T   &     NASDAQ OMX      \\ \hline
		                                 J                                   & Direct Edge A Stock Exchange &  W   & CBOE Stock Exchange \\ \hline
		                                 K                                   & Direct Edge X Stock Exchange &  X   &   NASDAQ OMX PSX    \\ \hline
		                                 M                                   &    Chicago Stock Exchange    &  Y   &   BATS Y-Exchange   \\ \hline
		                                 N                                   &   New York Stock Exchange    &  Z   &    BATS Exchange    \\ \hline
	\end{tabular}
	\label{ExchangeTable} 
\end{table}


\begin{figure}
	\centering
	\includegraphics[width=0.49\textwidth]{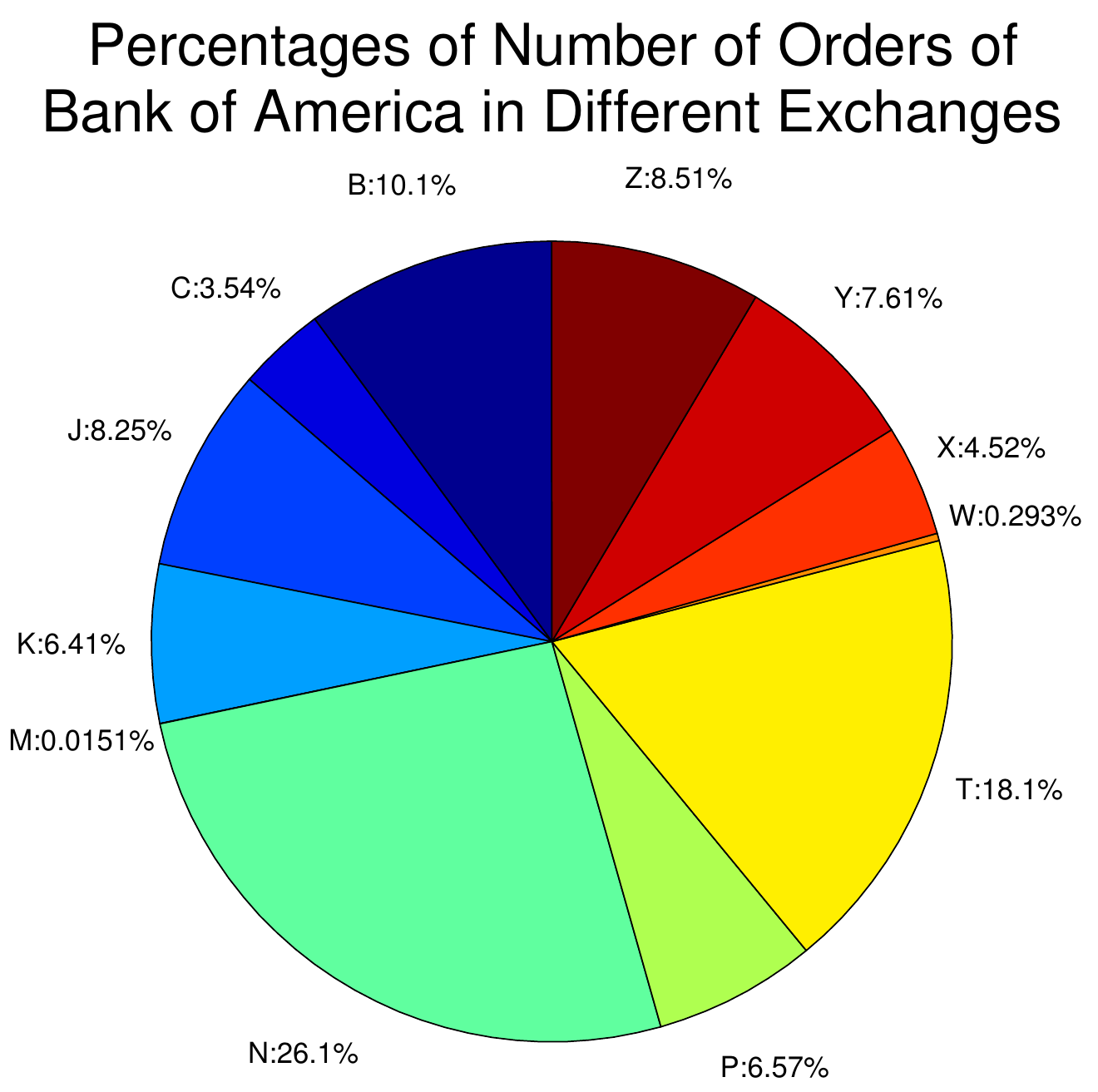}
	\includegraphics[width=0.49\textwidth]{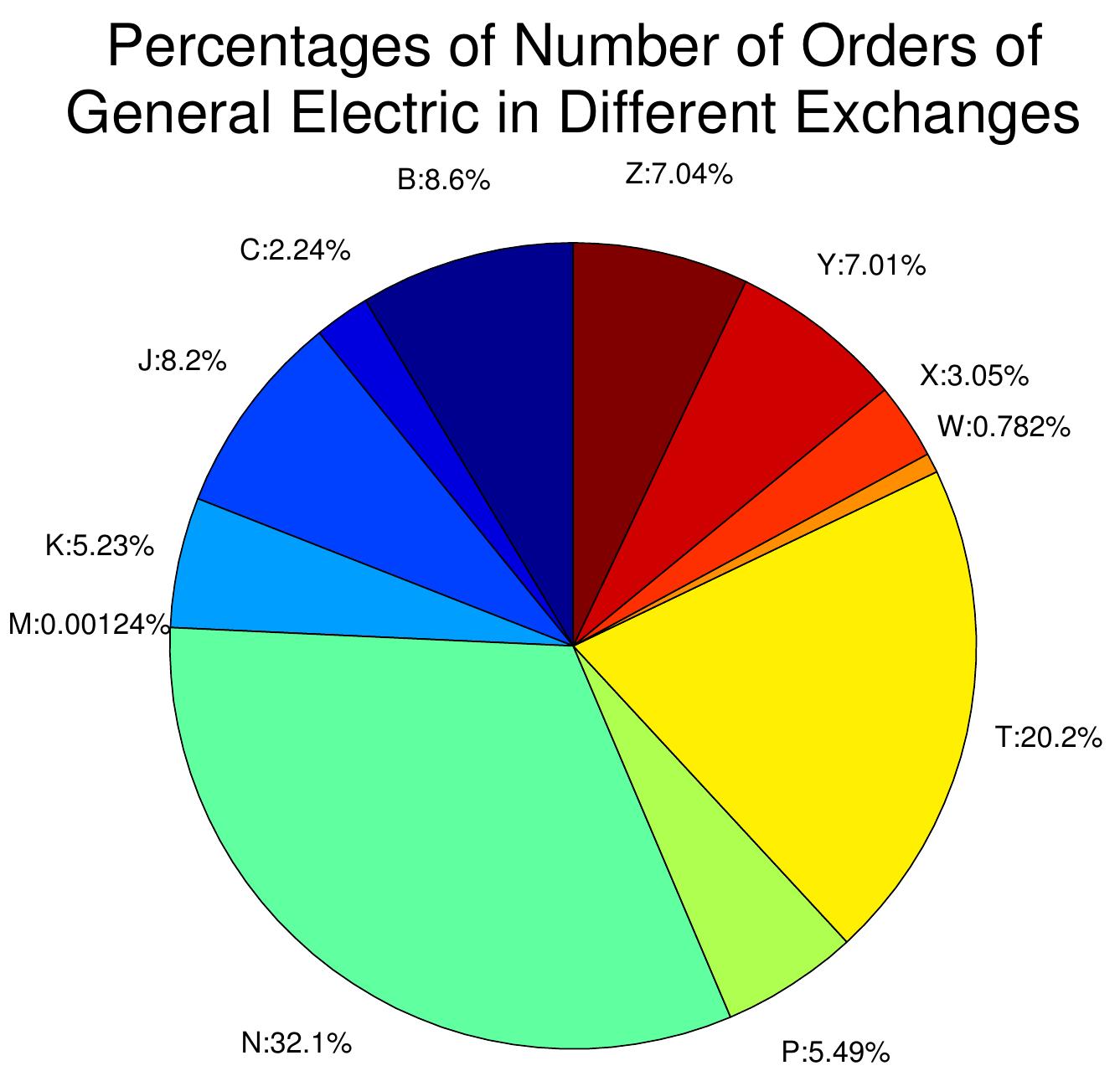}
	
	\includegraphics[width=0.49\textwidth]{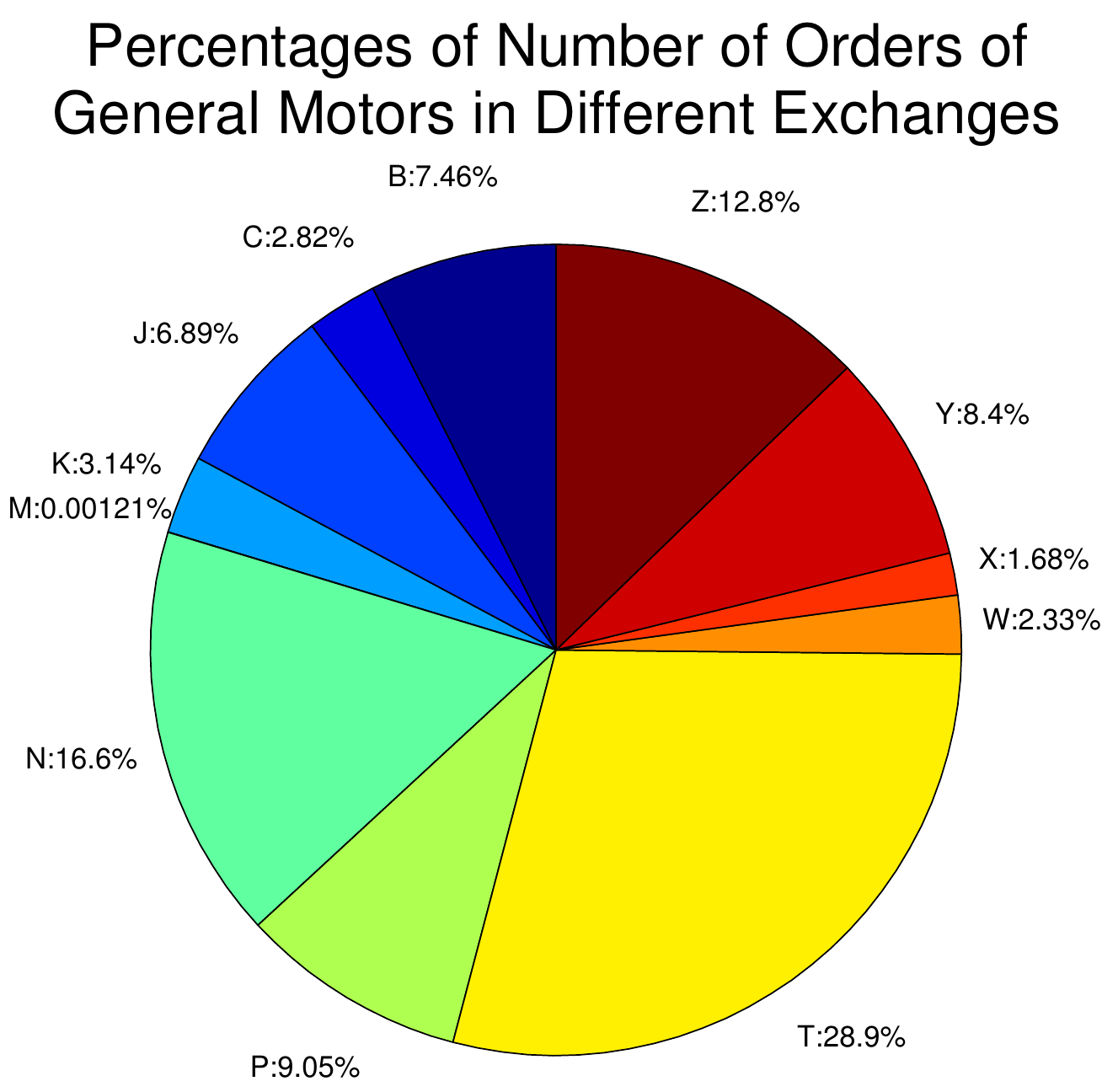}
	\includegraphics[width=0.49\textwidth]{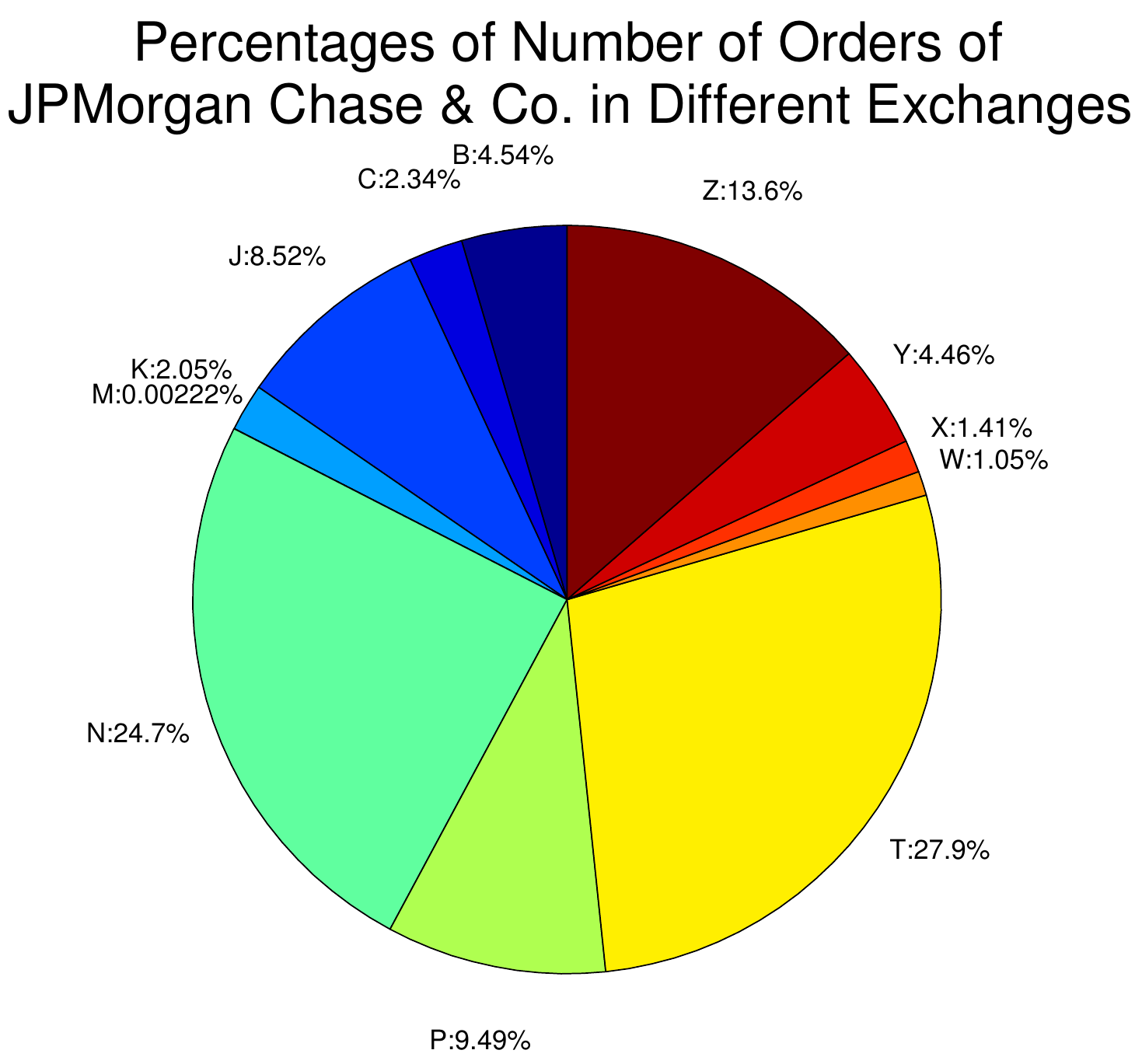}
	\caption{Pie Charts of Bank of America, General Electric, General Motors, and JP Morgan \& Chase}
	\label{PieCharts}
\end{figure}


First, we investigate the drift effect of the best bid and ask queue lengths. 
People have used both driftless diffusions, see e.g. \cite{Avellaneda2011} and diffusions with constant
drift, see e.g. \cite{Cont2013,Guo2015} to model the dynamics of the best bid and ask queues.
Therefore, we are interested in seeing from the data set if there is indeed any evidence
of the drift in the dynamics of the best and bid queues. If there is any evidence of drift, is the drift a constant, 
or a function depending on the queue lengths and the imbalance of the best bid and ask?
Our studies are summarized in Figure \ref{BACBar} for Bank of America, Figure \ref{GEBar} for General Electric,
Figure \ref{GMBar} for General Motors and Figure \ref{JPMBar} for JP Morgan \& Chase.
For example, in Figure \ref{BACBar}, we study the total volumes for the positive changes and the negative changes
at the best bid queues, best ask queues for both NASDAQ and NYSE for the Bank of America stock. 
In all these plots, the purple bars stand for the total volume of the negative changes at a particular imbalance level
and the yellow bars stand for the total volume of the positive changes at a particular imbalance level.
The red lines denote the ratio of the total volume of the positive changes
to the total volume of all the changes.

We can see that when the imbalance is neither too small or too large, there is little evidence
of the drift in the best bid and ask queues.

On the other hand, when the imbalance is very small or very large, we do observe drifts.
From the top left picture in Figure \ref{BACBar}, it is clear that there is evidence of negative drift
at the best bid queue when the imbalance is small (and hence the queue length is short)
and little evidence of drift otherwise. From the top right picture in Figure \ref{BACBar},
it is clear that there is evidence of negative drift
at the best ask queue when the imbalance is small (and hence the queue length is short)
and little evidence of drift otherwise. On the other hand, from the bottom two pictures in Figure \ref{BACBar},
we observe that exactly the opposite is true for the Bank of America stock traded on NYSE, that is,
there is positive drift at the best bid queue when the imbalance is small and at the best ask queue
when the imbalance is large, although the drift effect is weak.
One possible explanation is that there can be different scenarios when the queue lengths are short.
For example, it can happen that the queue length is short when the traded stock is about to move
to the next price level. There is the clustering effect of market orders and cancellations
of the limit orders that can explain the negative drift we observed in the top two pictures in Figure \ref{BACBar}.
On the other hand, it is also possible that the queue length is short because it is a new queue
and there is clustering effect of the arrivals of new limit orders at the new queue, which
results in the positive drift we observed in the bottom two pictures in Figure \ref{BACBar}. 
Similar patterns are also observed for the General Electric stock, see Figure \ref{GEBar}.
On the other hand, we see in Figure \ref{GMBar} that for the General Motors stock, 
for both NASDAQ and NYSE, there is positive drift at the best bid queue when the imbalance is small and at the best ask queue
when the imbalance is large and there is little drift otherwise. Similar patterns also hold
for the JP Morgan \& Chase stock, see Figure \ref{JPMBar}.

The statistics are summarized in Table \ref{TTable} and Table \ref{NTable}. 
BAC b and BAC a stand for the best bid and the best ask queues for BAC respectively. 
The number in each cell is obtained by computing
\begin{equation*}
	\frac{\text{Volume of Positive Change}}{\text{Volume of Positive Change}+\text{Volume of Negative Change}}
\end{equation*} 
From Table \ref{TTable} and Table \ref{NTable}, we can see that except when the imbalance is small or large,
there is little evidence of the drift in the best bid and ask queues.
In terms of modeling, this suggests that we can build up a model with no drift effect
when the imbalance is neither too small or too large. With the small and large imbalance, 
instead of modeling the drift effect by adding a drift term in the dynamics, 
we use the idea of the hidden liquidity from \cite{Avellaneda2011}
to better fit the model and explain the dynamics at the best bid and ask queues.

\begin{figure}
	\centering
	\includegraphics[width=0.49\textwidth]{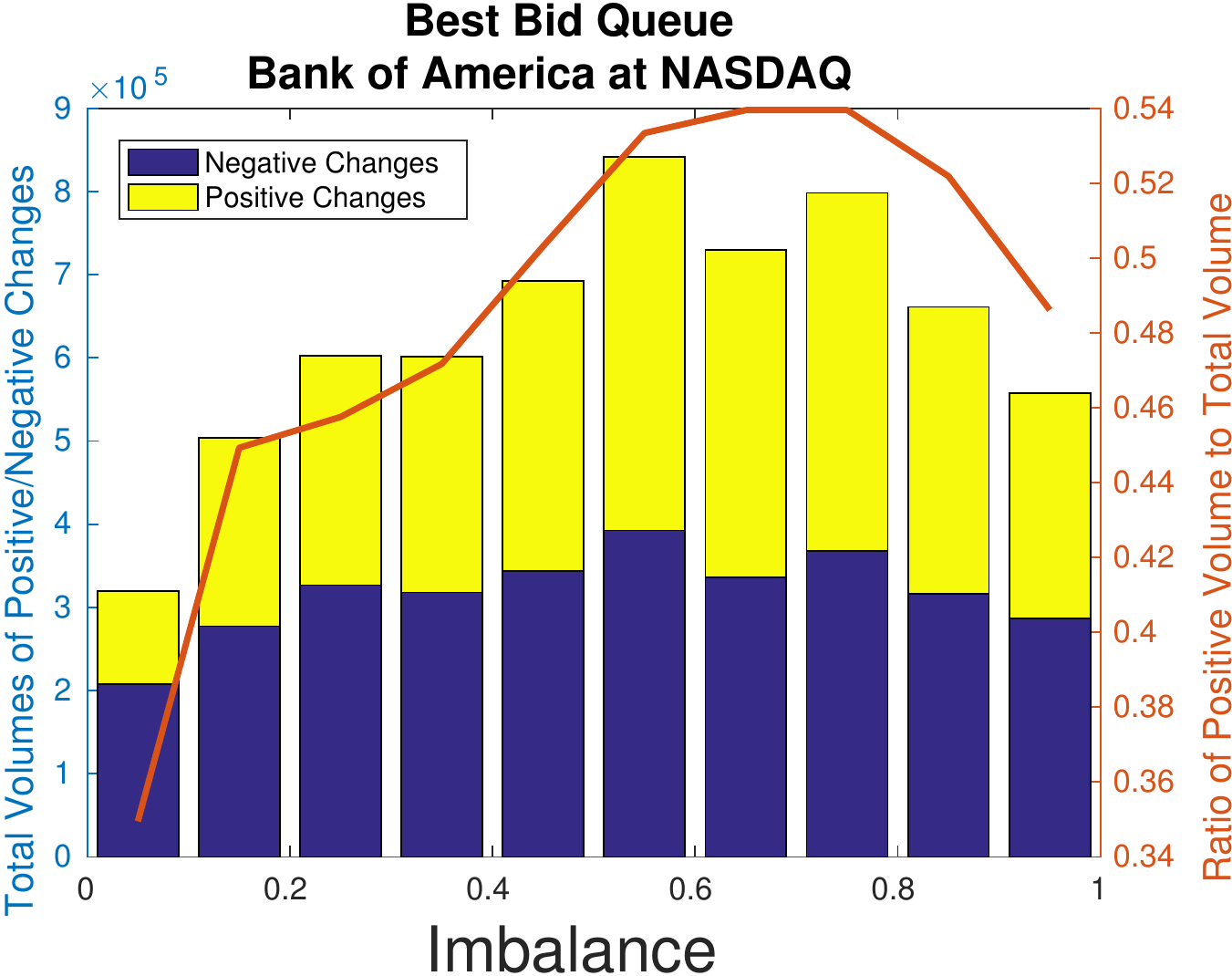}
	\includegraphics[width=0.49\textwidth]{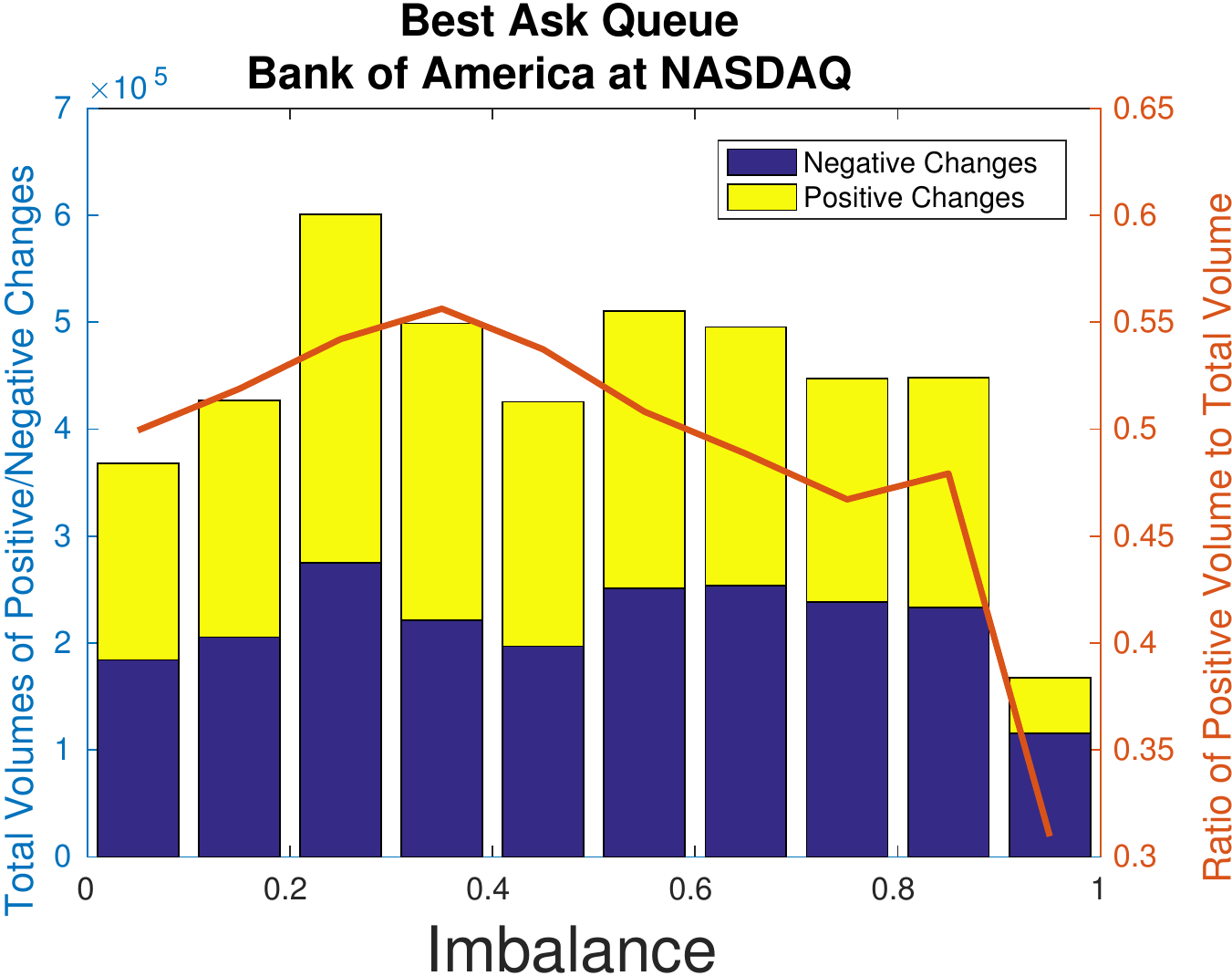}
	
	\includegraphics[width=0.49\textwidth]{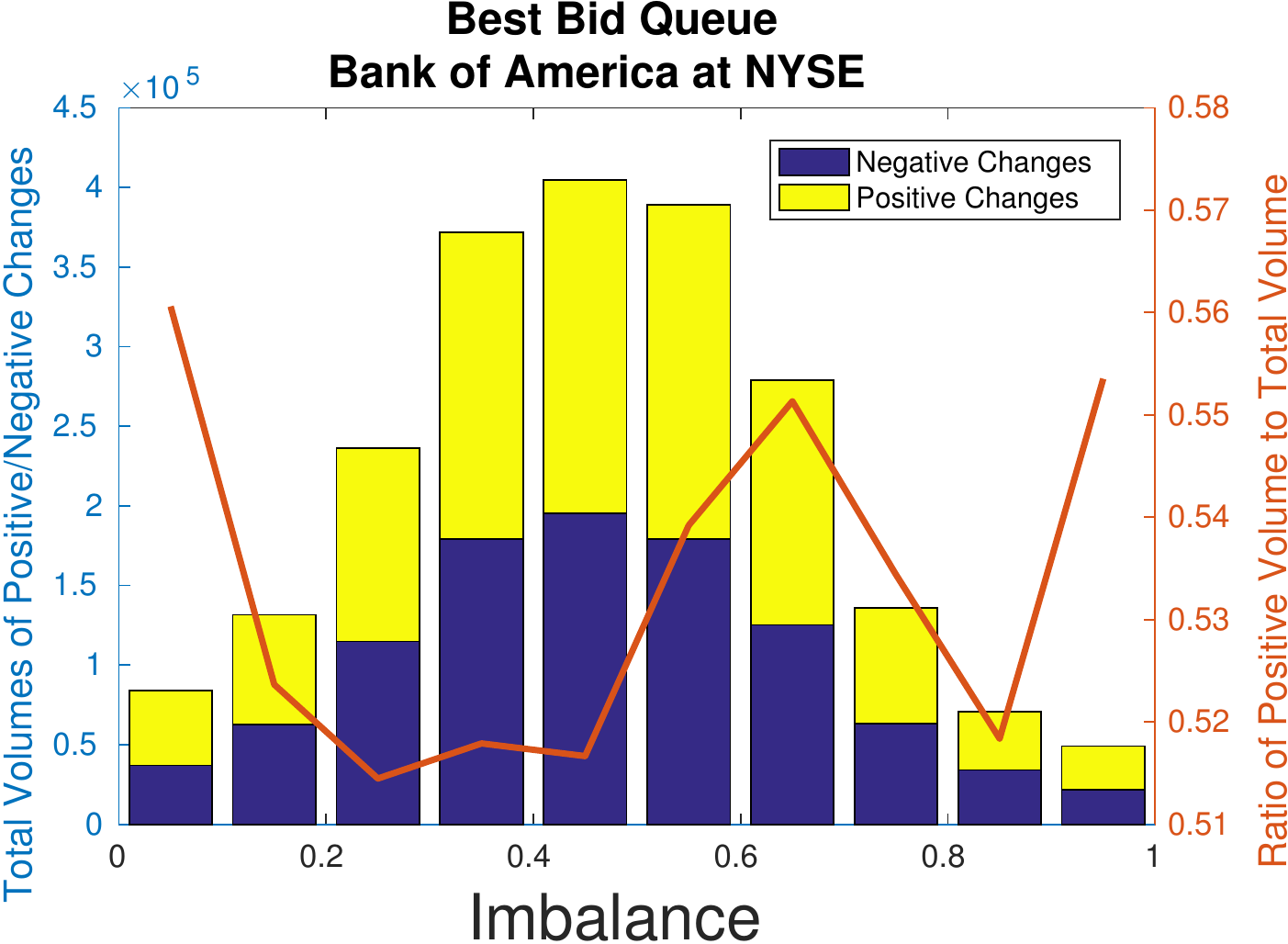}
	\includegraphics[width=0.49\textwidth]{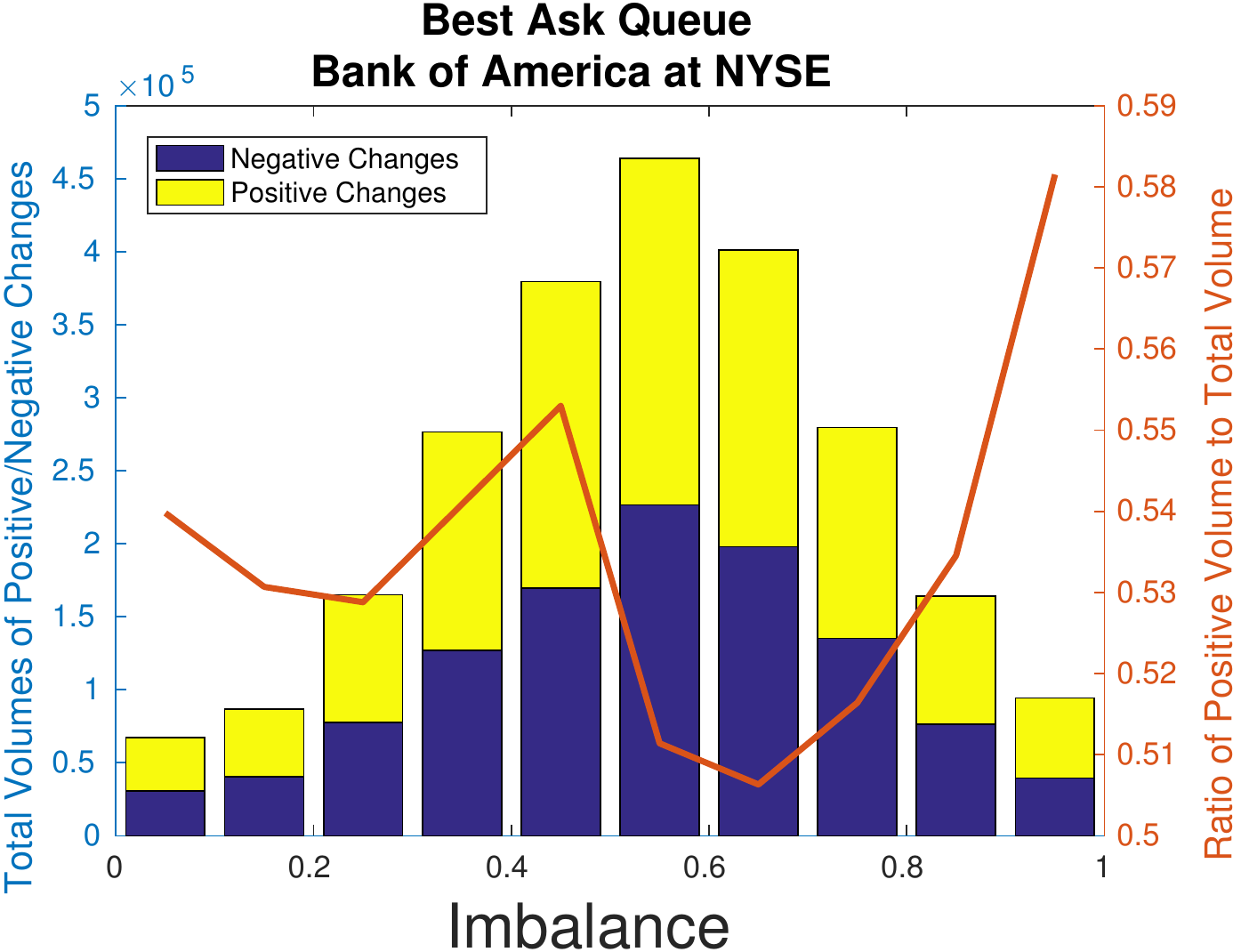}
	
	\caption{Positive and Negative Changes of the Volumes at the Best Bid and the Best Ask of Bank of America at NASDAQ and NYSE.  
	The curve is the ratio of the positive changes to the total changes.}
	\label{BACBar}
\end{figure}

\begin{figure}
	\centering
	\includegraphics[width=0.49\textwidth]{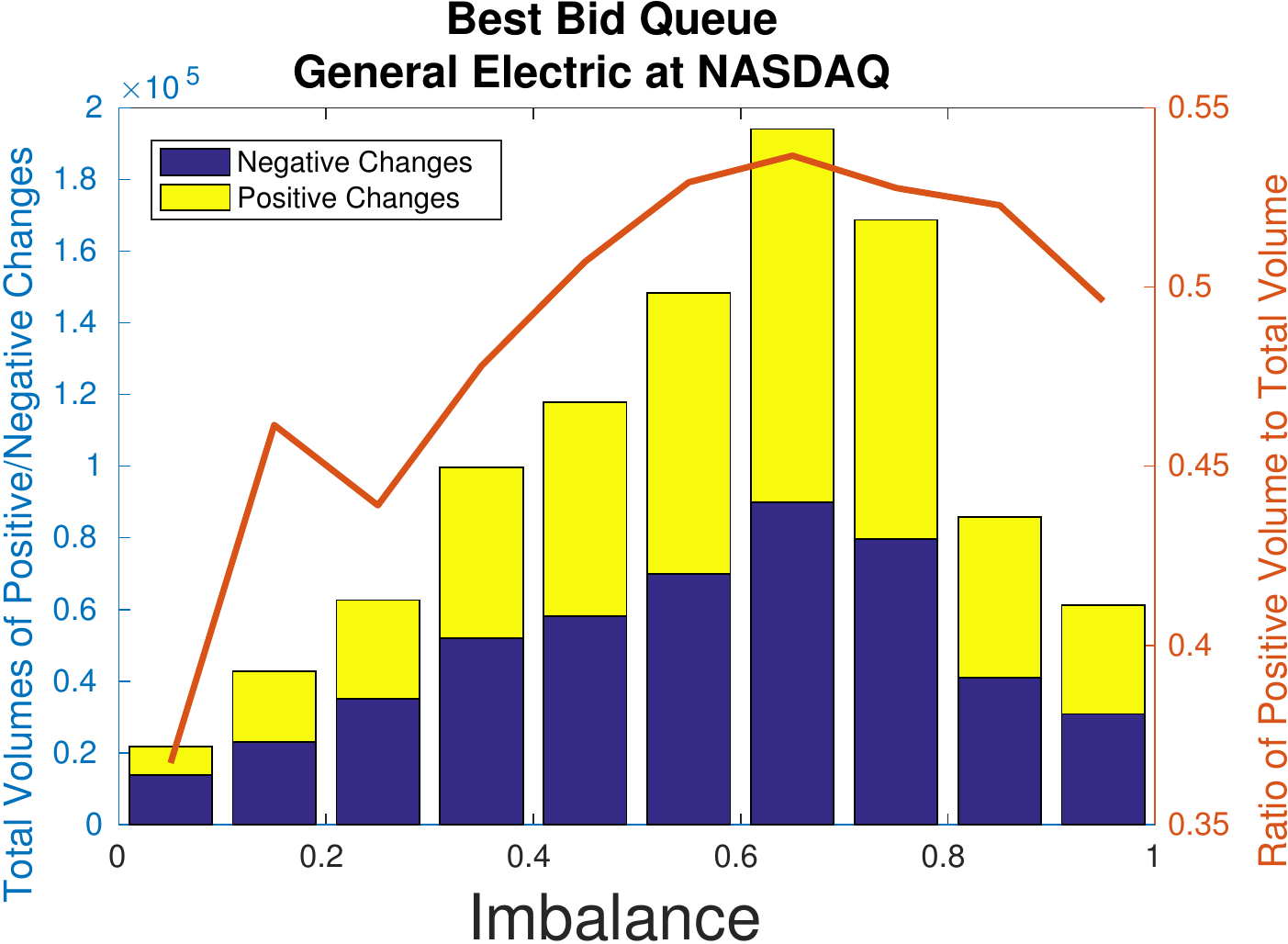}
	\includegraphics[width=0.49\textwidth]{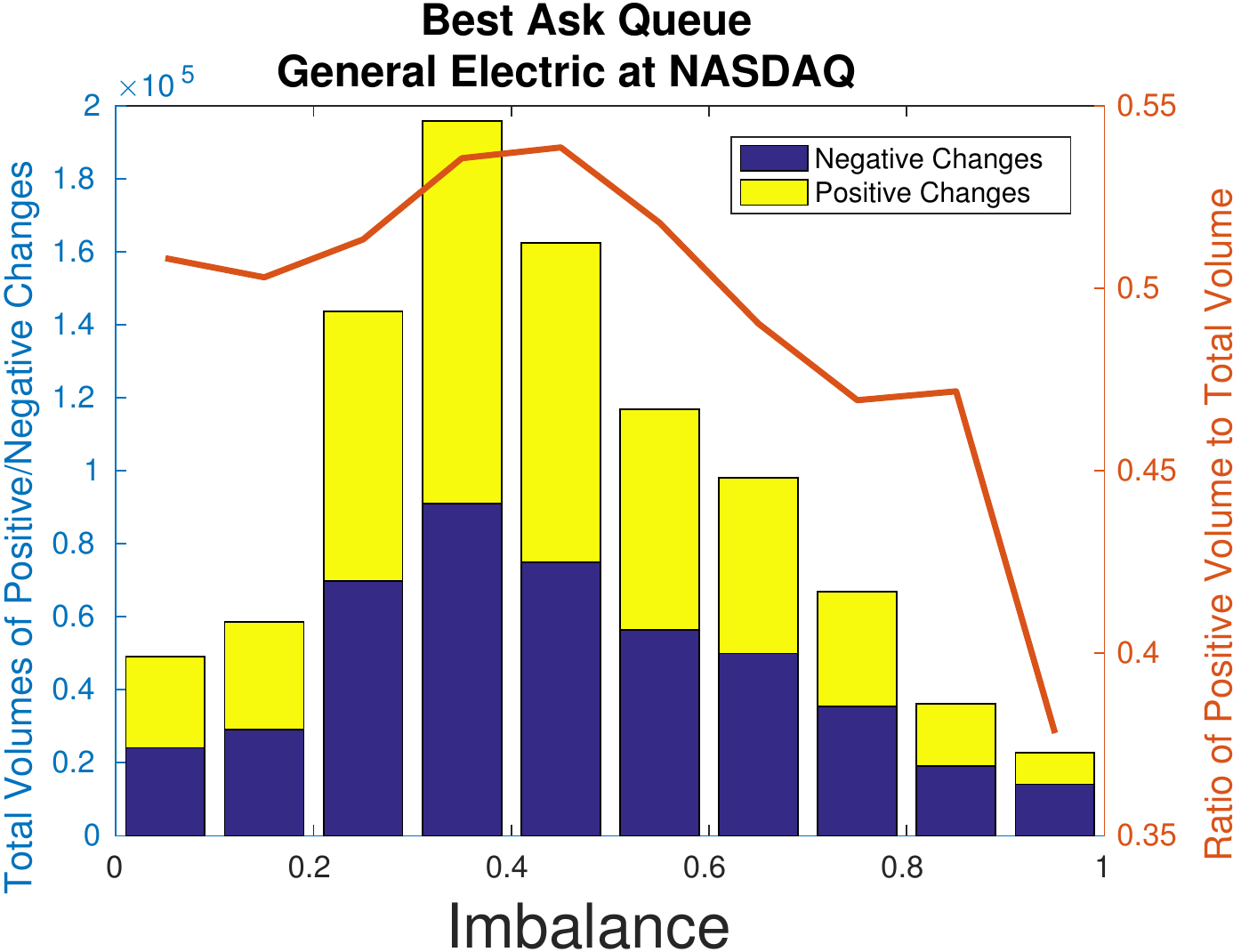}
	
	\includegraphics[width=0.49\textwidth]{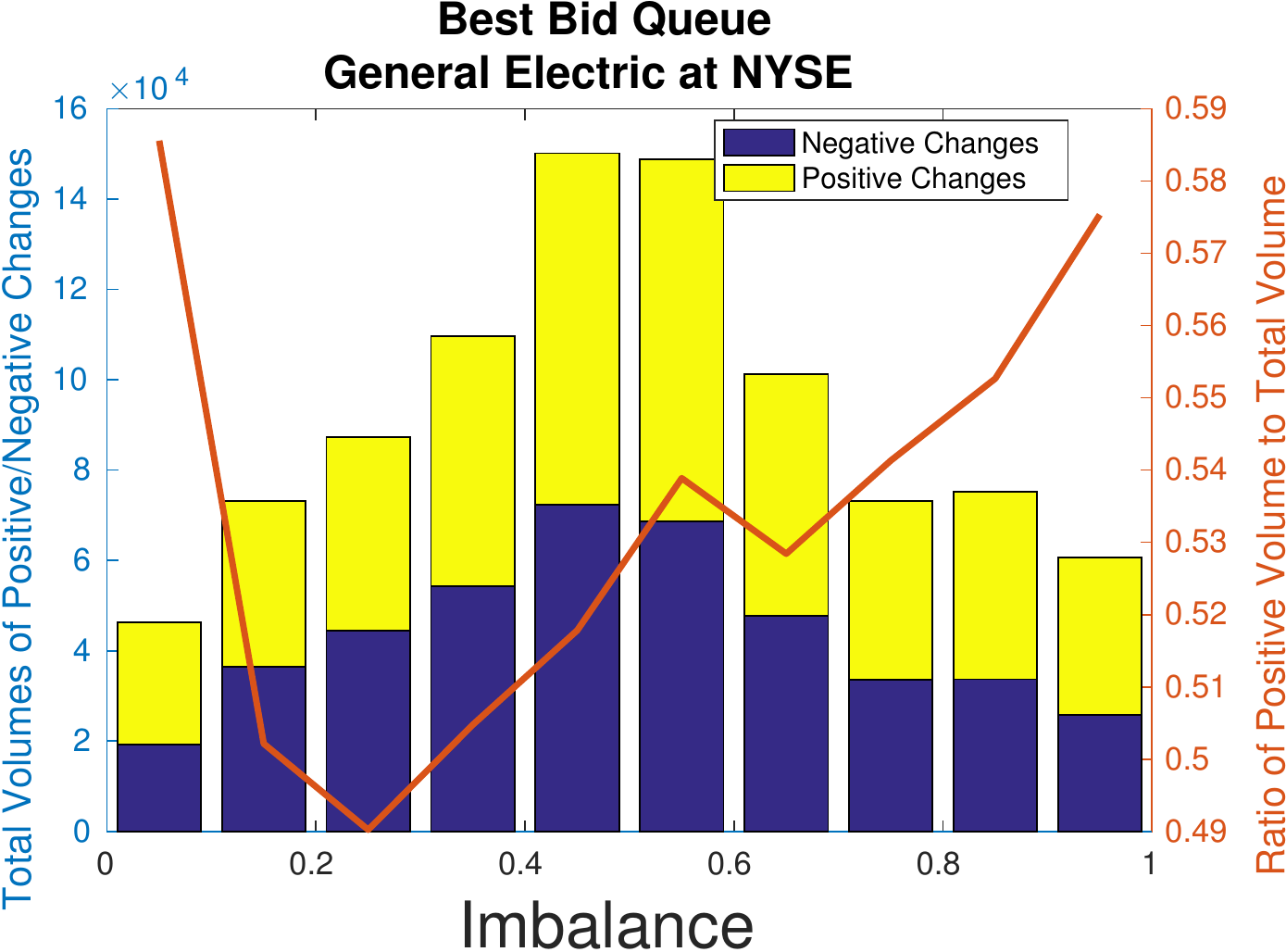}
	\includegraphics[width=0.49\textwidth]{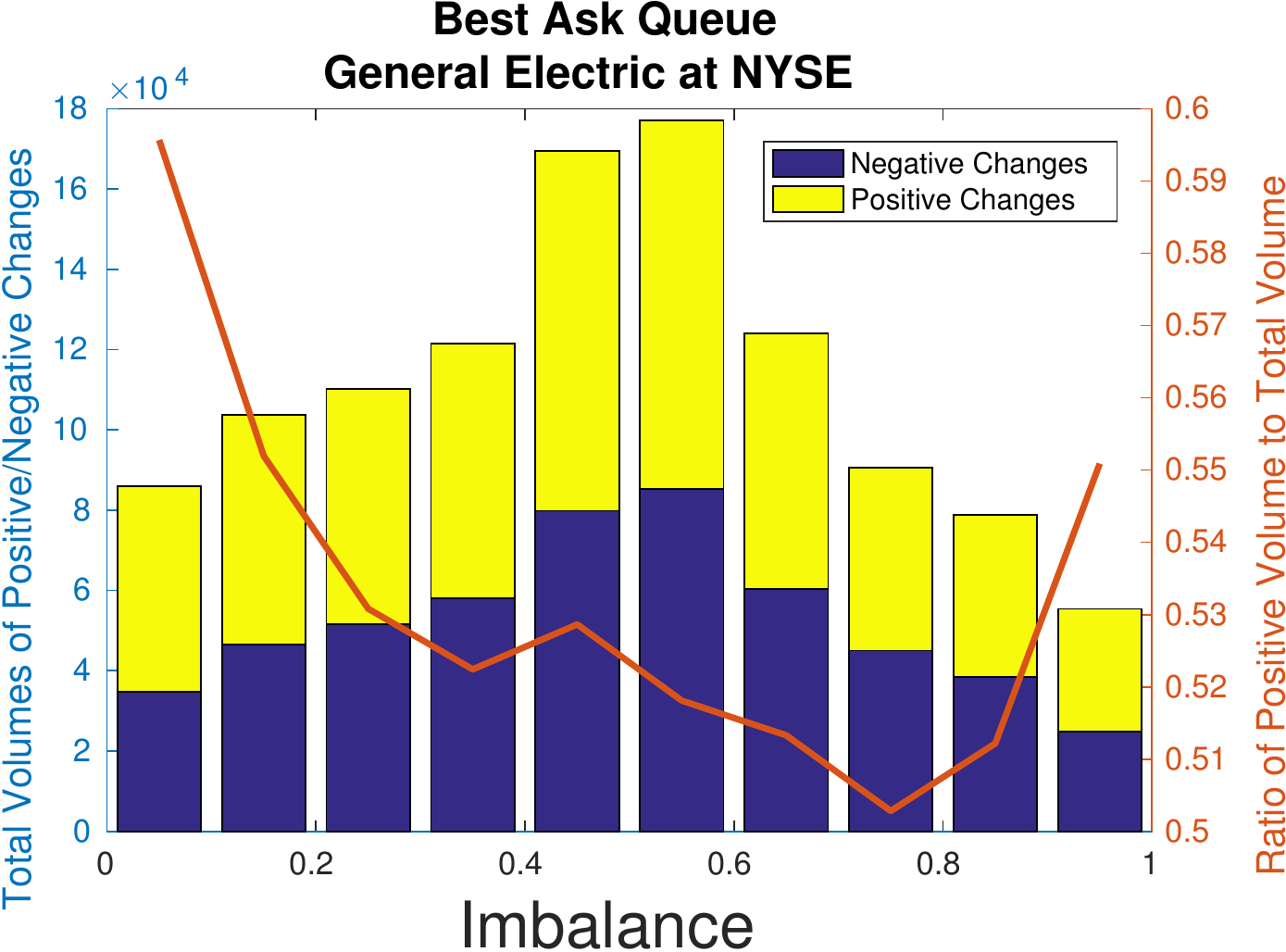}
	
	\caption{Positive and Negative Changes of the Volumes at the Best Bid and the Best Ask of General Electric at NASDAQ and NYSE.
	The curve is the ratio of the positive changes to the total changes.}
	\label{GEBar}
\end{figure}

\begin{figure}
	\centering
	\includegraphics[width=0.49\textwidth]{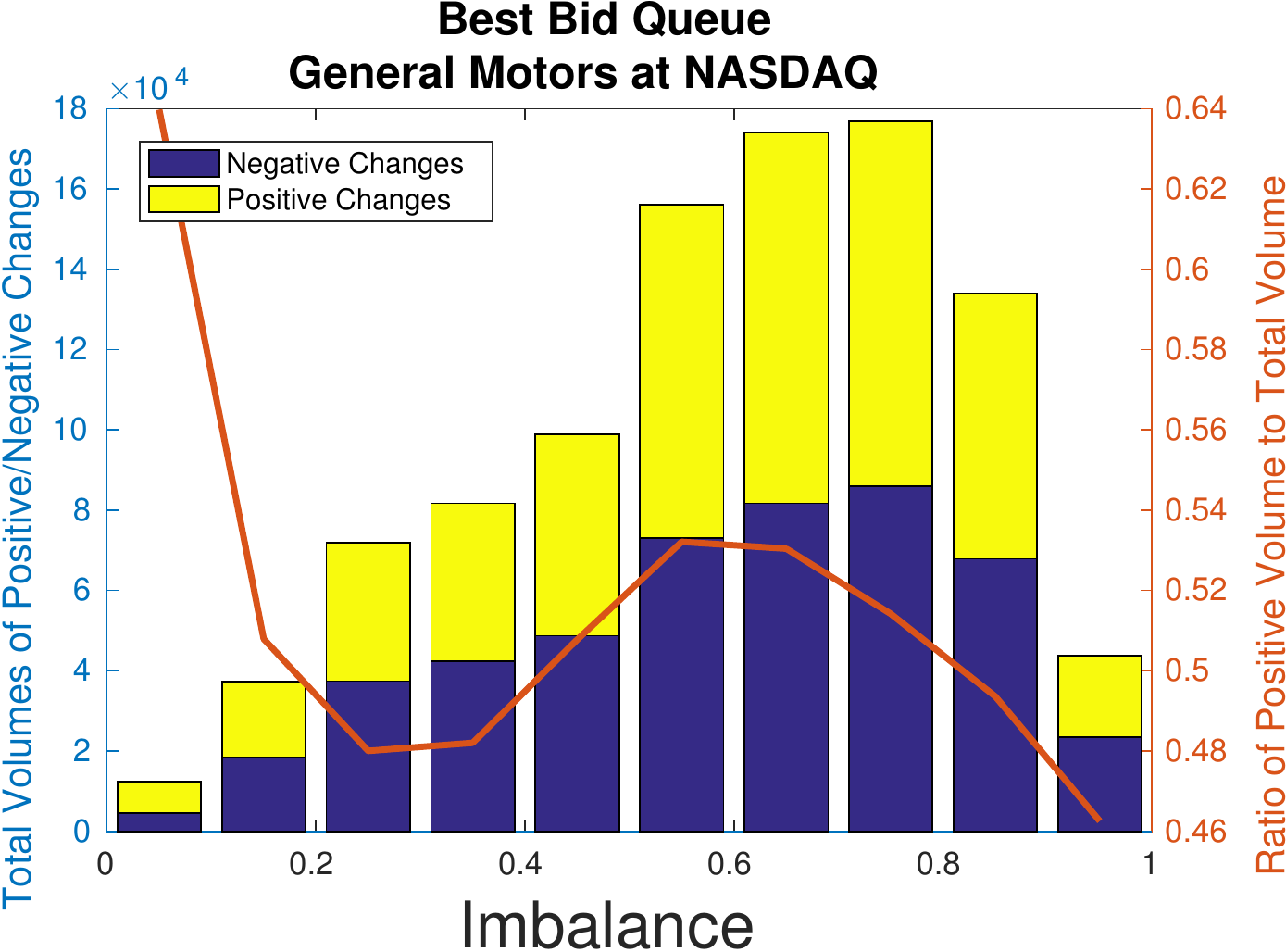}
	\includegraphics[width=0.49\textwidth]{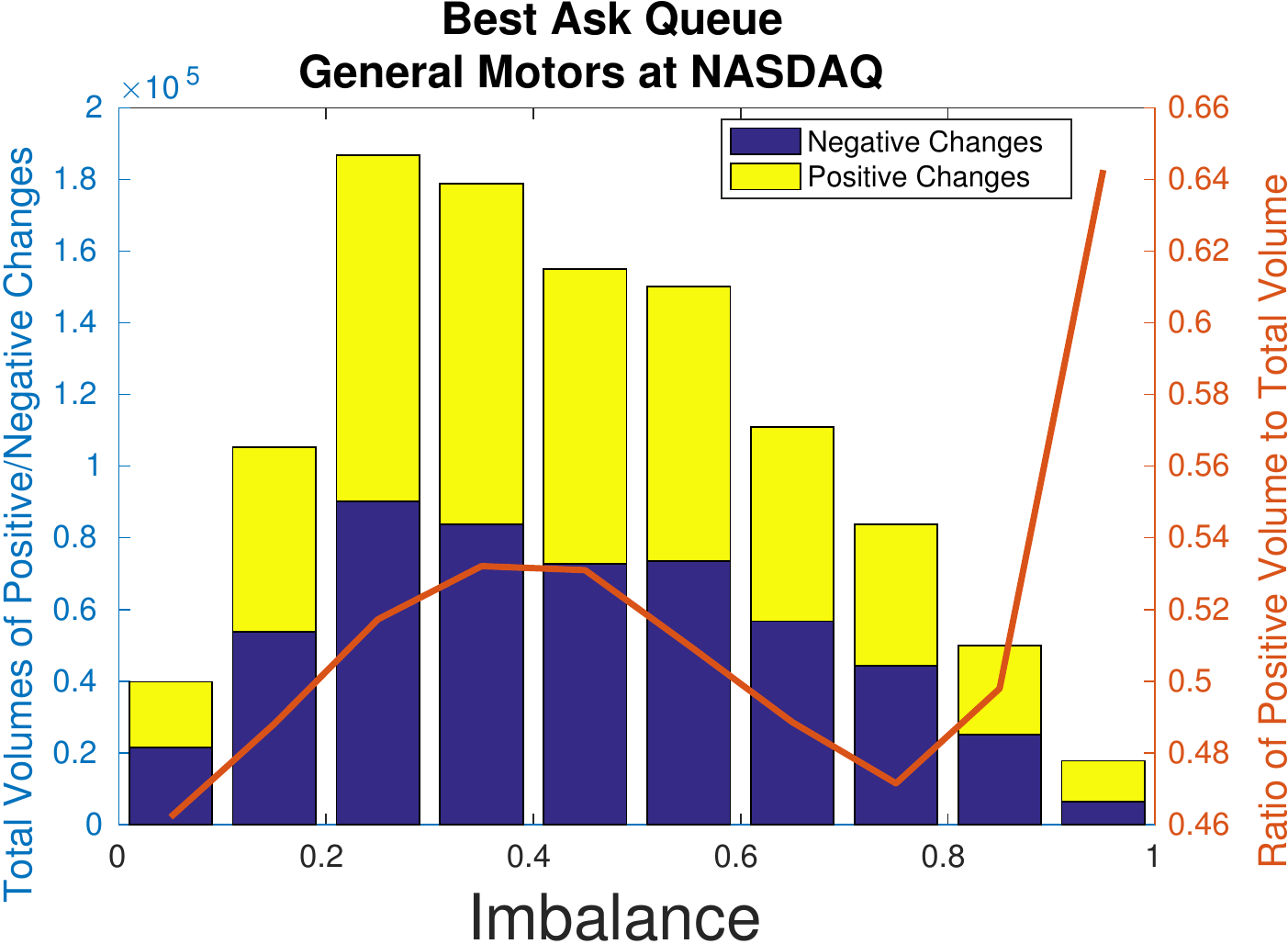}
	\includegraphics[width=0.49\textwidth]{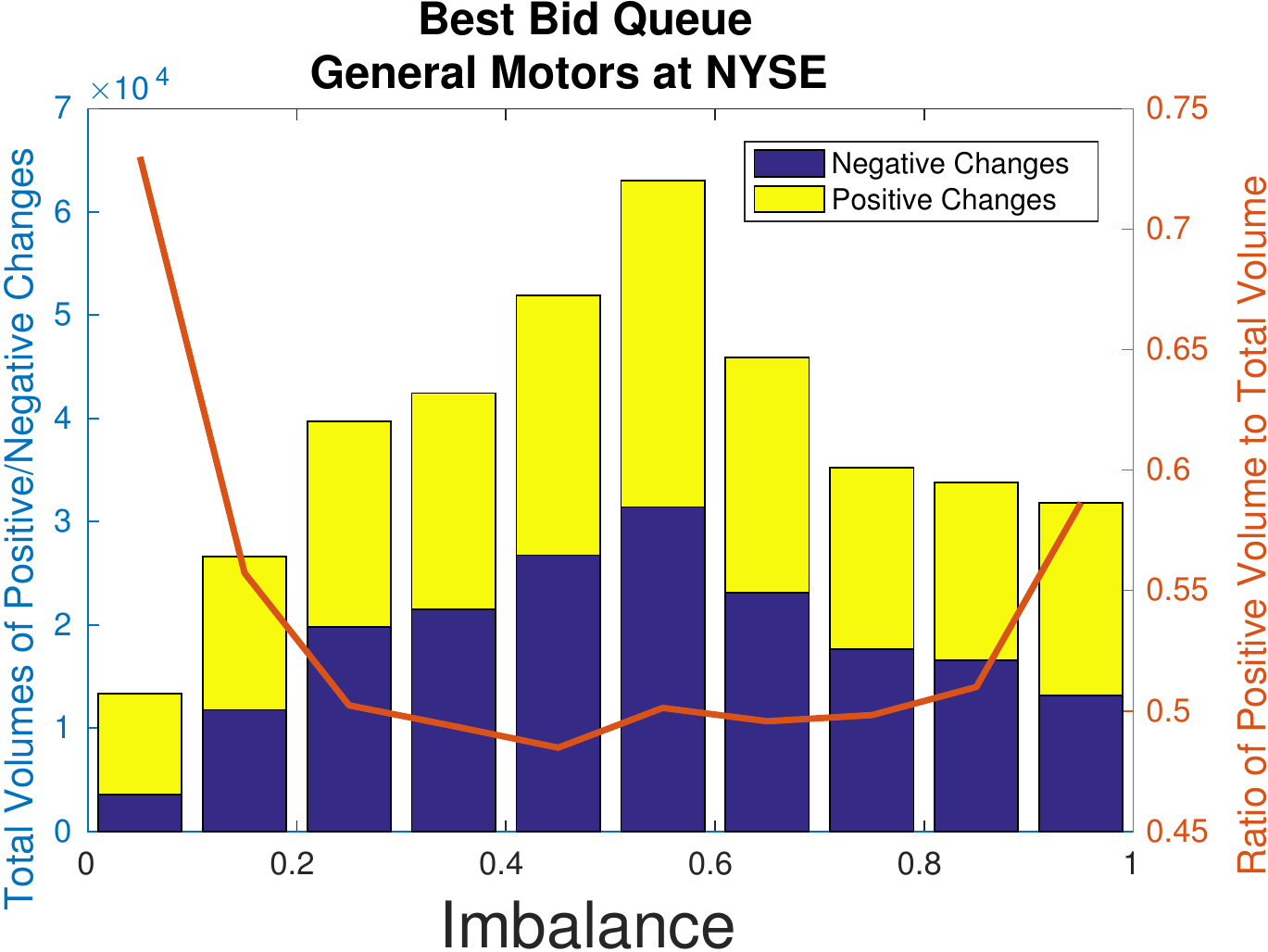}
	\includegraphics[width=0.49\textwidth]{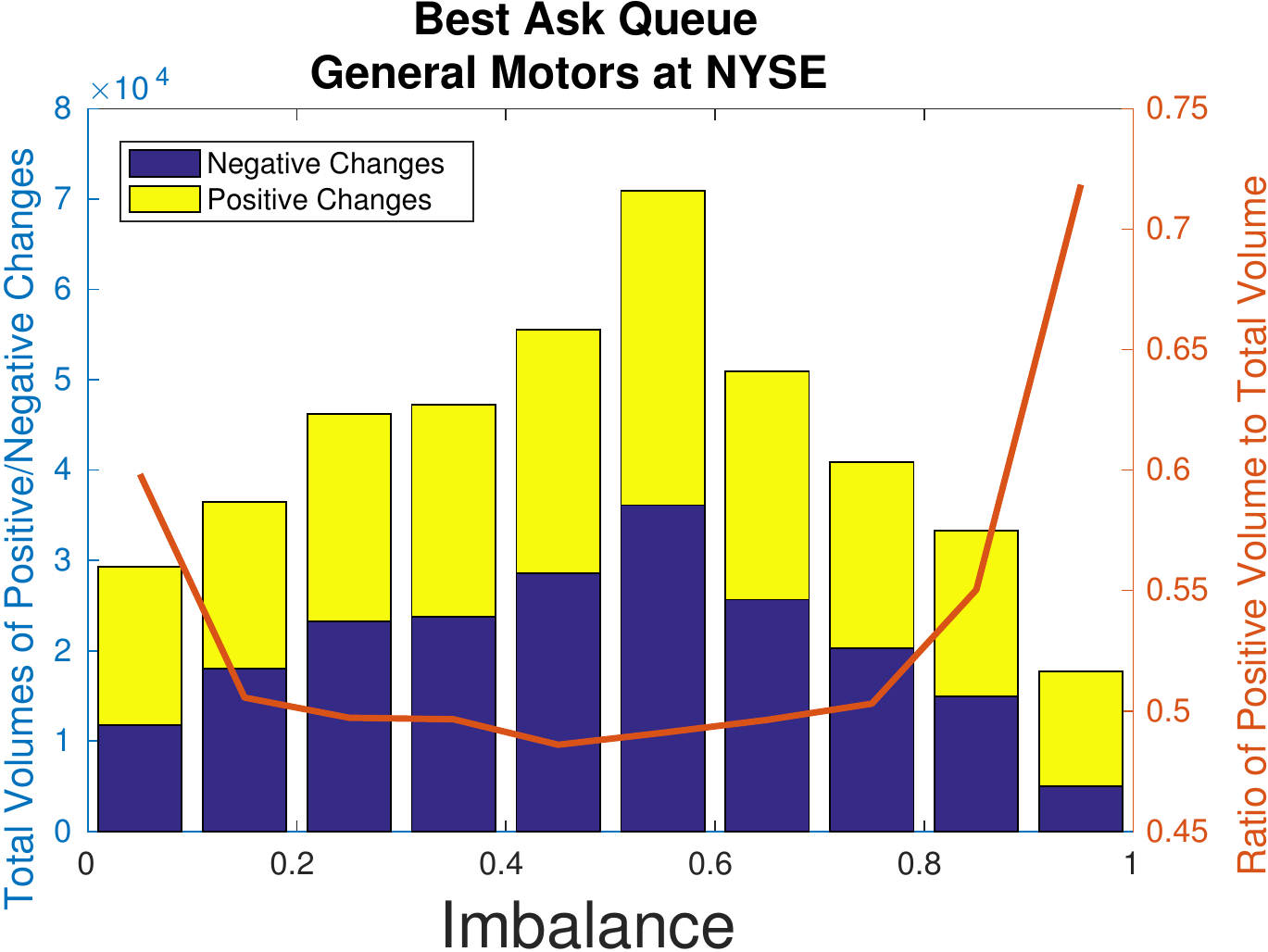}

	\caption{Positive and Negative Changes of the Volumes at the Best Bid and the Best Ask of General Motors at NASDAQ and NYSE.
	The curve is the ratio of the positive changes to the total changes.}
	\label{GMBar}
\end{figure}

\begin{figure}
	\centering
	\includegraphics[width=0.49\textwidth]{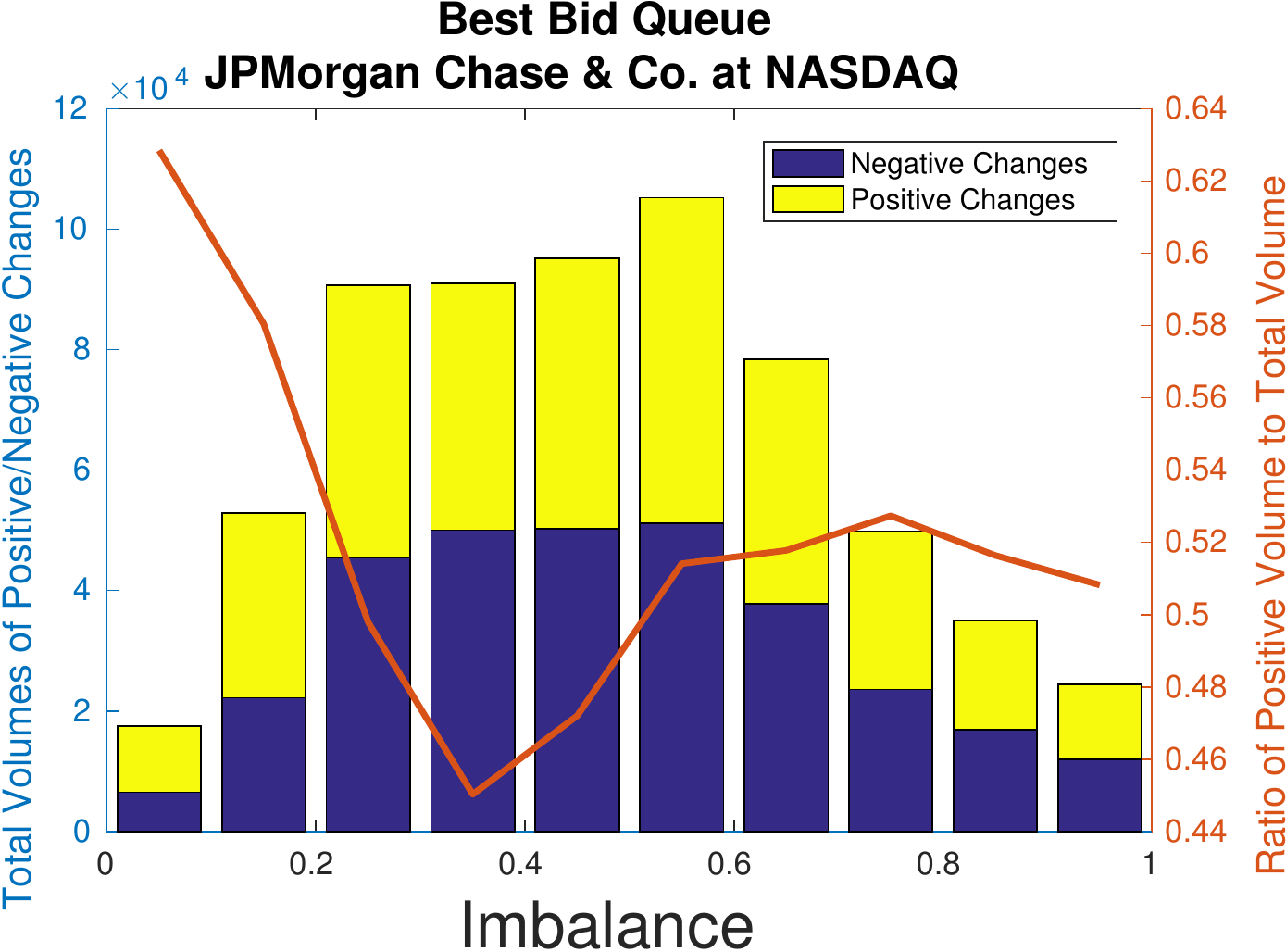}
	\includegraphics[width=0.49\textwidth]{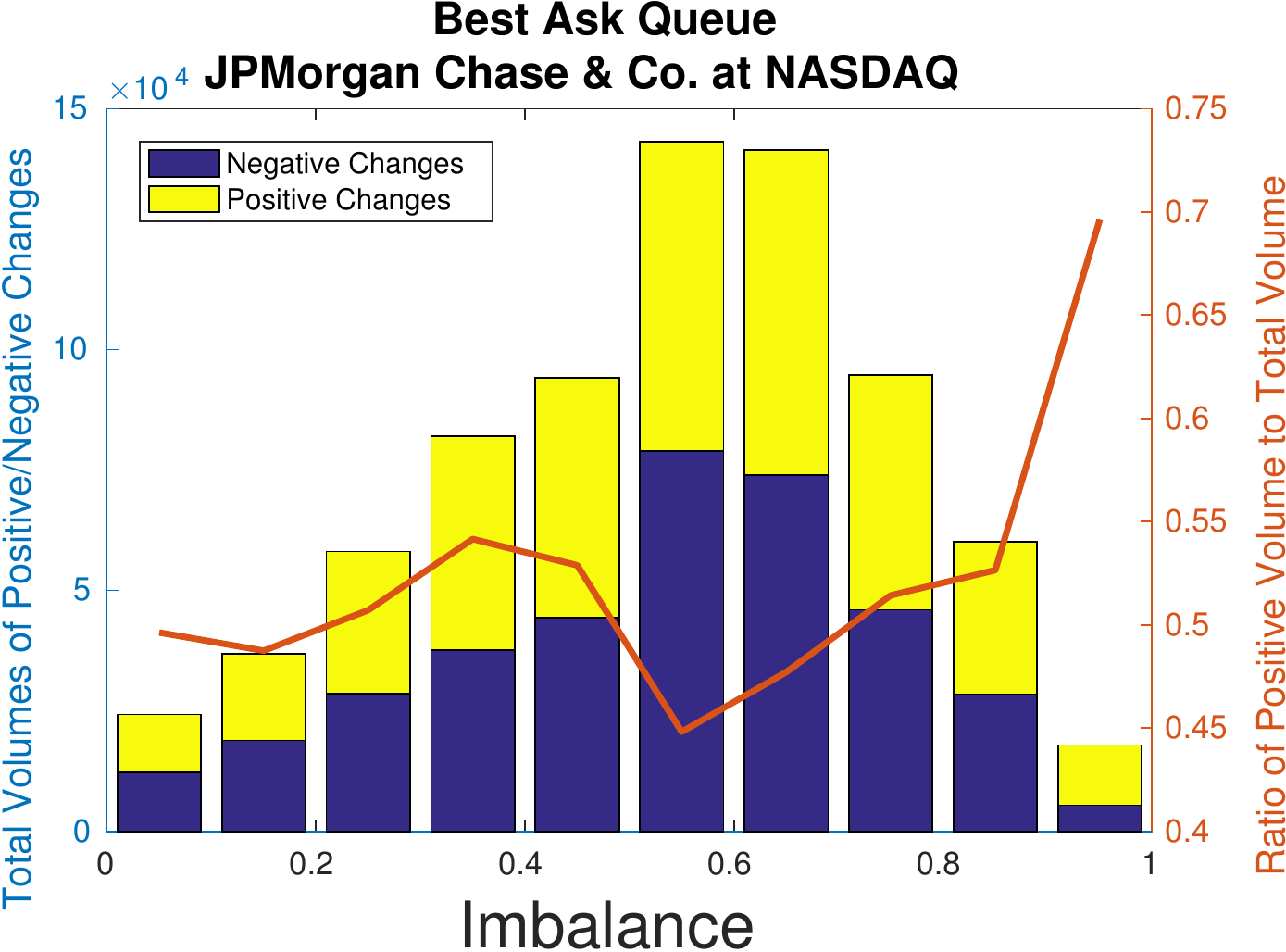}
	
	\includegraphics[width=0.49\textwidth]{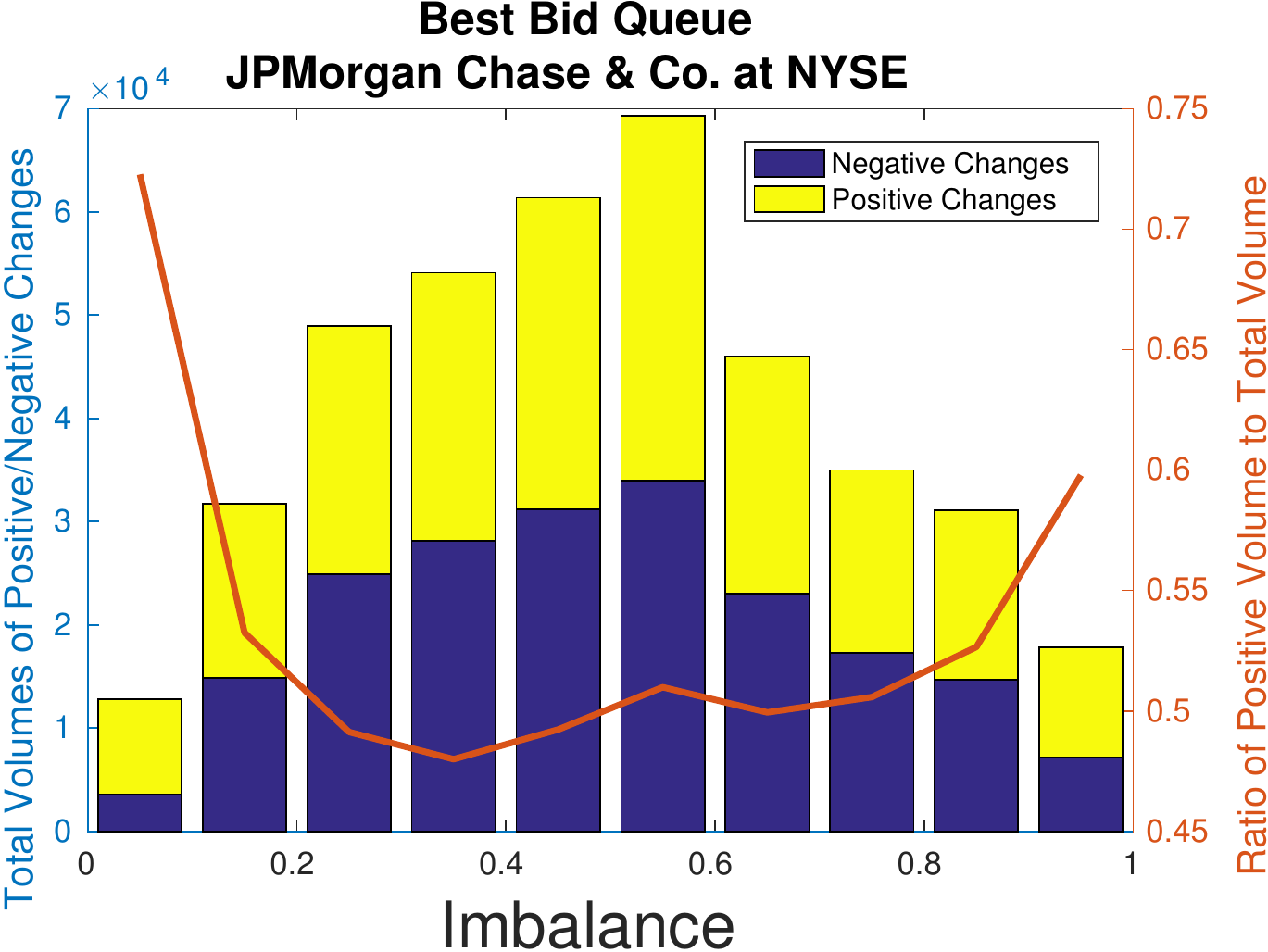}
	\includegraphics[width=0.49\textwidth]{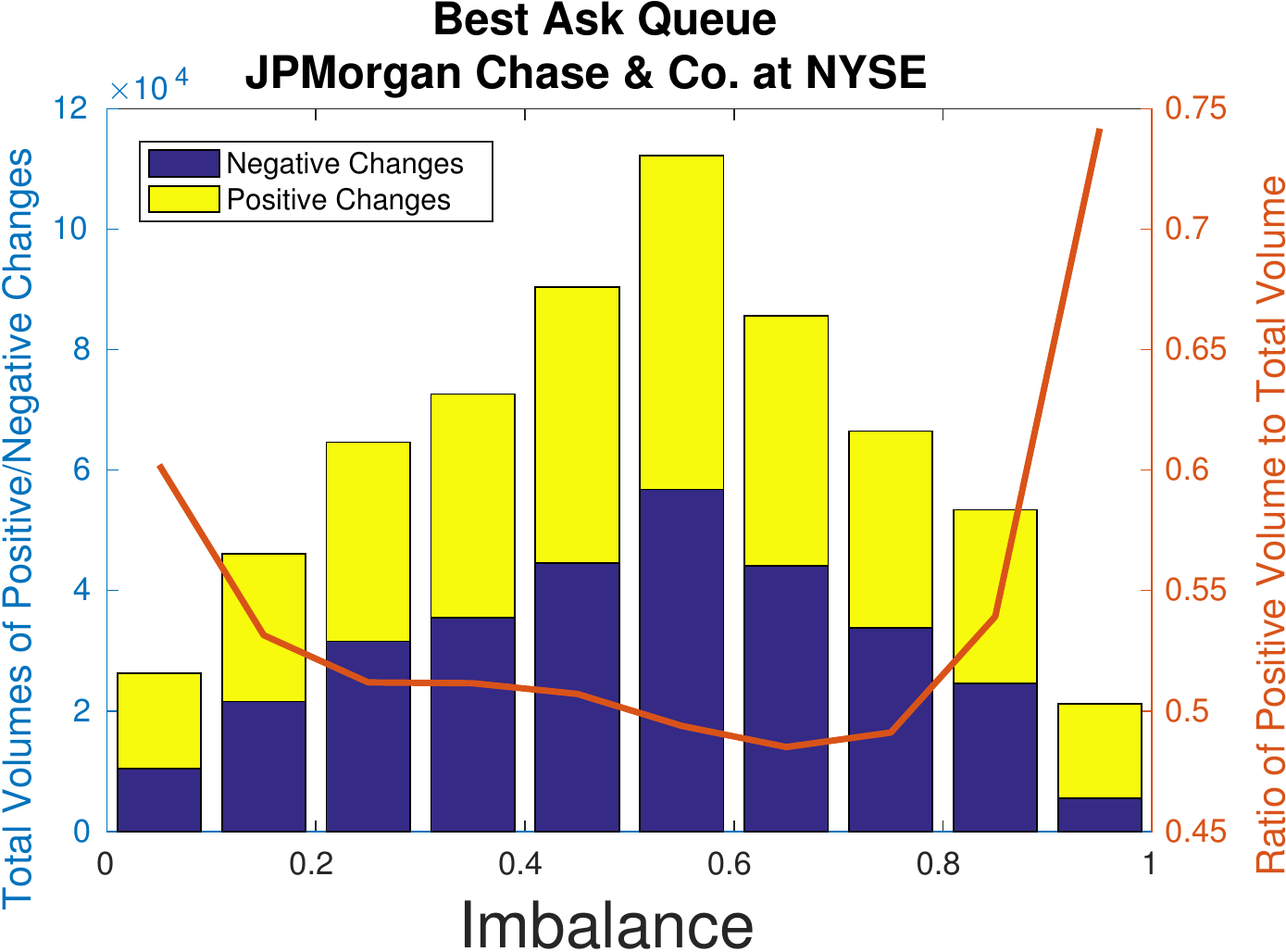}

	\caption{Positive and Negative Changes of the Volumes at the Best Bid and the Best Ask of JP Morgan \& Chase at NASDAQ and NYSE.
	The curve is the ratio of the positive changes to the total changes.}
	\label{JPMBar}
\end{figure}


Next, we investigate the correlation between the best bid and ask dynamics, the volatilities
at the best bid and ask queues, and their dependence on the imbalance. 
We summarized our observations for the Bank of America stock in Figure \ref{BACcorr},
the General Electric stock in Figure \ref{GEcorr}, 
the General Motors stock in Figure \ref{GMcorr},
and the JP Morgan \& Chase stock in Figure \ref{JPMcorr}.
For example, let us take a look at the summary for the Bank of America stock
in Figure \ref{BACcorr}. The top row in Figure \ref{BACcorr} stands
for the correlation between the size changes at the best bid and the best ask (top left),
the standard deviation of size changes at the best bid (top middle), and 
the standard deviation of size changes at the best ask (top right)
for the Bank of America stock traded in NASDAQ. 
Similar statistics for the Bank of America stock traded in NYSE are summarized
in the bottom row in Figure \ref{BACcorr}. As we can see from the top left picture,
the correlation as a function of the imbalance, is a $W$-shaped curve for the Bank of America stock traded in NASDAQ
and from the bottom left picture a $U$-shaped curve for the Bank of America stock traded in NYSE. 
Similar pattern is observed also for the General Electric stock traded in NASDAQ and NYSE, see Figure \ref{GEcorr}.
The $U$-shaped curve is observed for General Motors and JP Morgan \& Chase traded in both NASDAQ and NYSE,
see Figure \ref{GMcorr} and Figure \ref{JPMcorr}. Indeed, we studied some other stocks as well
in the WRDS database and empirical studies suggest that $U$-shape curves and $W$-shaped curves
are universal for the correlation between the size changes at the best bid and ask for most stocks.
It also holds that the correlation in general is negative but is far away from $-1$. It is curious why a typical relation
of the correlation between the size changes at the best bid and ask and the imbalance of the best bid and ask
can be represented by either a $U$-shaped curve or a $W$-shaped curve. It is also worth noting that
sometimes we get different shaped curves for different exchanges (Figure \ref{BACcorr}, Figure \ref{GEcorr})
and sometimes we get the same shaped curves for different exchanges (Figure \ref{GMcorr}, Figure \ref{JPMcorr}).
That can probably be explained by the fact that some high frequency and algorithmic trading firms
apply their trading strategies to a particular stock exchange only and the different trading strategies
result in the different patterns of the best bid and ask dynamics we observed from the data.
Figures \ref{BACcorr}, \ref{GEcorr}, \ref{GMcorr}, \ref{JPMcorr} also contain the information
about the standard deviations of the size changes at the best bid and best ask queues on NASDAQ and NYSE.
The general observation is that most of the time, the standard deviation increases as the imbalance increases
at the best bid queues and decreases as the imbalance increases at the best ask queues. 
Note that best bid size increases as imbalance increases and best ask size decreases as imbalance increases.
Hence, what we observed is that the standard deviations increases as the queue lengths increases.
This is not surprising at all. But what's interesting is that in many cases, it is not exactly monotone
and we see a sudden increase of the standard deviation when the imbalance is small for the best bid
and large for the best ask, that is, when the queue length is short. That suggests
that when the queue length is short, that is when the queue is about to get deleted, or when there is a new queue created,
the volatilities tend to be large. In general, the volatilities of the empirical data 
tend to be noisier than the correlations, 
which is either a $U$-shaped or a $W$-shaped curve. Nevertheless,
it is quite often to observe the skewed $U$-shaped curves. For example,
in top middle and top right pictures in Figure \ref{BACcorr}, Figure \ref{GEcorr}, Figure \ref{GMcorr}
and Figure \ref{JPMcorr}, we have the skewed $U$-shaped curves. 
For the best bid queues, it is skewed towards the left and for the best ask queues, it is skewed towards the right.
It is curious that for the stocks traded on NASDAQ, we have this universal skewed $U$-shapes for the volatilities.
But the data for the NYSE tend to be noisier and the pattern is not very clear. This once again
indicates the very different natures of the level-1 limit order dynamics across different exchanges.

We summarize the statistics of the correlations in Table \ref{corrTable}. As we can see, the correlation is almost
always negative. In terms of the numbers, the strongest correlation is $0.02$ achieved by
the Bank of America stock traded on NYSE with imbalance between $0.05$ and $0.10$.
The most negative correlation is achieved by JP Morgan traded on NYSE, which is $-0.34$, that is far
away from $-1$.
One interesting observation is that when the imbalance is between $0.2$ and $0.8$, 
from Table \ref{corrTable}, we can see that the correlation of the stock traded on NYSE is always
more negative than the correlation of the same stock traded on NASDAQ\footnote{with the exception of JP Morgan when the imbalance is between 0.55 and 0.60}. 
As we mentioned earlier, 
the fragmentation and discrepancy of the stock exchanges is well documented in the literature. 
For example, we can ask the question why the correlation of stocks traded on NYSE is more negative
than that of NASDAQ. 

\begin{figure}
	\centering
	\includegraphics[width=0.32\textwidth]{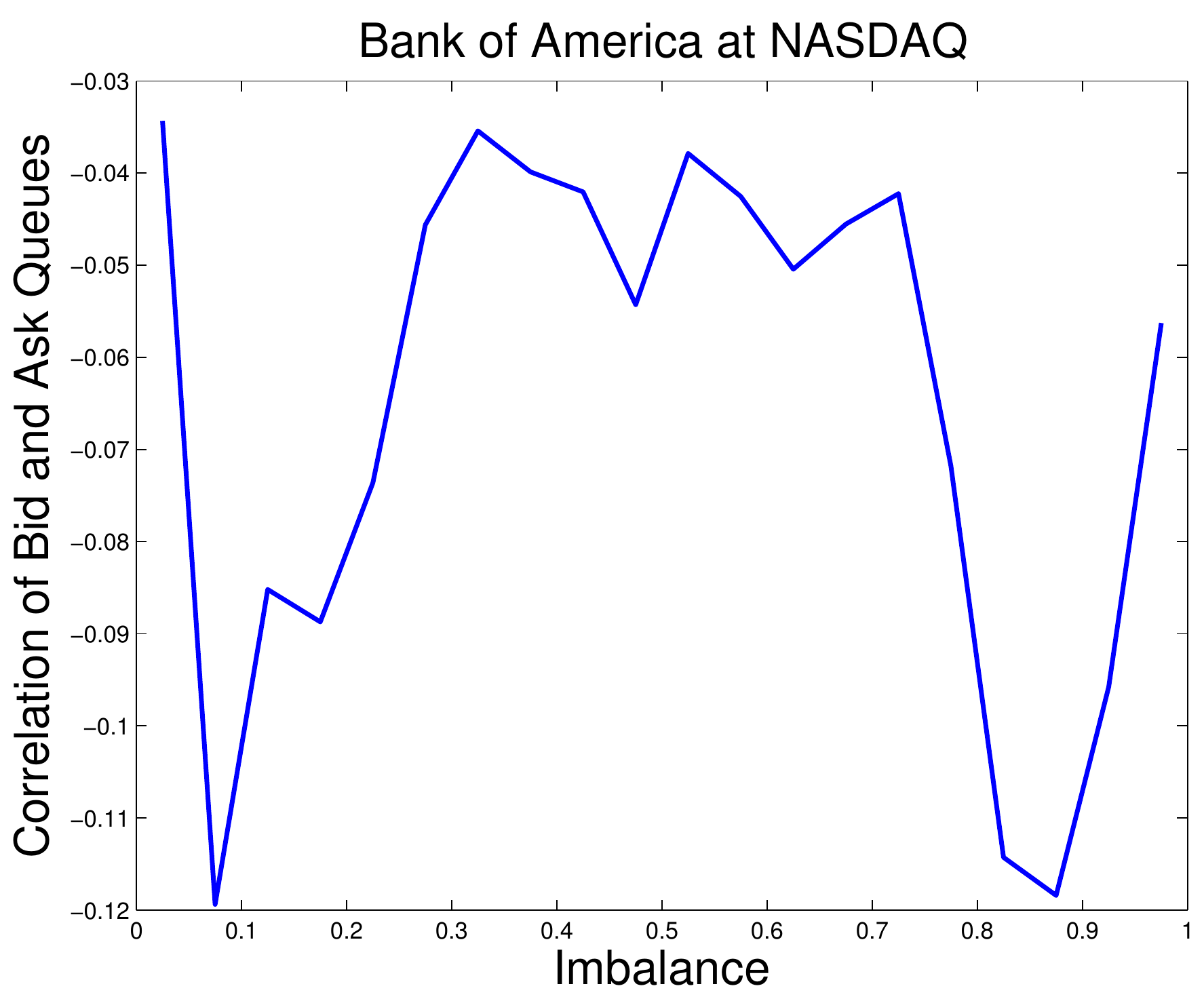}
	\includegraphics[width=0.32\textwidth]{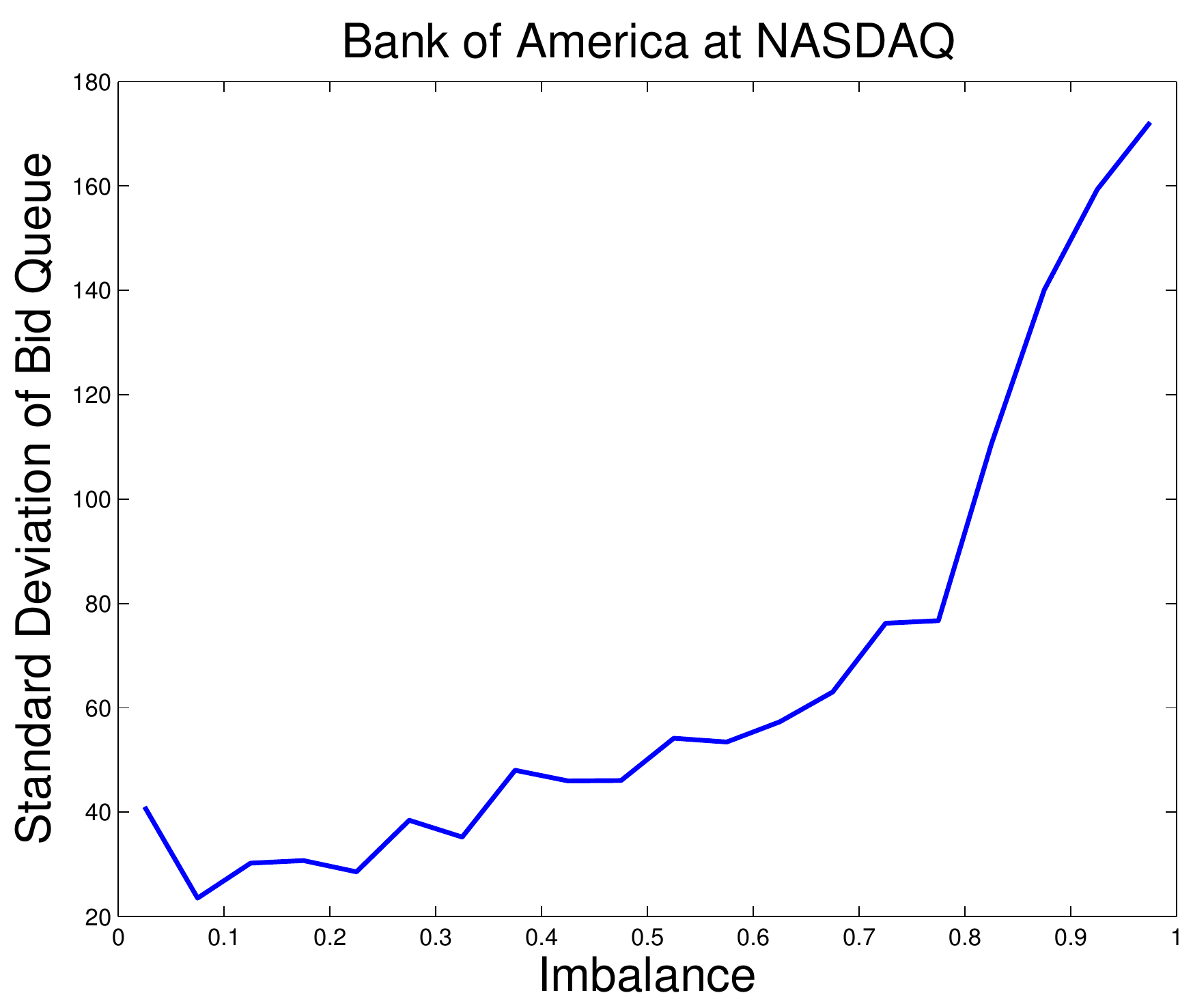}
	\includegraphics[width=0.32\textwidth]{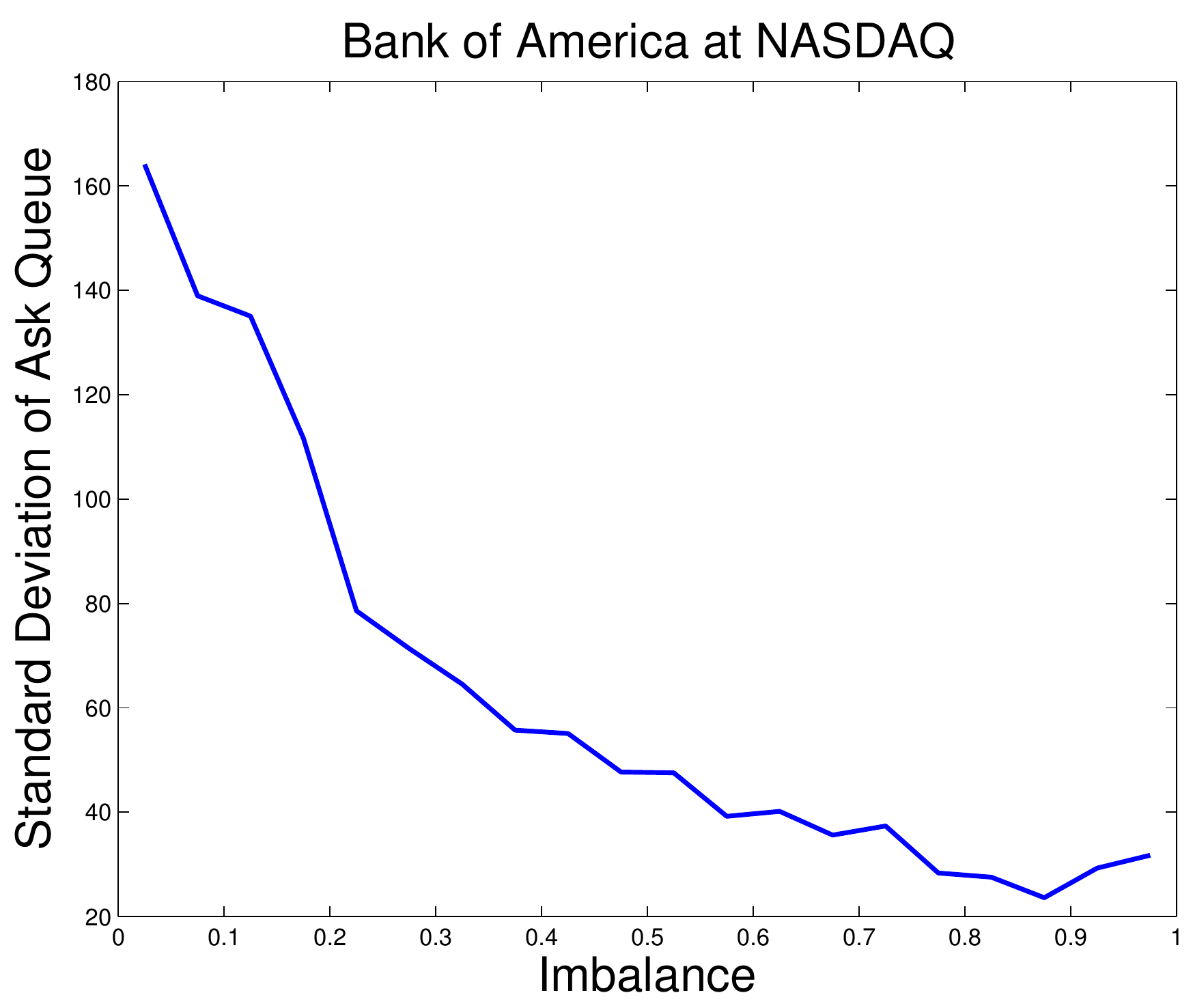}
	
	\includegraphics[width=0.32\textwidth]{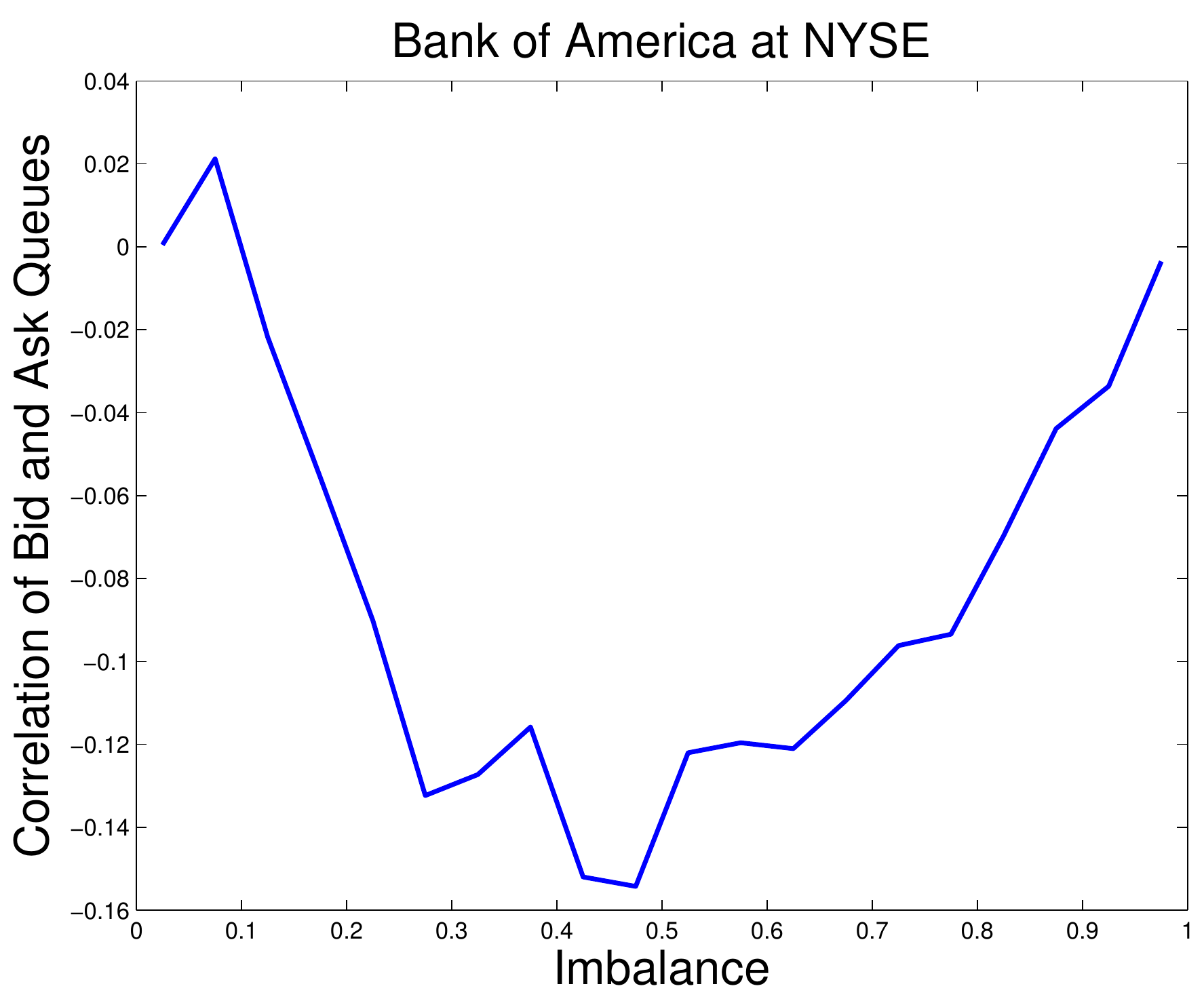}
	\includegraphics[width=0.32\textwidth]{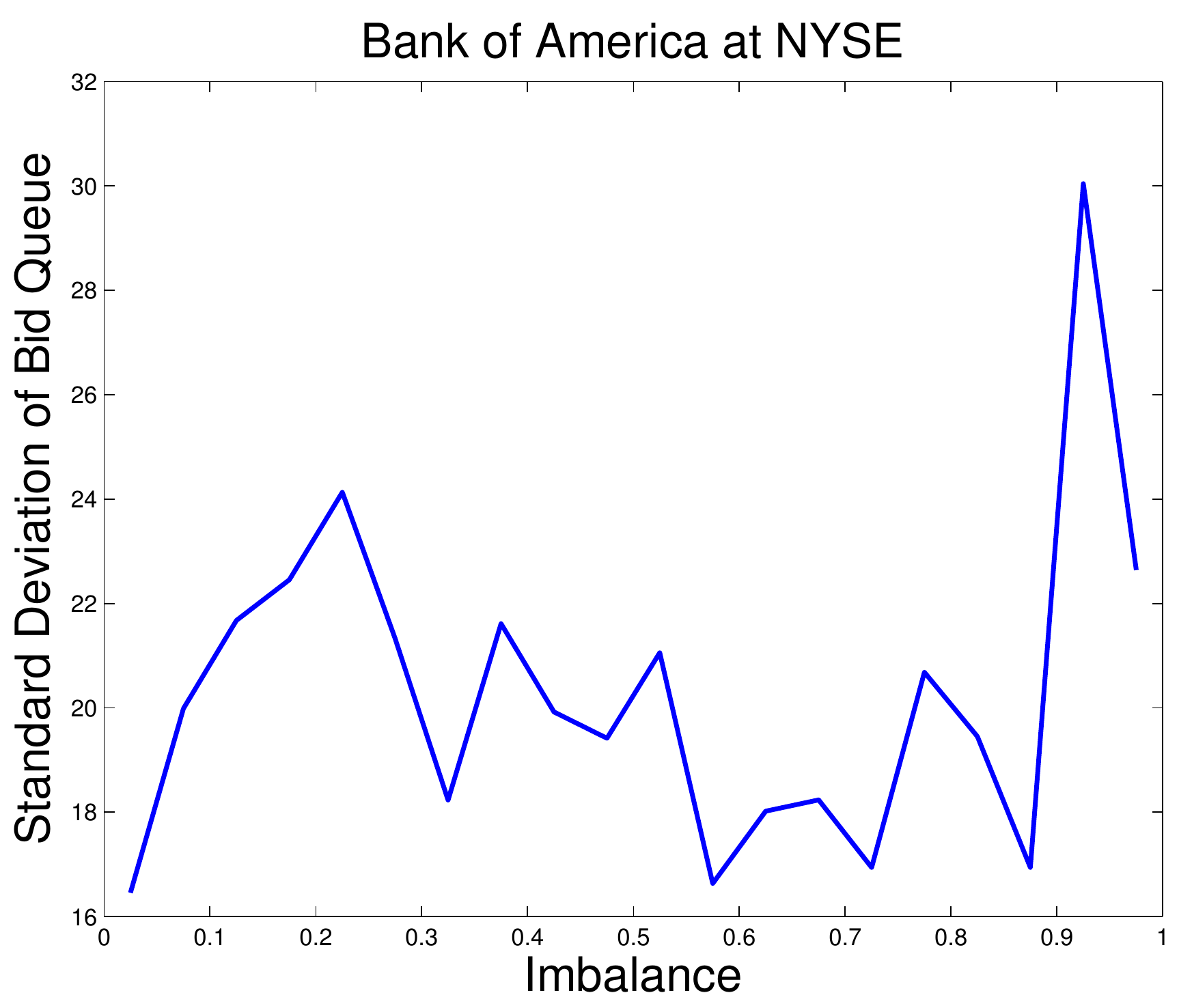}
	\includegraphics[width=0.32\textwidth]{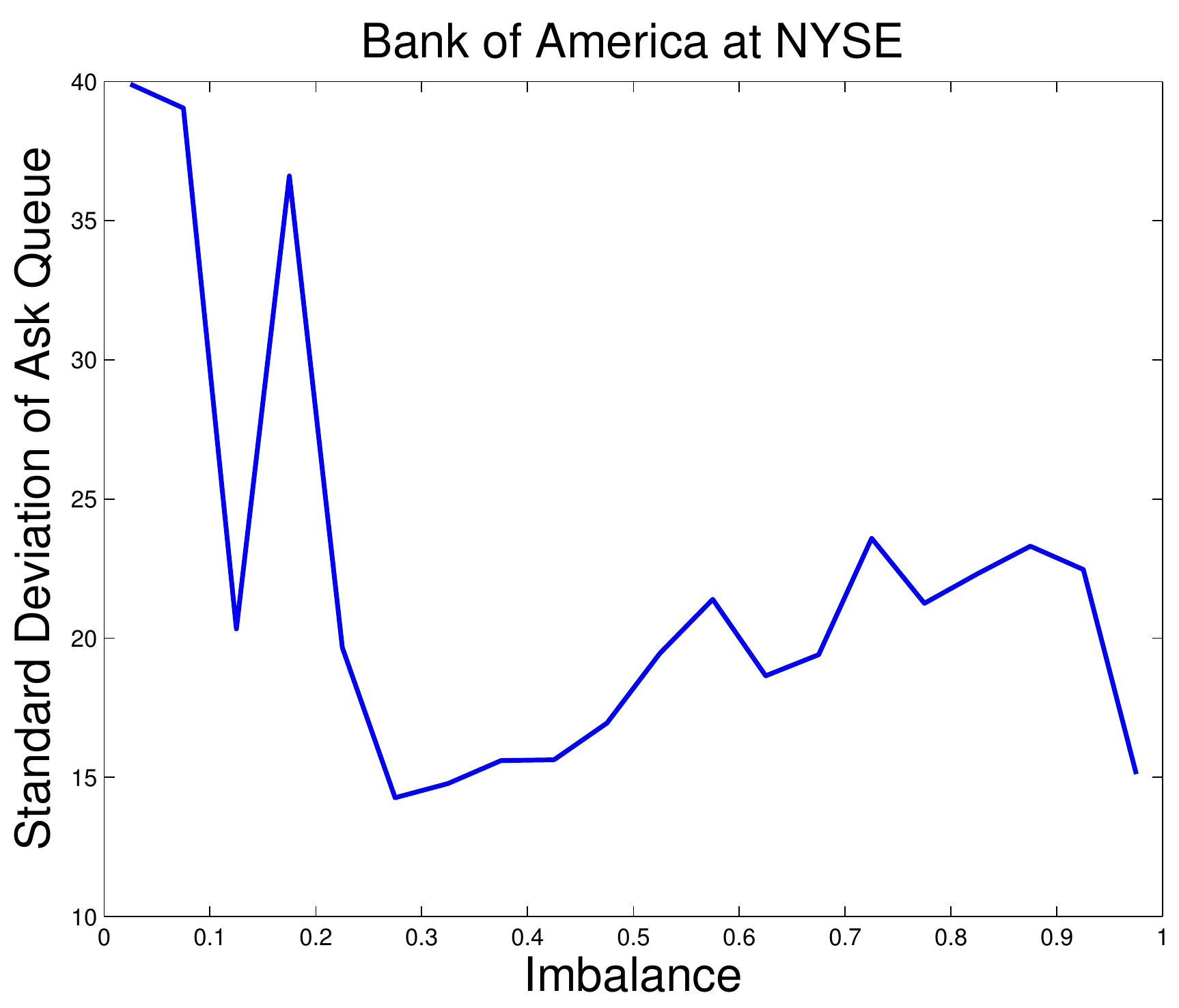}

	\caption{Correlations and Standard Deviations of the Volumes at the Best Bid and the Best Ask
		of Bank of America at NASDAQ and NYSE}
	\label{BACcorr}
\end{figure}

\begin{figure}
	\centering
	\includegraphics[width=0.32\textwidth]{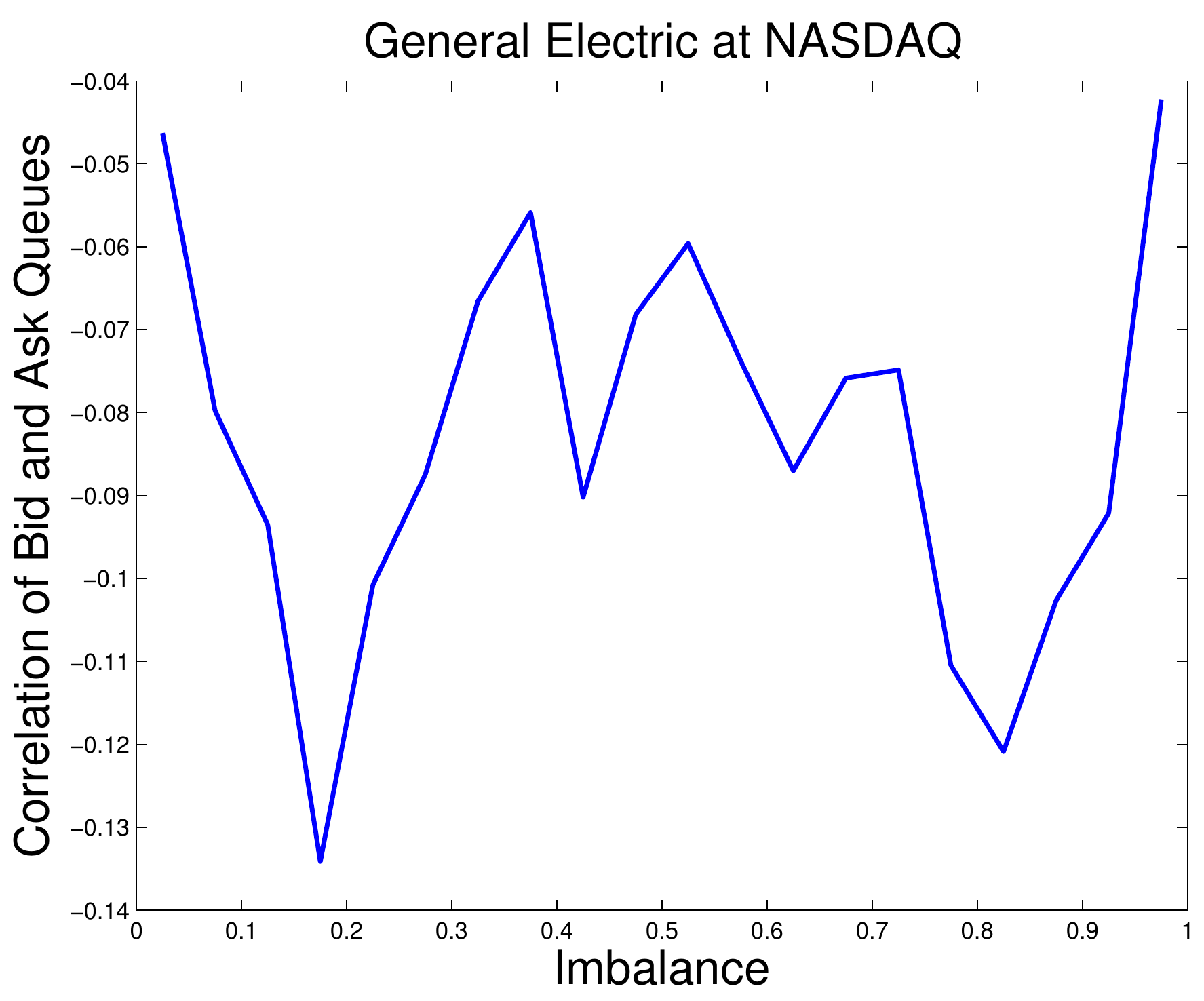}
	\includegraphics[width=0.32\textwidth]{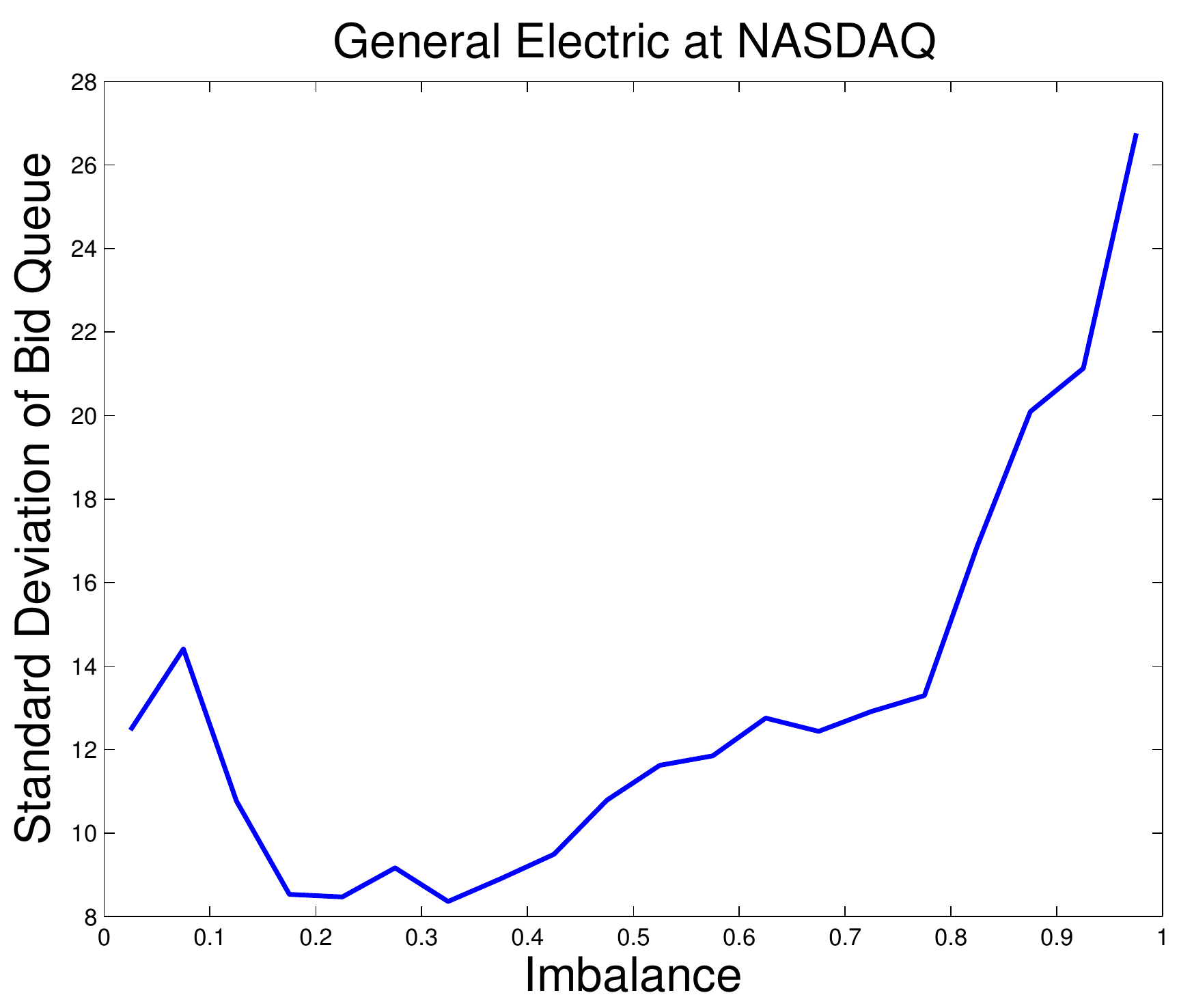}
	\includegraphics[width=0.32\textwidth]{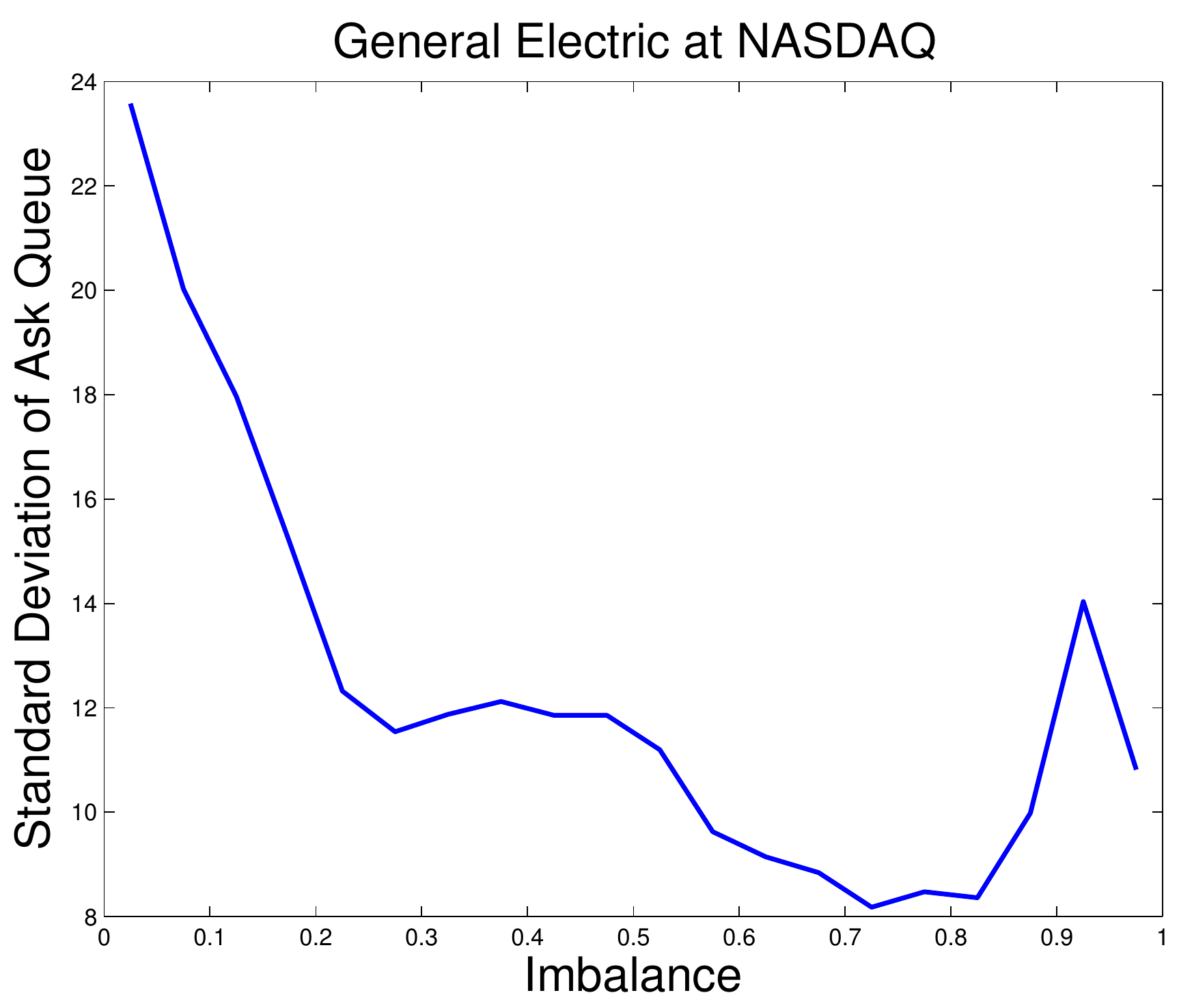}
	
	\includegraphics[width=0.32\textwidth]{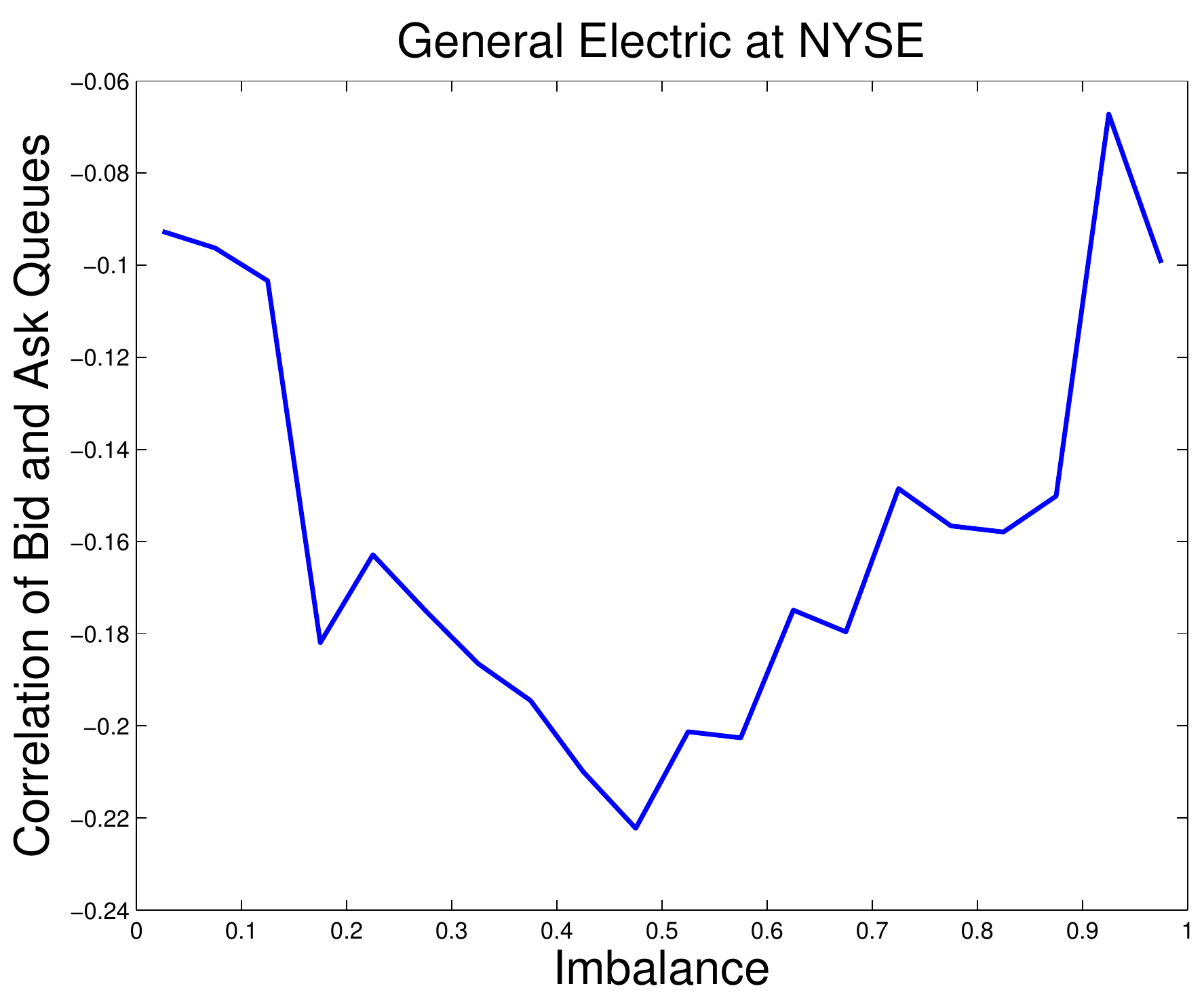}
	\includegraphics[width=0.32\textwidth]{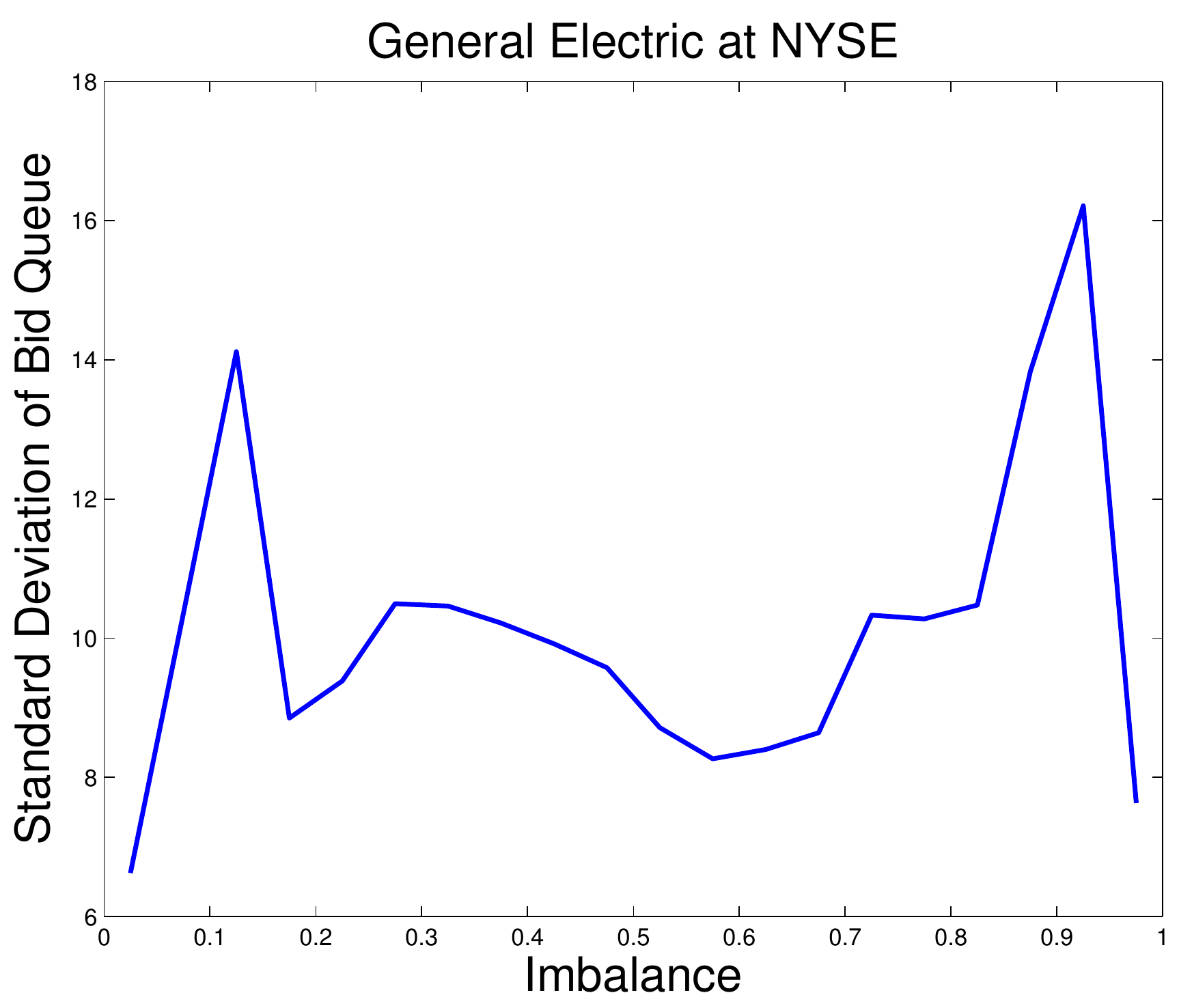}
	\includegraphics[width=0.32\textwidth]{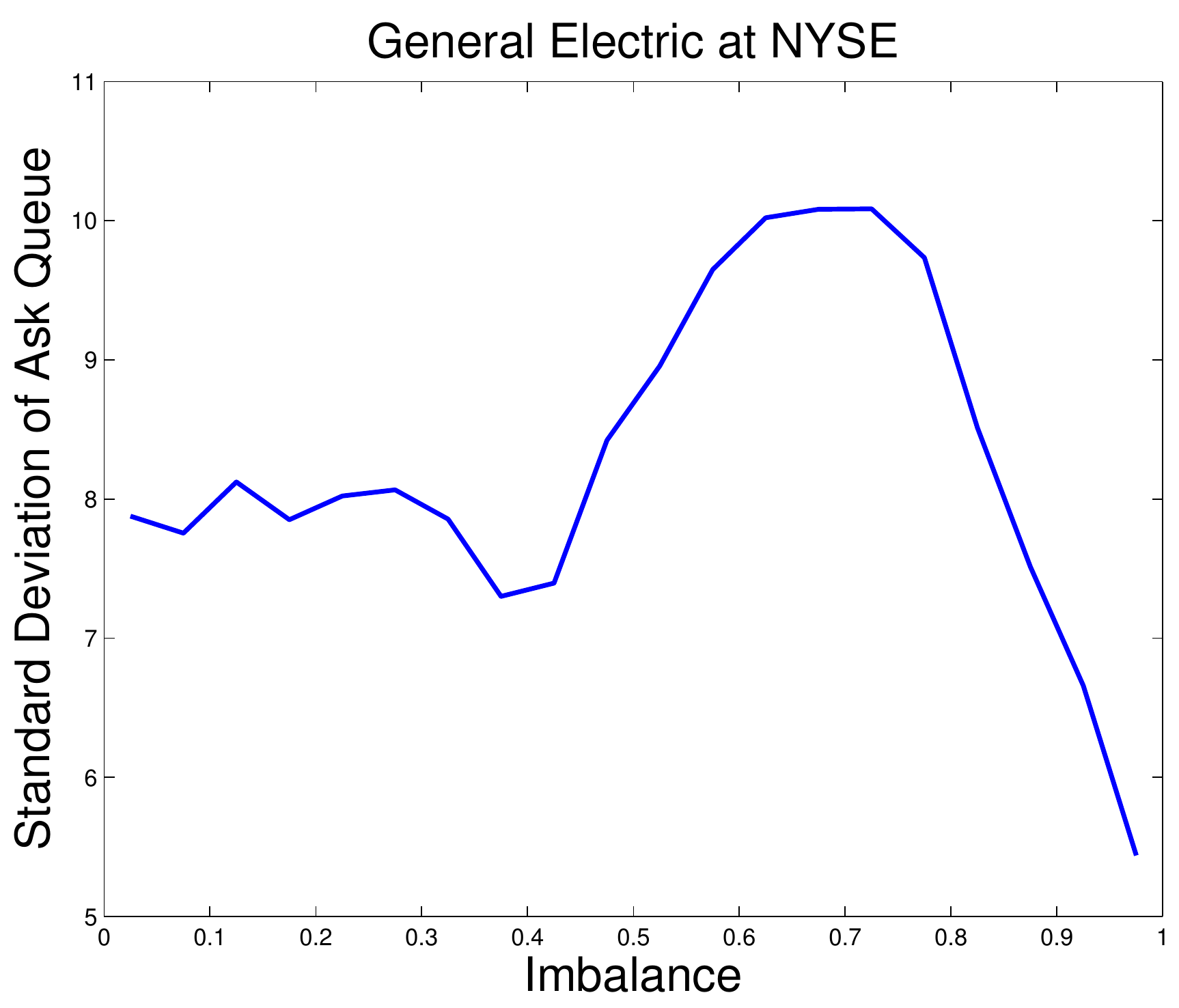}

	\caption{Correlations and Standard Deviations of the Volumes at the Best Bid and the Best Ask
		of General Electric at NASDAQ and NYSE}
	\label{GEcorr}
\end{figure}

\begin{figure}
	\centering
	\includegraphics[width=0.32\textwidth]{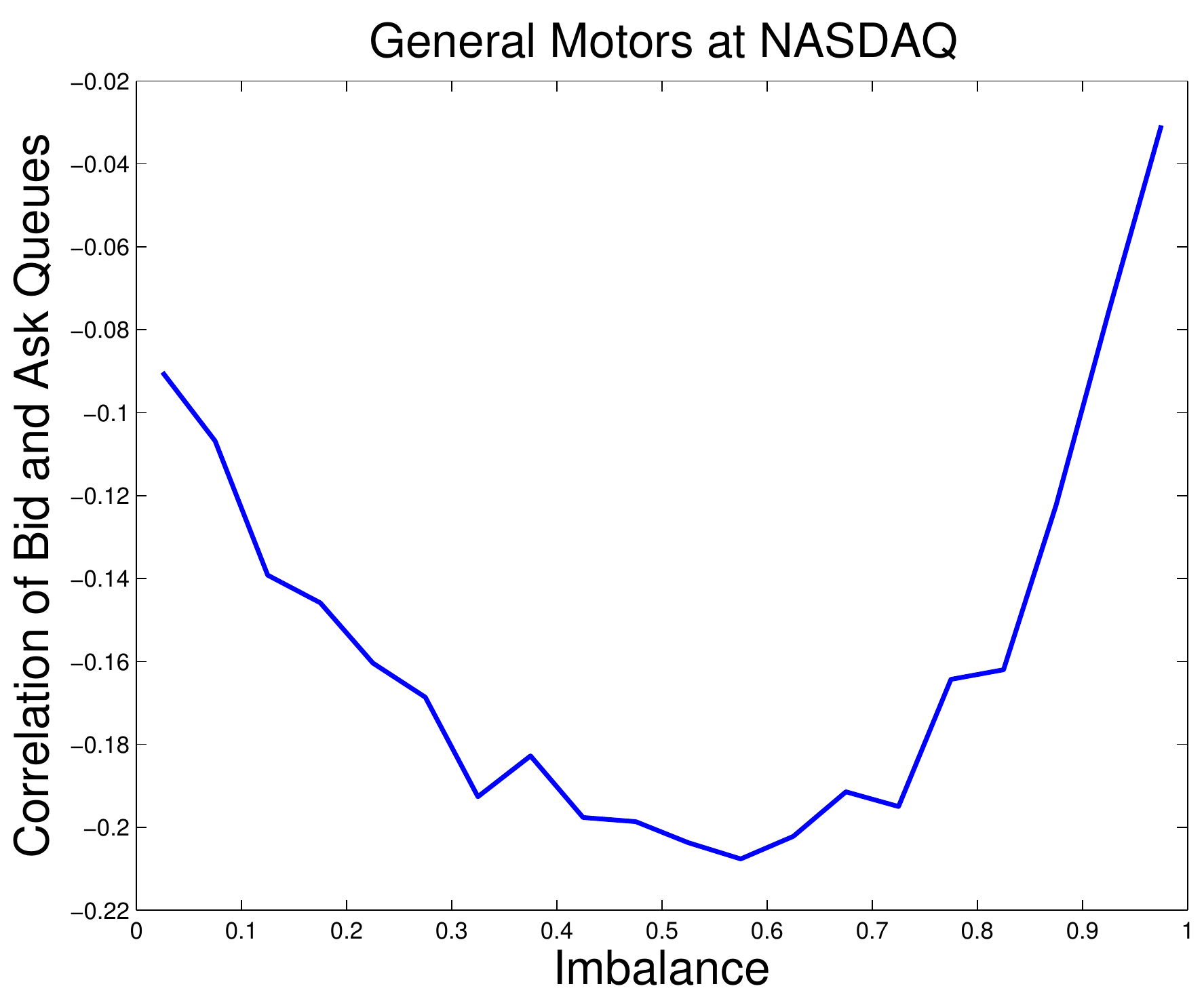}
	\includegraphics[width=0.32\textwidth]{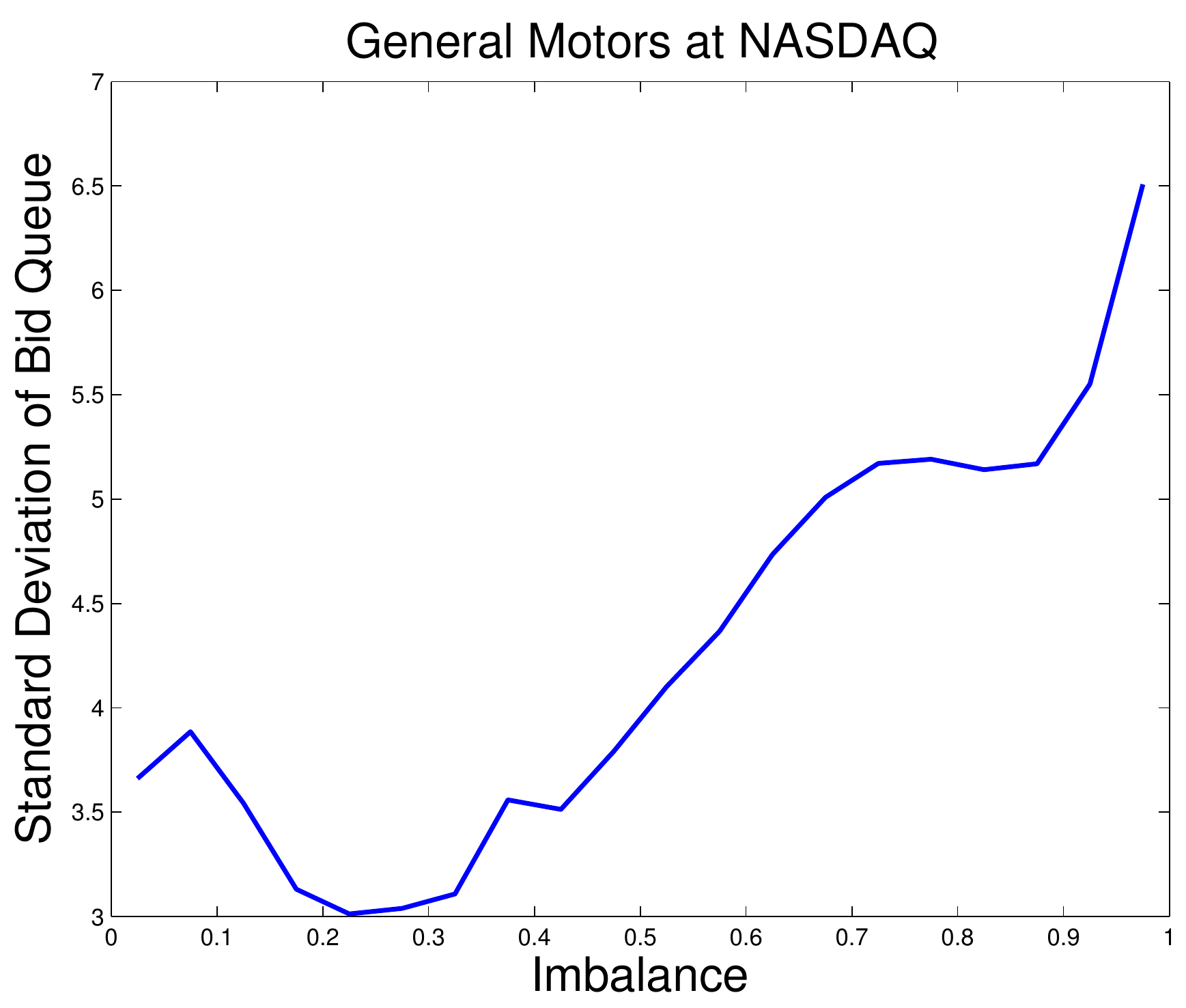}
	\includegraphics[width=0.32\textwidth]{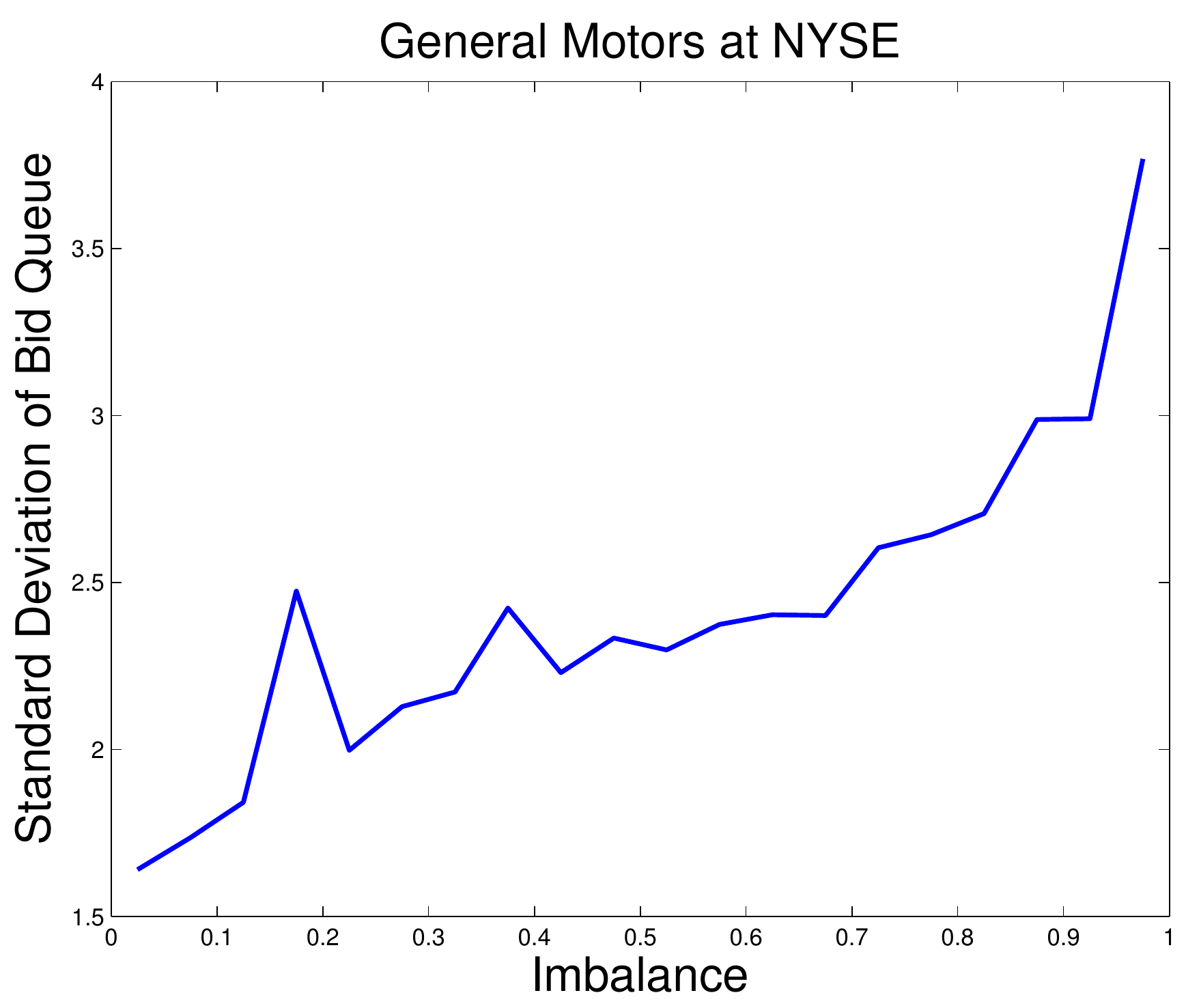}
	
	\includegraphics[width=0.32\textwidth]{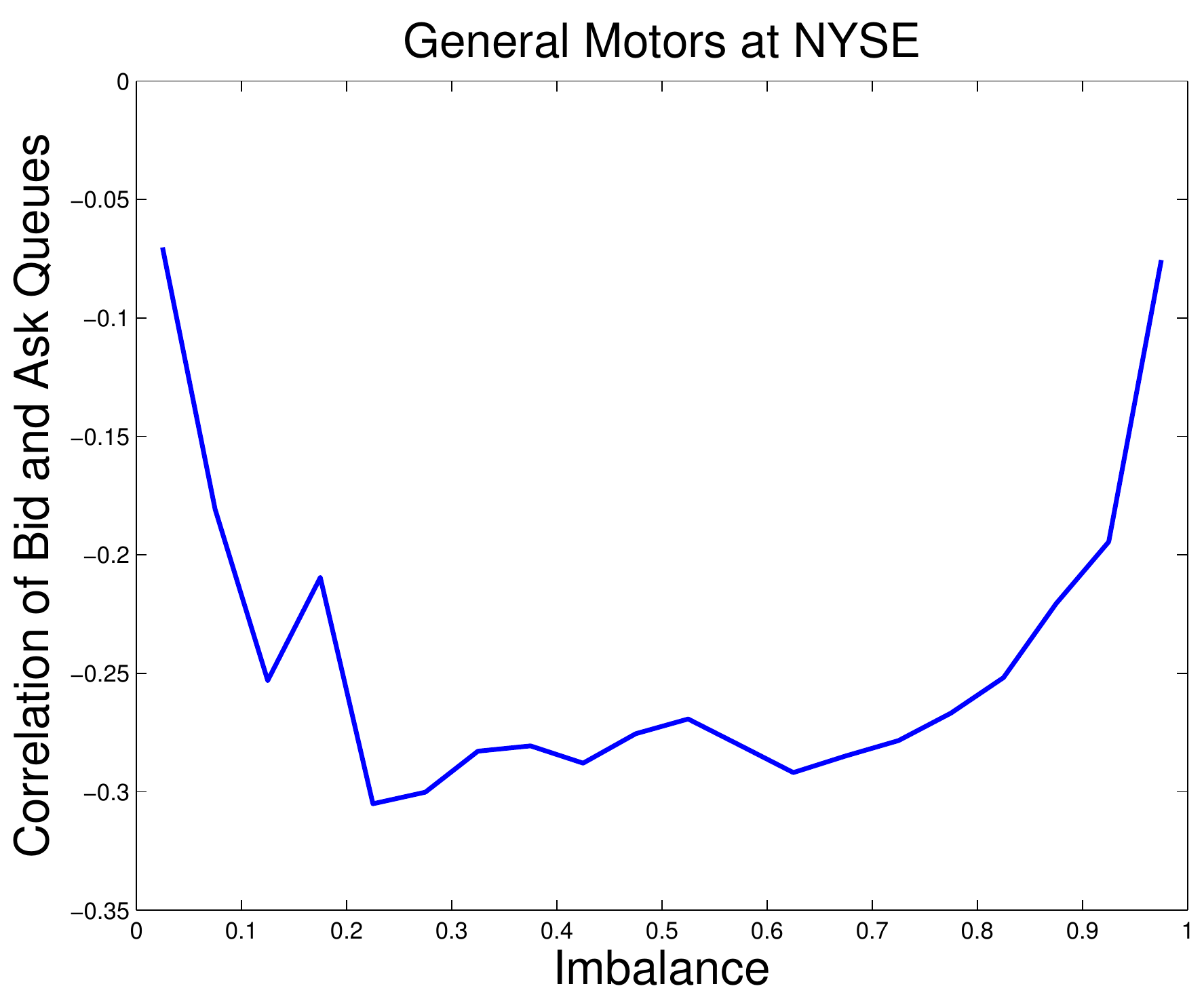}
	\includegraphics[width=0.32\textwidth]{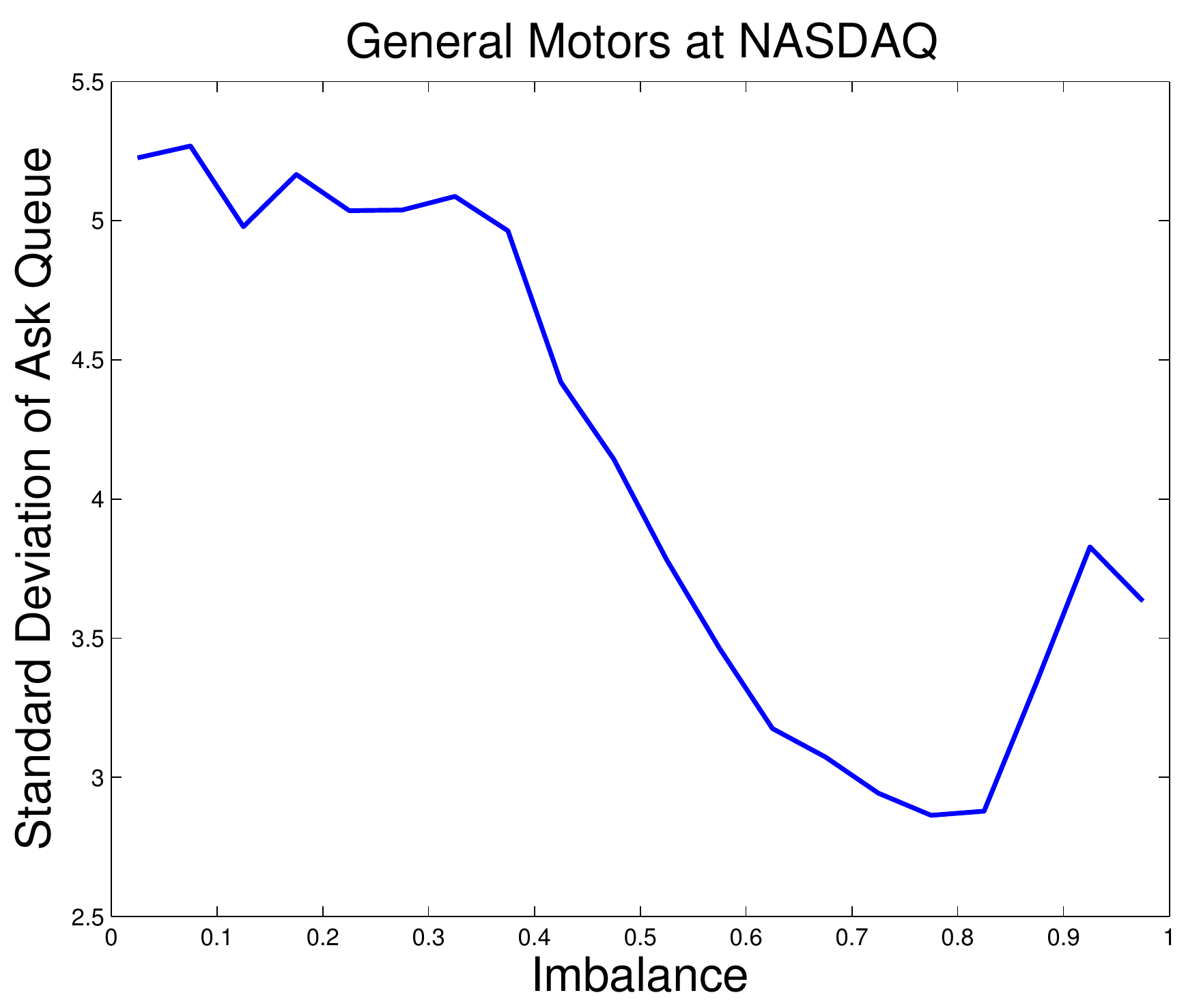}
	\includegraphics[width=0.32\textwidth]{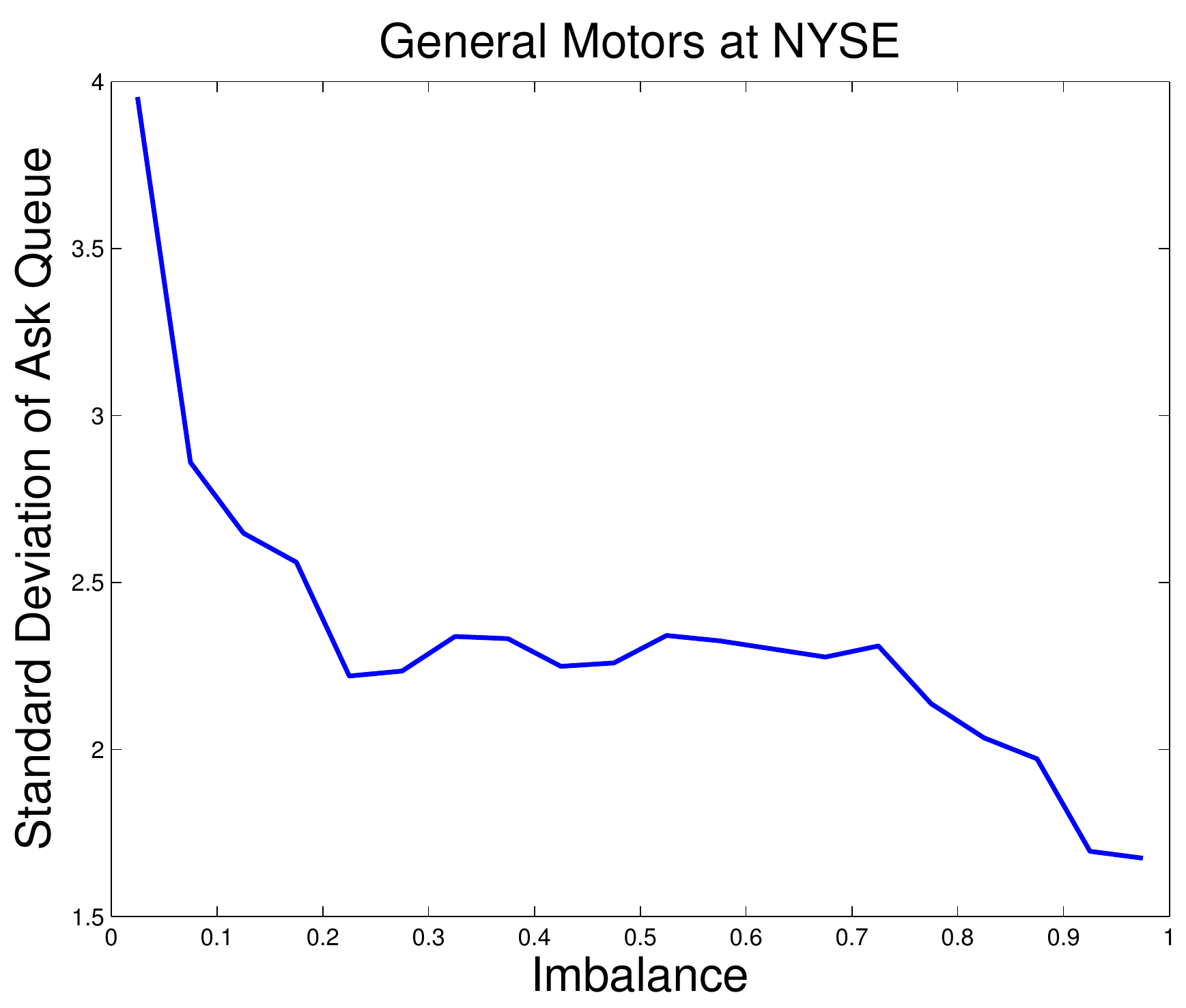}

	\caption{Correlations and Standard Deviations of the Volumes at the Best Bid and the Best Ask
		of General Motors at NASDAQ and NYSE}
	\label{GMcorr}
\end{figure}

\begin{figure}
	\centering
	\includegraphics[width=0.32\textwidth]{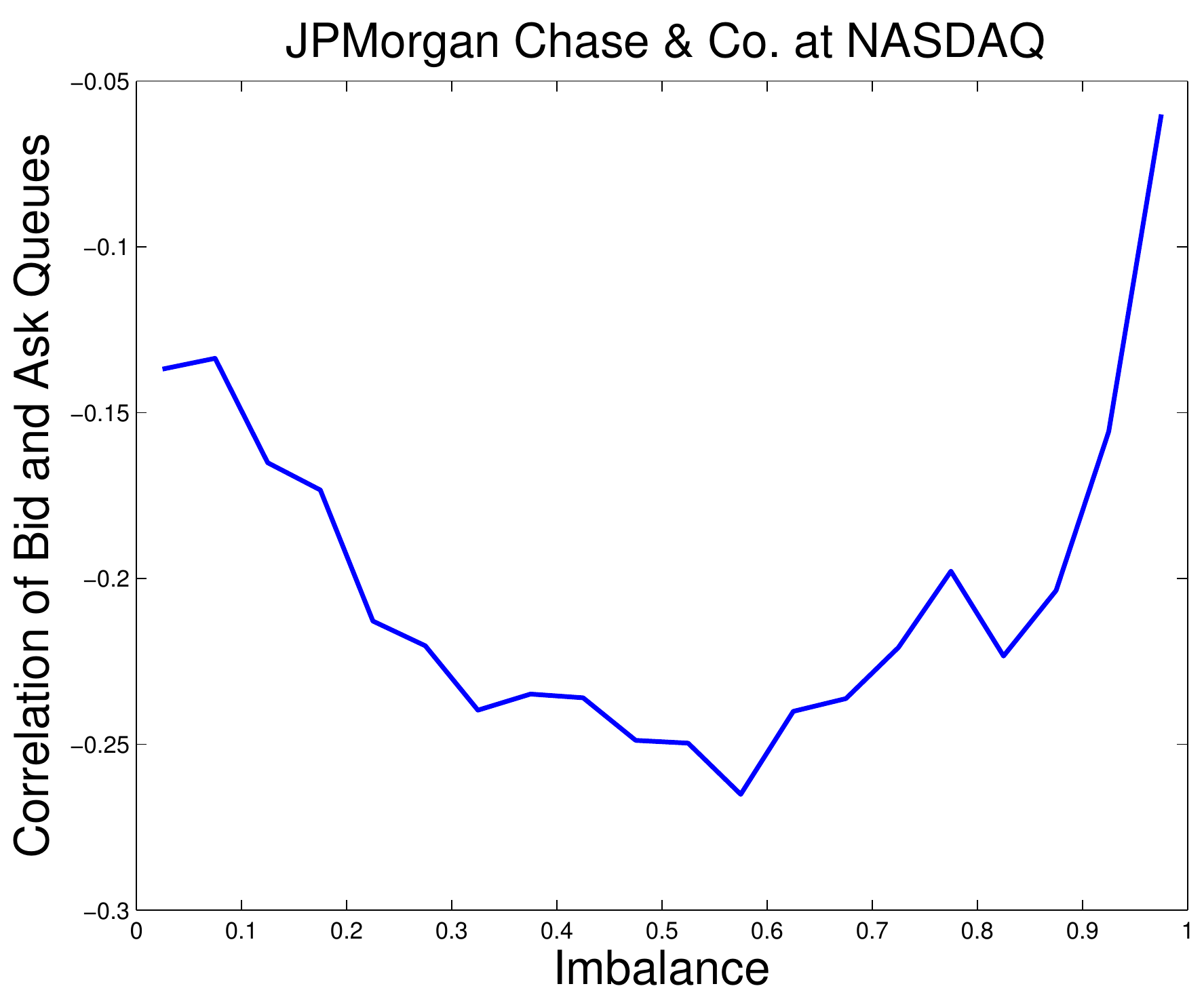}
	\includegraphics[width=0.32\textwidth]{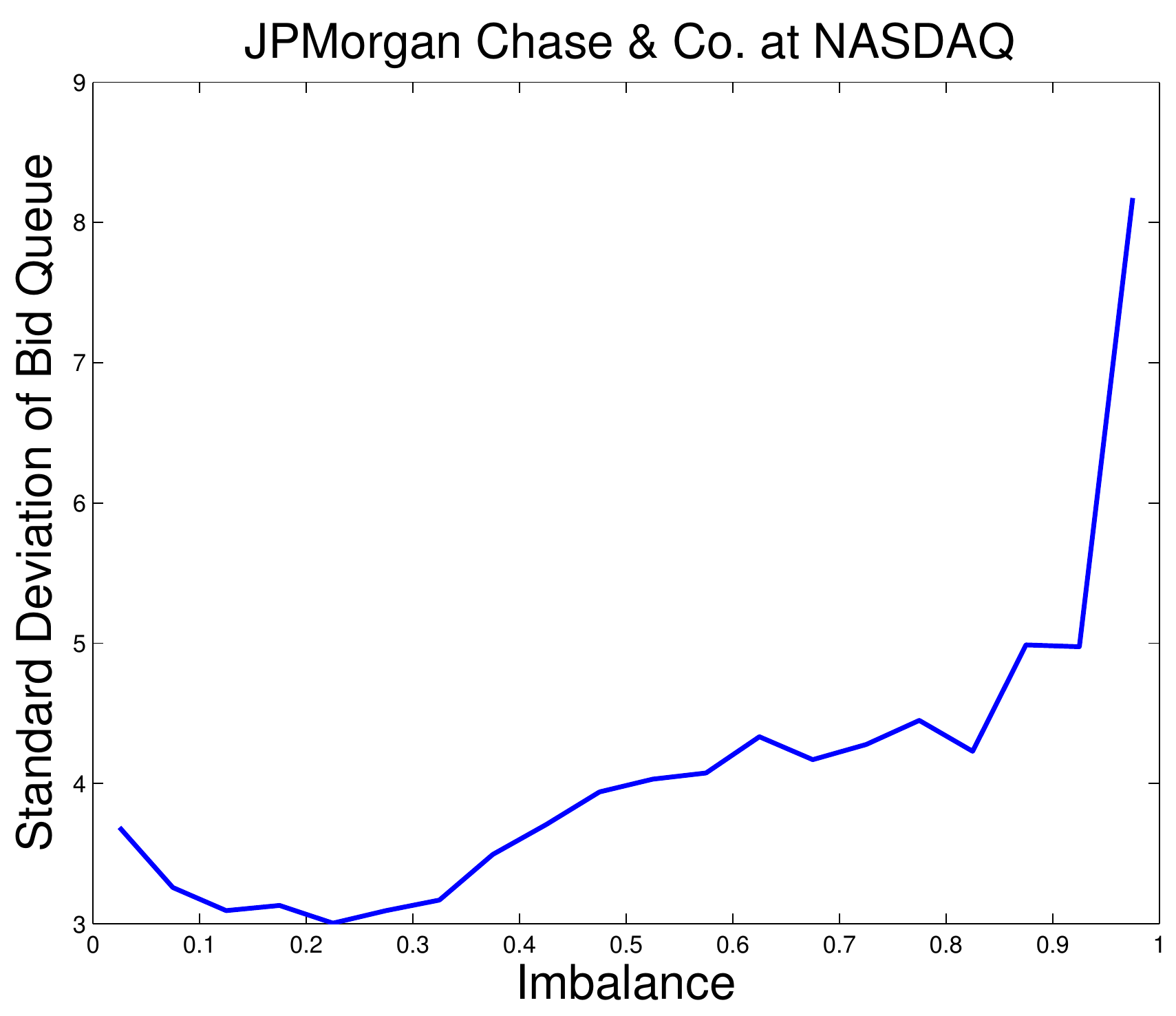}
	\includegraphics[width=0.32\textwidth]{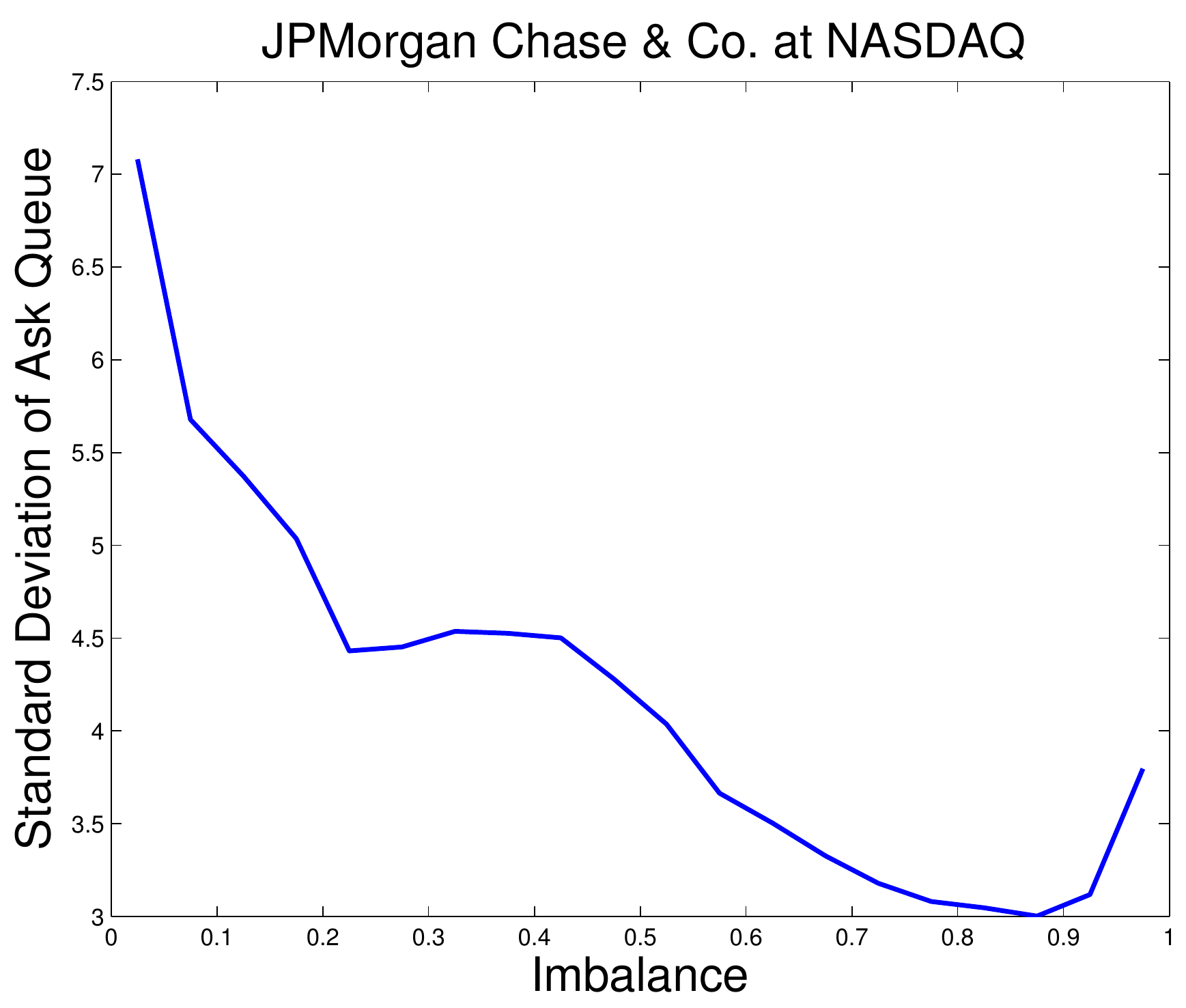}
	
	\includegraphics[width=0.32\textwidth]{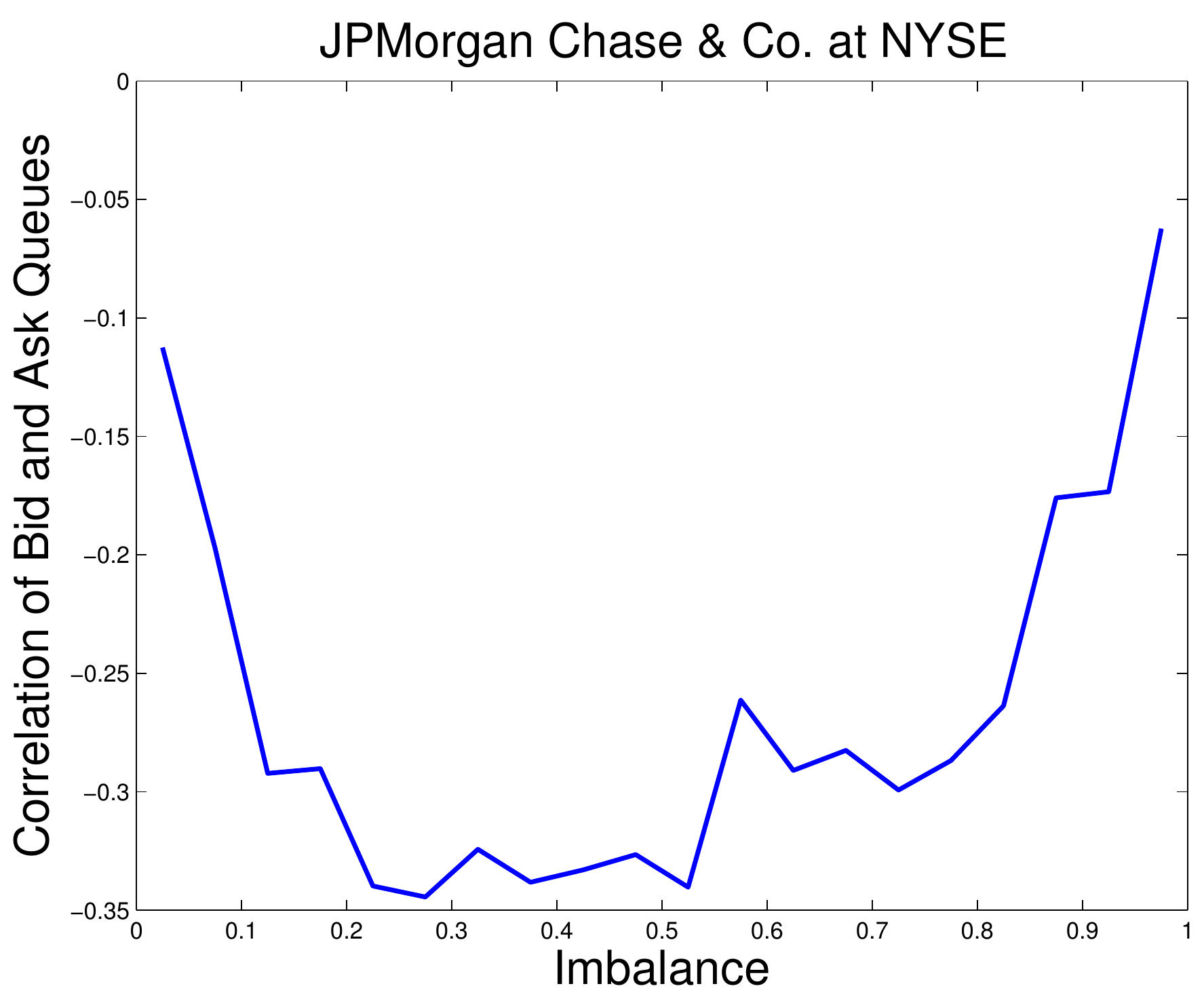}
	\includegraphics[width=0.32\textwidth]{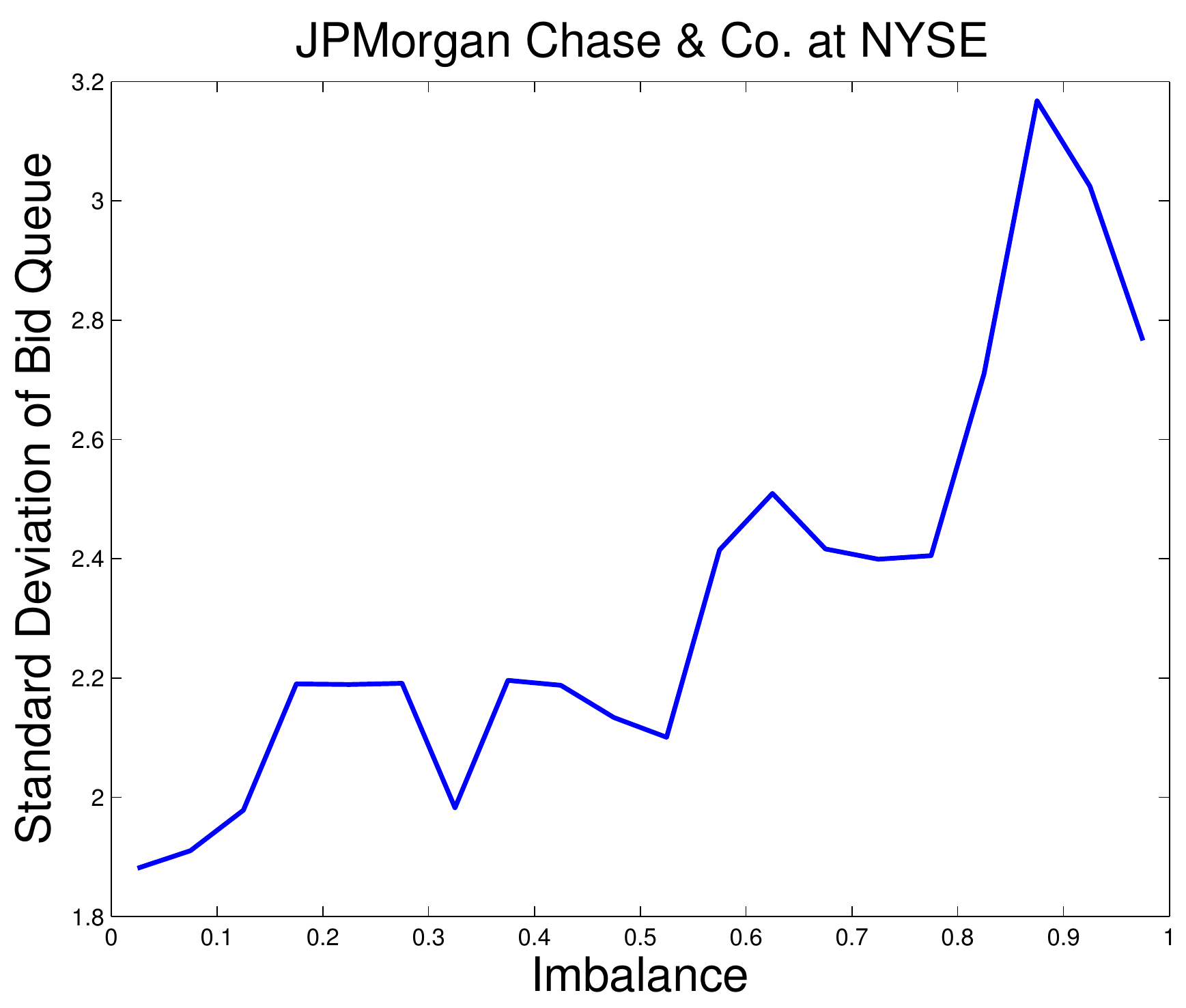}
	\includegraphics[width=0.32\textwidth]{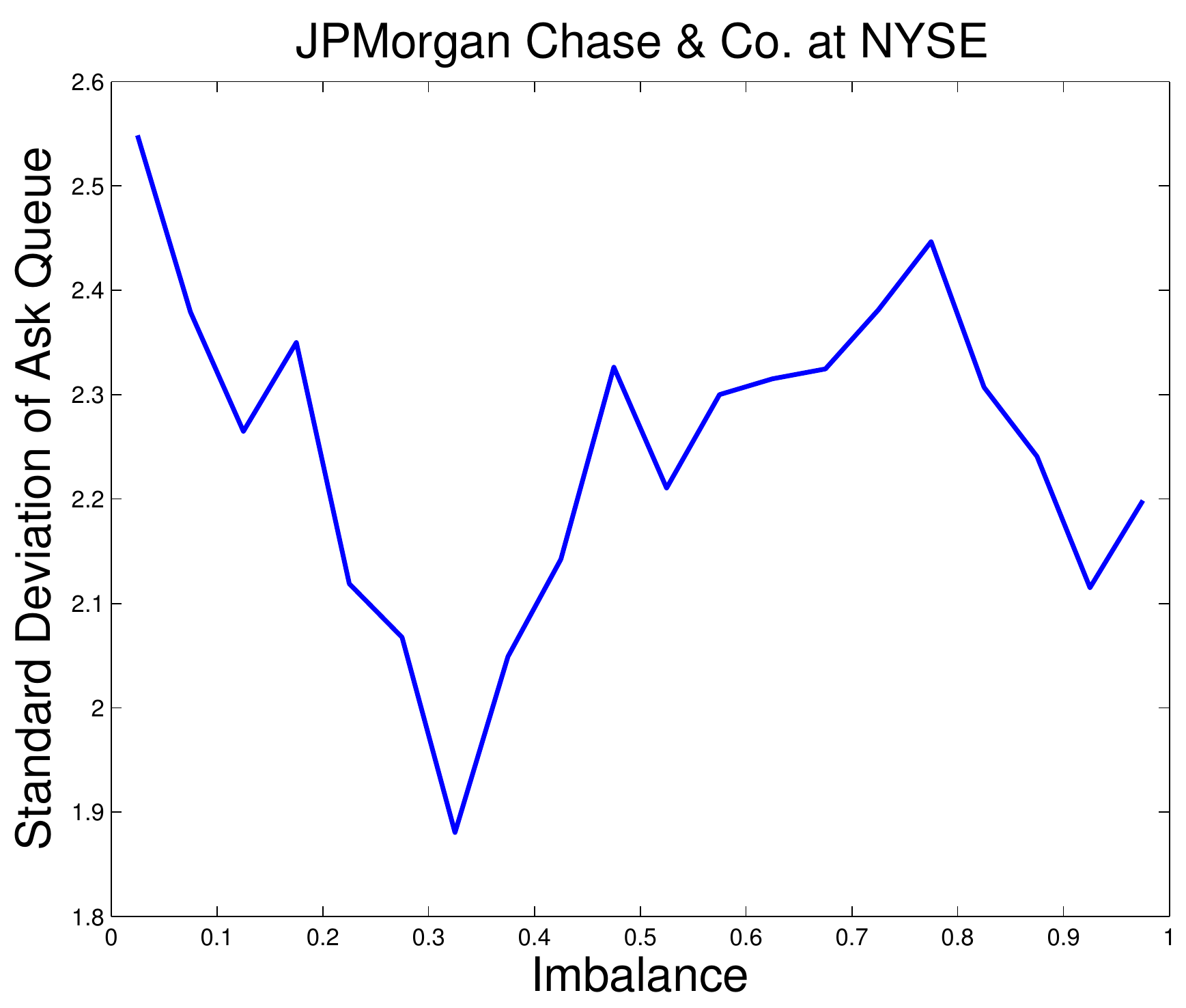}

	\caption{Correlations and Standard Deviations of the Volumes at the Best Bid and the Best Ask
		of JP Morgan \& Chase at NASDAQ and NYSE}
	\label{JPMcorr}
\end{figure}

\begin{figure}
	\centering
	\includegraphics[width=0.49\textwidth]{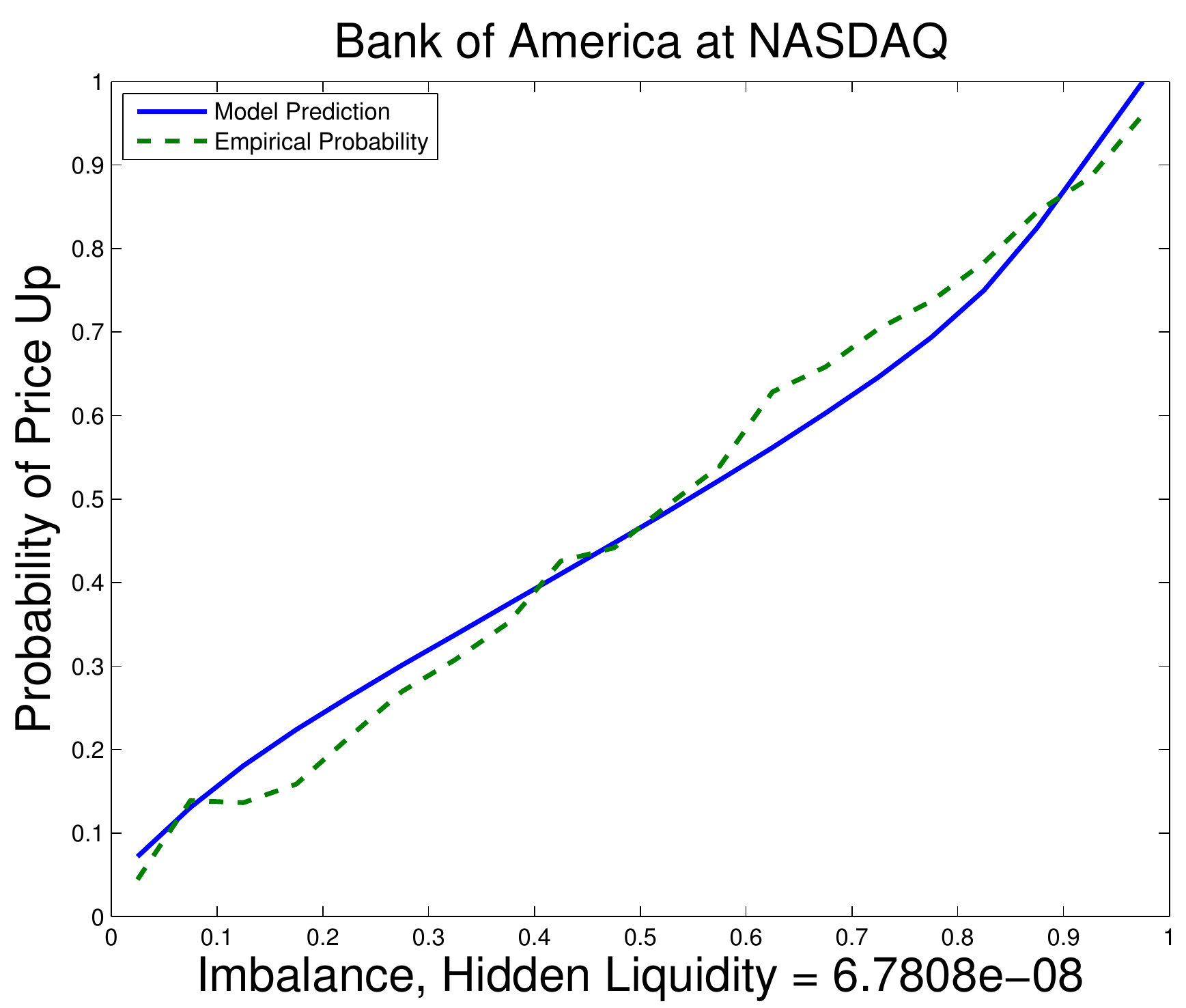}
	\includegraphics[width=0.49\textwidth]{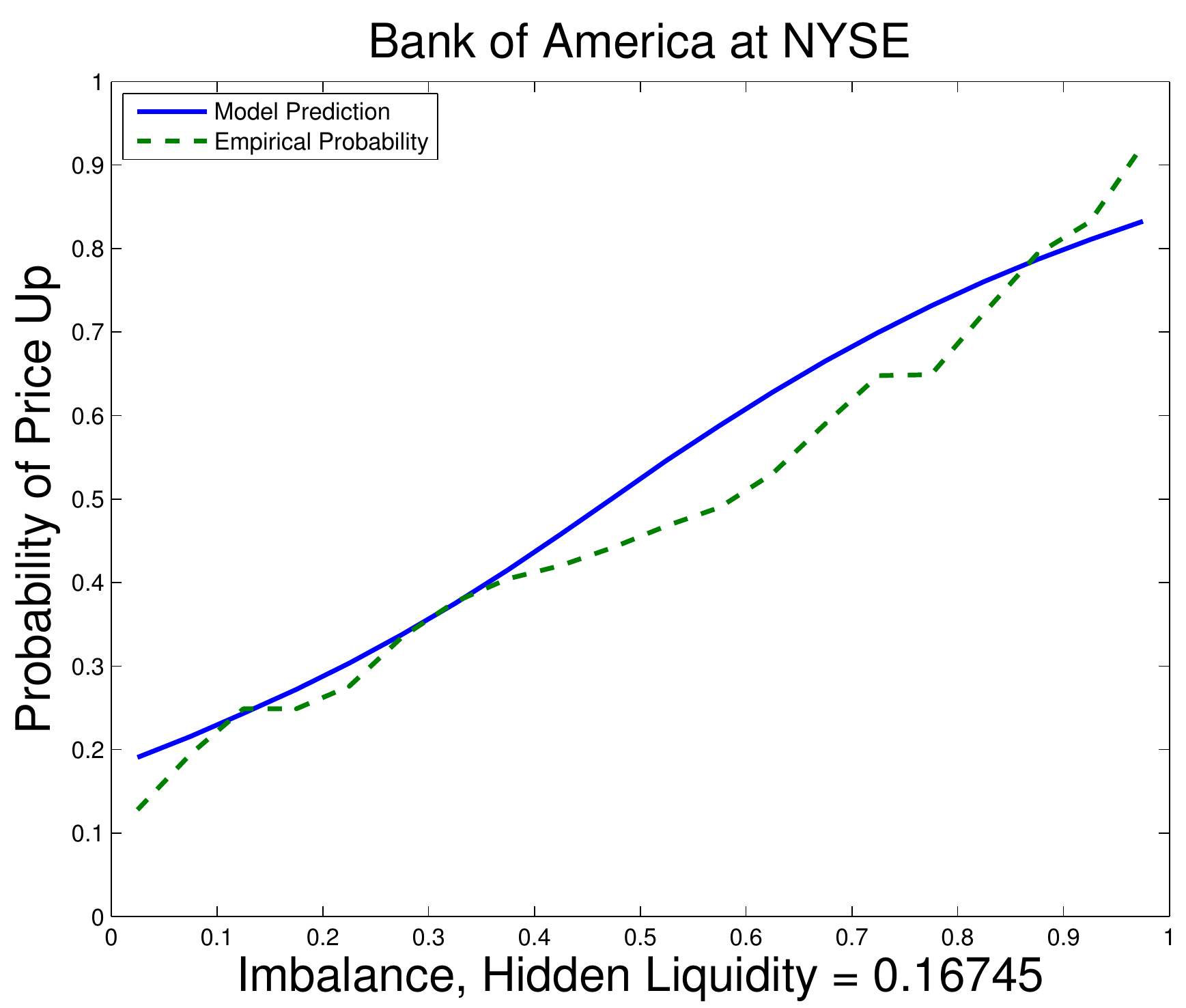}

	\caption{Empirical Probability (Dotted Lines) and Model Prediction (Solid Lines) of Bank of America}
	\label{ProbBAC}
\end{figure}

\begin{figure}
	\centering
	\includegraphics[width=0.49\textwidth]{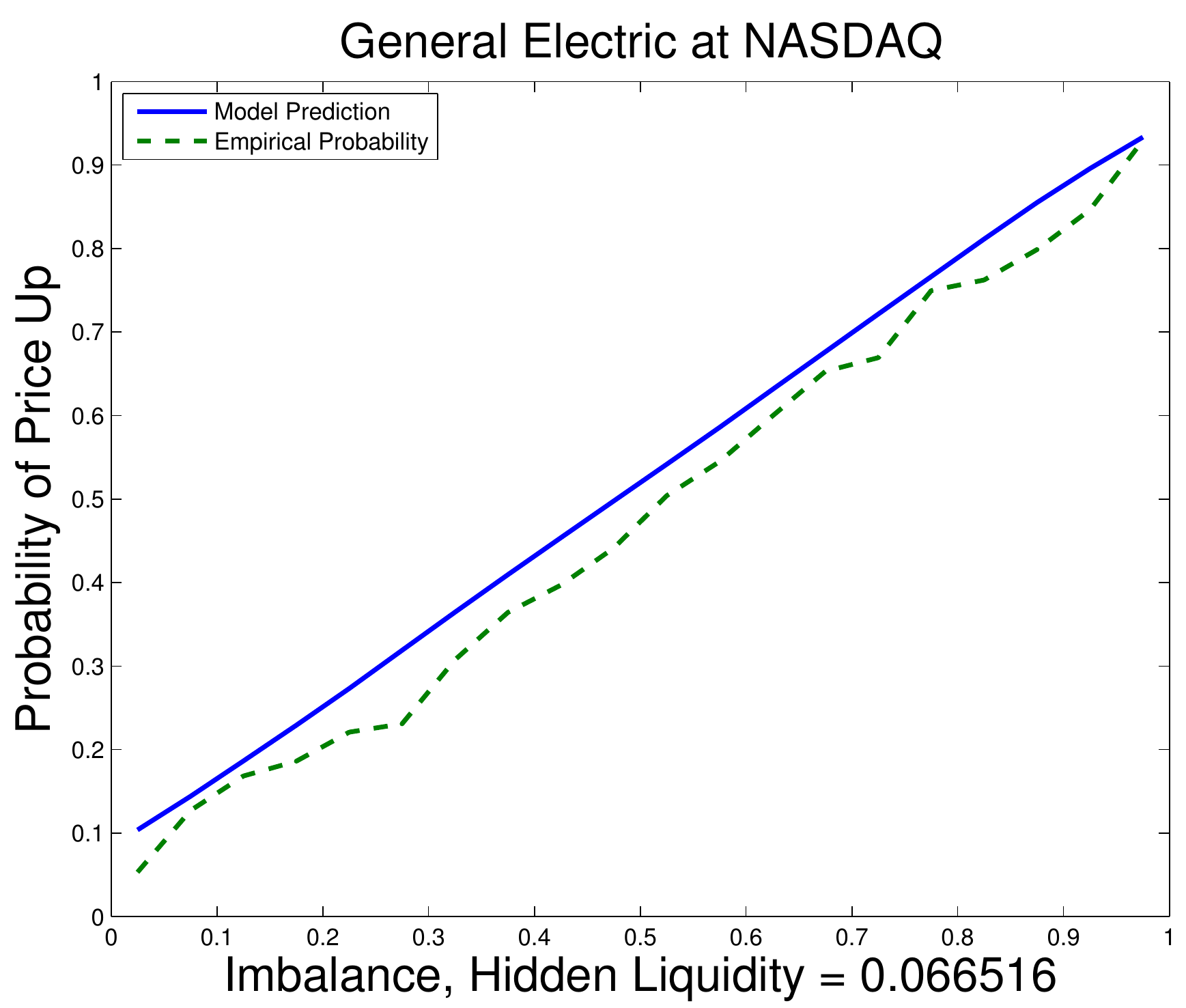}
	\includegraphics[width=0.49\textwidth]{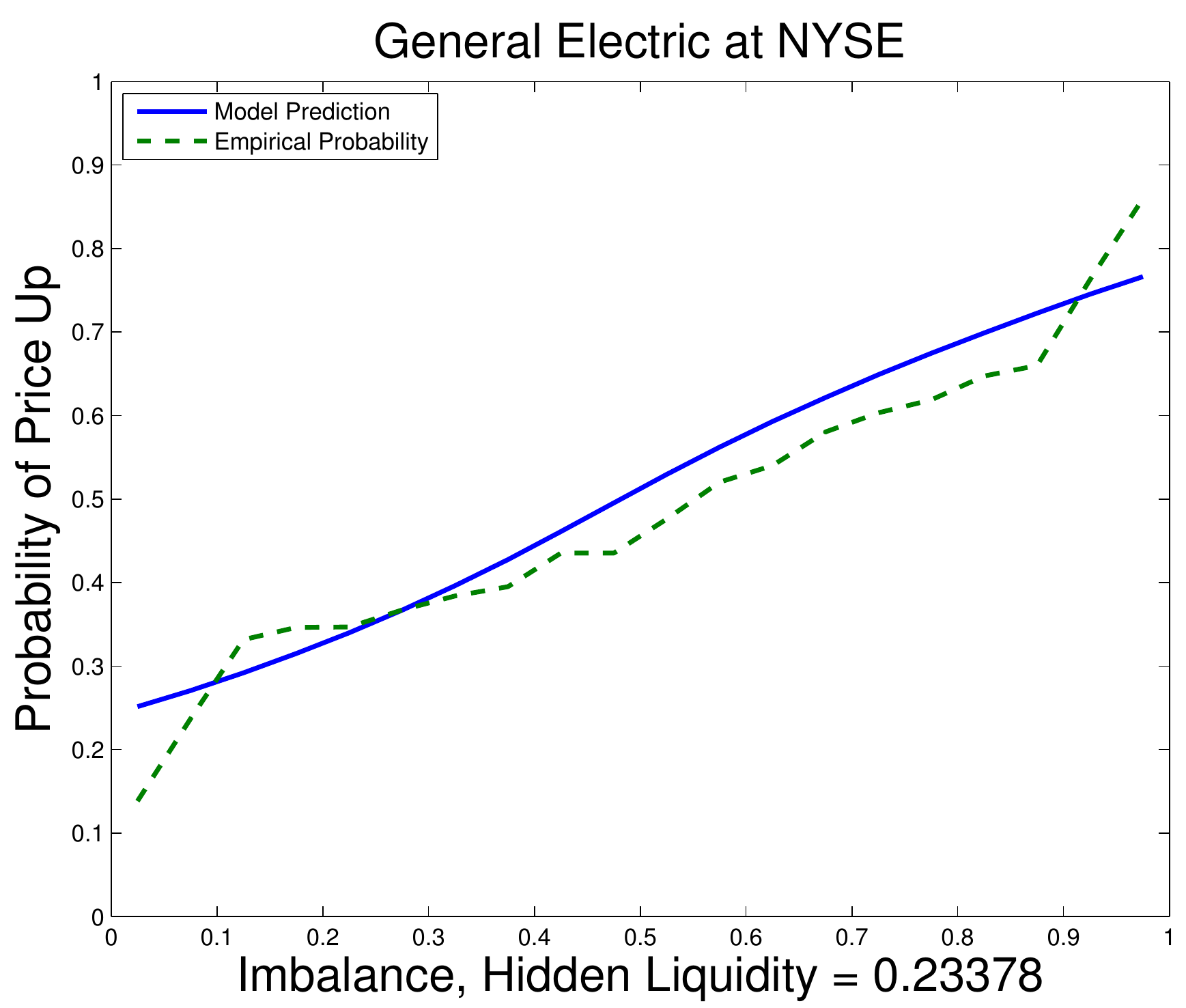}

	\caption{Empirical Probability (Dotted Lines) and Model Prediction (Solid Lines) of General Electric}
	\label{ProbGE}
\end{figure}

\begin{figure}
	\centering
	\includegraphics[width=0.49\textwidth]{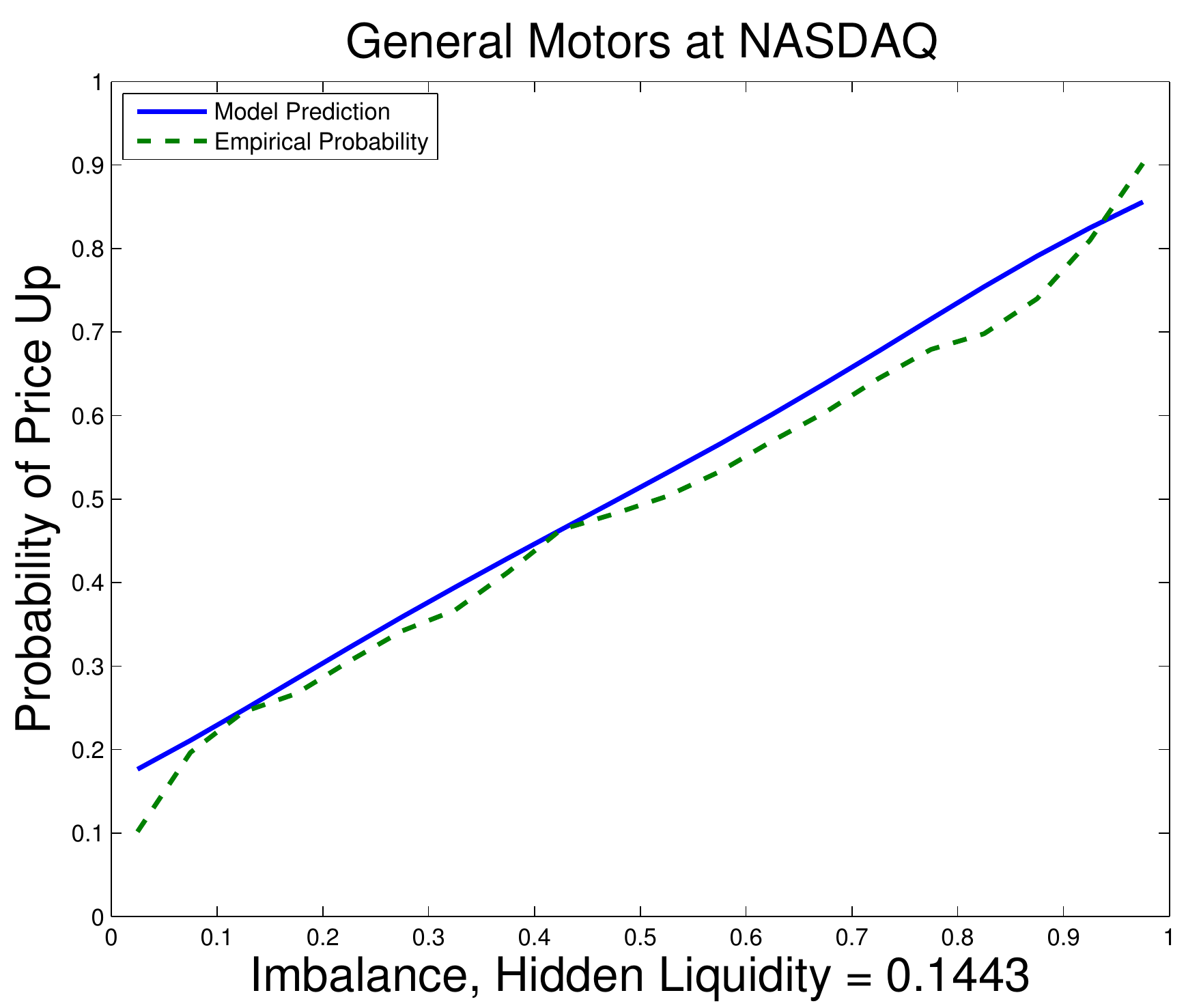}
	\includegraphics[width=0.49\textwidth]{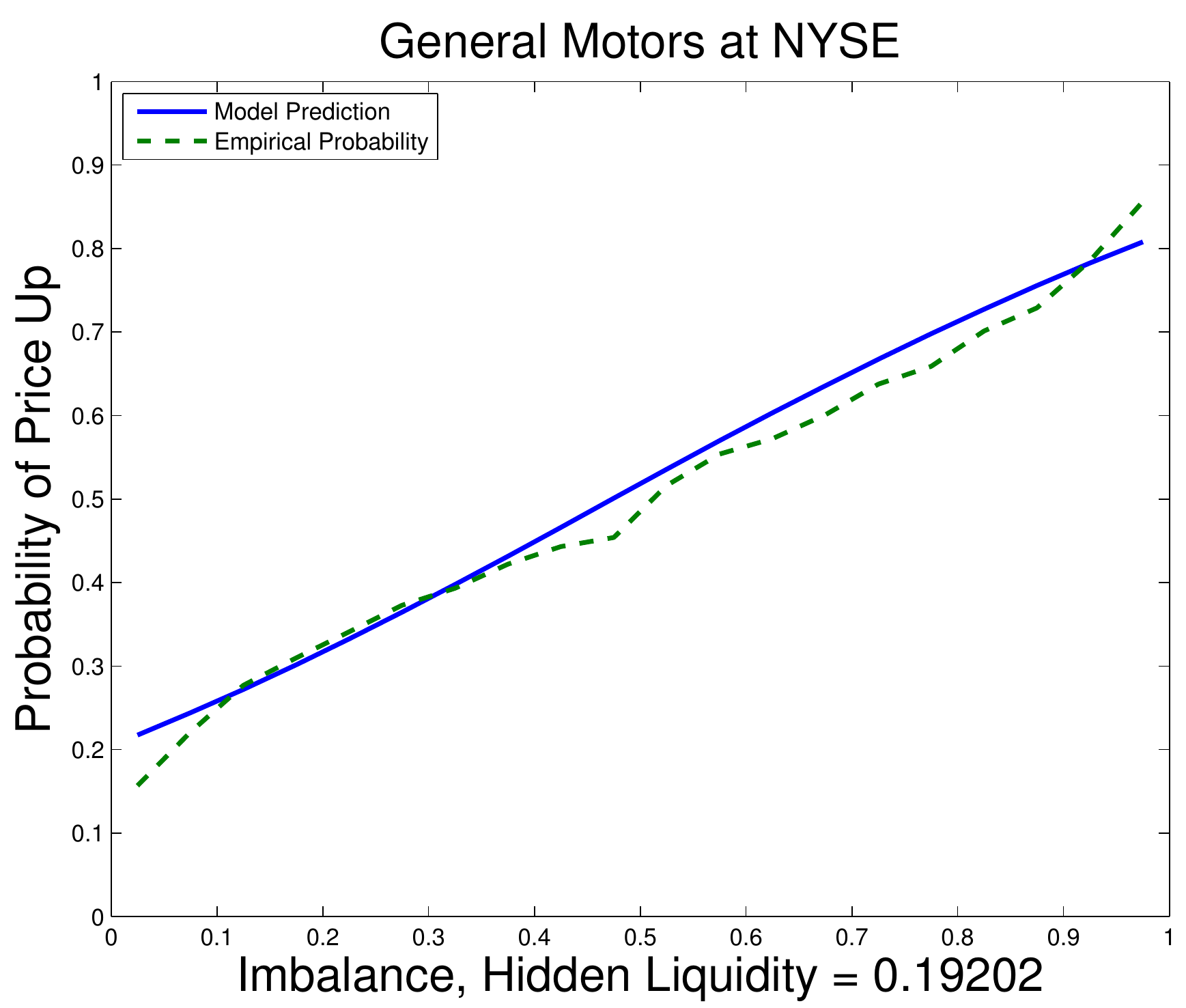}

	\caption{Empirical Probability (Dotted Lines) and Model Prediction (Solid Lines) of General Motors}
	\label{ProbGM}
\end{figure}

\begin{figure}
	\centering
	\includegraphics[width=0.49\textwidth]{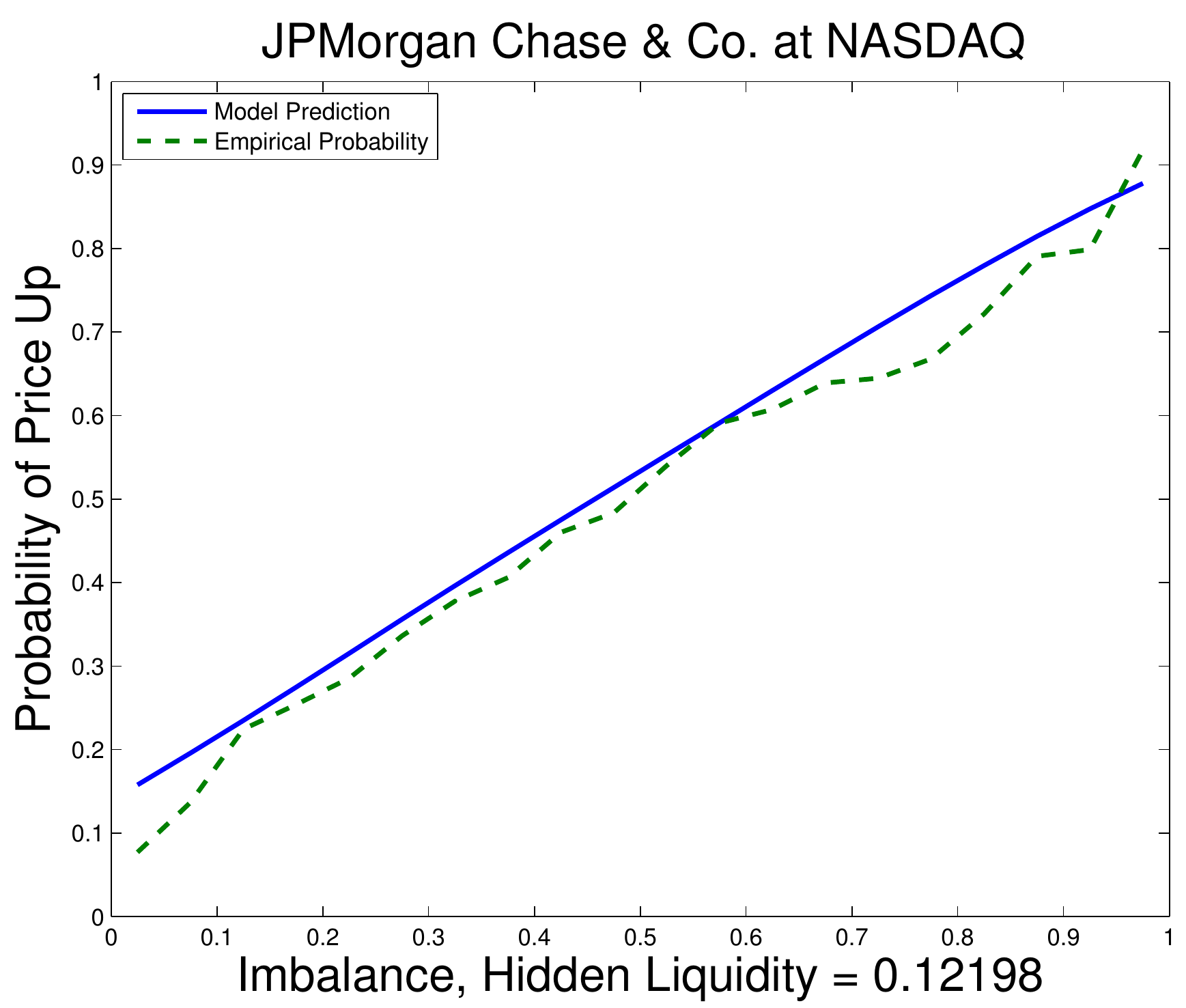}
	\includegraphics[width=0.49\textwidth]{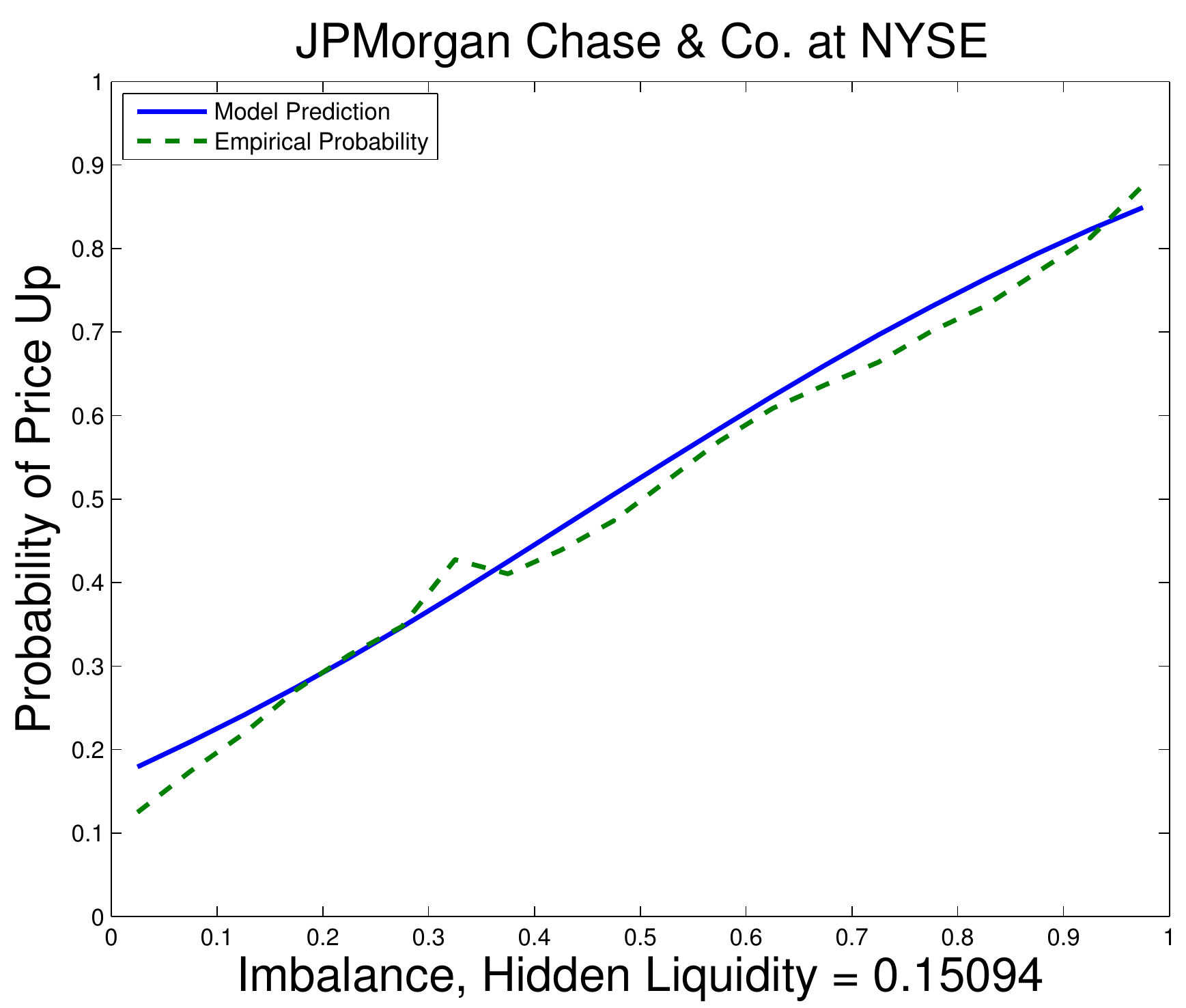}

	\caption{Empirical Probability (Dotted Lines) and Model Prediction (Solid Lines) of JP Morgan Chase \& Co.}
	\label{ProbJPM}
\end{figure}
\section{A Reduced-Form Model}\label{ReducedSection}

The simplest continuous time diffusion model to describe the dynamics of the level-1 limit order book
is the correlated Brownian motions, where $Q^{b}(t)$ and $Q^{a}(t)$ are the queue lengths at the best bid and the best ask
normalized by the median size of the queues, see e.g. Avellaneda \textit{et al.} \cite{Avellaneda2011}:
\begin{align}
	dQ^{b}(t)&=\sigma dW^{b}(t),\quad Q^{b}(0)=x,\\
	dQ^{a}(t)&=\sigma dW^{a}(t),\quad Q^{a}(0)=y,
\end{align}
where $W^{b}(t)$ and $W^{a}(t)$ are two correlated standard Brownian motions
with correlation $-1\leq\rho\leq 1$. 

We are interested in the probability of the price movement. 
The probability that the price moves up and down
are given respectively by
\begin{equation}
	\label{eq:prob of price movement}
	P_{\text{up}}=\mathbb{P}(\tau^{a}<\tau^{b}),\qquad
	P_{\text{down}}=\mathbb{P}(\tau^{b}<\tau^{a}),
\end{equation}
where
\begin{equation}
	\label{eq:time to price movement}
	\tau^{a}:=\inf\{t>0: Q^{a}(t)\leq 0\},\qquad
	\tau^{b}:=\inf\{t>0: Q^{b}(t)\leq 0\}.
\end{equation}

Let the probability of price moving up be:
\begin{equation}
u(x,y):=\mathbb{P}(\tau^{a}<\tau^{b}|Q^{b}(0)=x,Q^{a}(0)=y). 
\end{equation}
It is known that, see e.g. Avellaneda \textit{et al.} \cite{Avellaneda2011}:
\begin{equation}
	\label{BMmodel}
	u(x,y)=\frac{1}{2}\left[1-\frac{\arctan\left(\sqrt{\frac{1+\rho}{1-\rho}}\frac{y-x}{y+x}\right)}
	{\arctan\left(\sqrt{\frac{1+\rho}{1-\rho}}\right)}\right].
\end{equation}
When there is no correlation, i.e., $\rho=0$:
\begin{equation}
	\label{ZeroCorr}
	u(x,y)=\frac{2}{\pi}\arctan\left(\frac{x}{y}\right).
\end{equation}
When the correlation is perfectly negative , i.e., $\rho=-1$:
\begin{equation}
	\label{NegCorr}
	u(x,y)=\frac{x}{x+y}.
\end{equation}
From \eqref{BMmodel}, we can see that the probability of price moving up can be written 
as a function depending only on the imbalance:
\begin{equation}
	P_{\text{up}}(z)=\frac{1}{2}\left[1-\frac{\arctan\left(\sqrt{\frac{1+\rho}{1-\rho}}(1-2z)\right)}
	{\arctan\left(\sqrt{\frac{1+\rho}{1-\rho}}\right)}\right],
\end{equation}
where $z=\frac{x}{x+y}$. Moreover, $P_{\text{up}}(z)$ is monotonically increasing in the imbalance $z$.

\begin{remark}
	More generally, we can assume that the diffusion processes
	are correlated Brownian motions with constant drifts:
	\begin{align}
		dQ^{b}(t)&=\mu^{b}dt+\sigma^{b}dW^{b}(t),\qquad Q^{b}(0)=x,\\
		dQ^{a}(t)&=\mu^{a}dt+\sigma^{a}dW^{a}(t),\qquad Q^{a}(0)=y.
	\end{align}
	Based on the results in Iyengar \cite{Iyengar1985} and Metzler \cite{Metzler2010}, we have
	\begin{equation}
		u(x,y) = \int_{0}^{\infty}\int_{0}^{\infty}
		e^{\gamma^{a}(r\cos\alpha-z^{a})+\gamma^{b}(r\sin\alpha-z^{b})	-\frac{(\gamma^{a})^{2}+(\gamma^{b})^{2}}{2}t}g(t,r)drdt,
	\end{equation}
	where 
	\begin{equation}
	\left[
	\begin{array}{c}
	\gamma^{a}
	\\
	\gamma^{b}
	\end{array}
	\right]
	=
	\left[
	\begin{array}{cc}
	\sigma^{a}\sqrt{1-\rho^{2}} & \sigma^{a}\rho
	 \\
	0 & \sigma^{b}
	\end{array}
	\right]^{-1}
	\left[
	\begin{array}{c}
	\mu^{a}
	\\
	\mu^{b}
	\end{array}
	\right],
	\qquad
	\left[
	\begin{array}{c}
	z^{a}
	\\
	z^{b}
	\end{array}
	\right]
	=
	\left[
	\begin{array}{cc}
	\sigma^{a}\sqrt{1-\rho^{2}} & \sigma^{a}\rho 
	 \\
	0 & \sigma^{b}
	\end{array}
	\right]^{-1}
	\left[
	\begin{array}{c}
	y
	\\
	x
	\end{array}
	\right],
	\end{equation}
	and $g(t,r)=\frac{\pi}{\alpha^{2}tr}\exp(-\frac{r^{2}+r_{0}^{2}}{2t})
	\sum_{n=1}^{\infty}n\sin(n\pi(\alpha-\theta_{0})/\alpha)I_{n\pi/\alpha}(rr_{0}/t)$,
	where $I_{\nu}$ is the modified Bessel function of the first kind and
	\begin{align*}
		\alpha &:=
		\begin{cases}
			\pi+\arctan(-\sqrt{1-\rho^2}/\rho), &\rho>0\\
			\frac{\pi}{2},  &\rho=0\\
			\arctan(-\sqrt{1-\rho^2}/\rho), &\rho<0
		\end{cases}\\
		r_{0} &:= \sqrt{(x/\sigma^{b})^2+(y/\sigma^{a})^2-2\rho(x/\sigma^{b})(y/\sigma^{a})}/\sqrt{1-\rho^2}\\
		\theta_{0} &:=
		\begin{cases}
			\pi+\arctan\left(\frac{(x/\sigma^{b})\sqrt{1-\rho^2}}{y/\sigma^{a}-\rho x/\sigma^{b}}\right), &y/\sigma^{a}<\rho x/\sigma^{b},\\
			\frac{\pi}{2}, &y/\sigma^{a}=\rho x/\sigma^{b},\\
			\arctan\left(\frac{(x/\sigma^{b})\sqrt{1-\rho^2}}{y/\sigma^{a}-\rho x/\sigma^{b}}\right), &y/\sigma^{a}>\rho x/\sigma^{b}.
		\end{cases}
	\end{align*}
	In particular, when $\mu^{a}=\mu^{b}=0$,
	\begin{equation}
		\label{ZDrift}
		u(x,y)=\frac{\theta_{0}}{\alpha}.
	\end{equation}
\end{remark}

In Avellaneda \textit{et al.} \cite{Avellaneda2011}, the authors fitted the empirical probability of mid-price moving up 
by the correlated Brownian motion model when the correlation is $\rho=-1$, that is, \eqref{NegCorr}.
From Figures \ref{ProbBAC}-\ref{ProbJPM} and Table \ref{NTable}-\ref{corrTable} the empirical
probability of mid-price moving up is indeed linearly dependent on the imbalance. 
However, as we have already seen in Figures \ref{BACcorr}, \ref{GEcorr}, \ref{GMcorr}, \ref{JPMcorr} 
that the correlation is negative, but far away from $-1$, and it also depends
on the level of the imbalance. Therefore, a perfect negatively correlated Brownian motions model
might not fit both the empirical probability and the empirical correlation. 
We will propose a non-parametric diffusion model that can fit the empirical correlation, empirical volatilities, and empirical probability
of price movement simultaneously.

The correlated Brownian motion is simple yet still captures the phenomenon
that the price movement is mainly driven by the imbalance at the best bid and ask level. 
We are interested to investigate further the relation of the dynamics of the volumes
at the best bid and ask level and the imbalance. 
The assumption in the model \eqref{BMmodel} that the correlation and volatility
of the volumes at the best bid and ask levels are constant might be oversimplified
and not consistent with the real data. 
Indeed, the empirical studies we did in Section \ref{DataSection}
suggests that the correlation of the movements of the volumes at the best bid and ask level
is non-trivially dependent on the imbalance. Two universal shapes for the correlation as a function
of the imbalance are the $U$-shaped curve and $W$-shaped curve. For a $U$-shaped correlation
function of the imbalance, the correlation is negative and it is close to zero when
the imbalance is close to $0$ and $1$ and it is the most negative when the imbalance is close to $\frac{1}{2}$.
Similarly we also observe $W$-shaped correlation curves. The correlations are consistently negative though far away from $-1$.
We also observe that the volatilities of the volumes at the best bid and ask levels also depend non-trivially
on the imbalance. The volatility is in general large when the imbalance is small or large and the volatility
is small when the imbalance is moderate. The difference here is that instead of a symmetric $U$-shaped or $W$-shaped curve,
we often get two skewed $U$-shaped curves, depending on whether we consider the best bid or the best ask.
Therefore, our goal is to improve the model \eqref{BMmodel} to allow the correlation and volatilities
to be non-constant and depend on the level of the imbalance. 

In a very loose analogy, 
in the literature of the pricing of derivative securities, it is well known that the stock price has the so-called
leverage effect, that is, the volatility of a stock tends to increase when the stock price drops, which is one of the key reasons
that people have used the CEV models and other local volatility models as an alternative to the Black-Scholes model
in which the volatility is always constant.

We are interested to build up a model for the dynamics of the level-1 limit order books, 
that can capture the empirical evidence that we observed from the data.

Let us build a discrete model and find its diffusion approximation.
Let $X(t),Y(t)$ be the queue lengths at the best bid and the best ask at time $t$
and $Z_{t}=\frac{X(t)}{X(t)+Y(t)}$ be the imbalance. 
Let us assume that
\begin{itemize}
\item
The limit orders that arrive at the best bid is a simple point process $N^{1}(t)$
with intensity $\lambda^{1}(Z_{t-})$ at time $t$;
\item
The market orders or cancellations that arrive at the best bid is a simple point process $N^{2}(t)$
with intensity $\lambda^{2}(Z_{t-})$ at time $t$;
\item
The limit orders that arrive at the best ask is a simple point process $N^{3}(t)$
with intensity $\lambda^{3}(Z_{t-})$ at time $t$;
\item
The market orders or cancellations that arrive at the best ask is a simple point process $N^{4}(t)$
with intensity $\lambda^{4}(Z_{t-})$ at time $t$;
\item
There are simultaneous cancellations at
the best ask and limit orders at the best bid that is a simple point process $N^{5}(t)$
with intensity $\lambda^{5}(Z_{t-})$ at time $t$;
\item
There are simultaneous cancellations at
the best bid and limit orders at the best ask that is a simple point process $N^{6}(t)$
with intensity $\lambda^{6}(Z_{t-})$ at time $t$;
\end{itemize}
The last two assumptions above are made due to the observation that the empirical correlation
between the best bid and ask queues are always negative. Note that $\lambda^{1}$ is the arrival
rate for the idiosyncratic limit orders at the best bid, so the total arrival rate
for the limit orders at the best bid is $\lambda^{1}+\lambda^{5}$. Similarly, 
the total arrival rate for the limit orders at the best ask is $\lambda^{3}+\lambda^{6}$.
For $1\leq j\leq 6$, we assume that $\lambda^{j}(z) = \lambda^{j}(\frac{x}{x+y}):\mathbb{R}\rightarrow\mathbb{R}^{+}$
are continuous and bounded (there is singularity when $x+y=0$ and we assume analytic continuation of $\lambda^{j}$
at the singularity). Finally, for simplicity, we assume
that the order size has unit size $1$. Note that all the following arguments work if we assume
constant order sizes for different types of orders.
Therefore, the dynamics at the best bid and ask are given by:
\begin{align}
	\label{DiscreteModel}
	dX(t)&=dN^{1}(t)-dN^{2}(t)+dN^{5}(t)-dN^{6}(t),\\
	dY(t)&=dN^{3}(t)-dN^{4}(t)+dN^{6}(t)-dN^{5}(t).\nonumber
\end{align}
Since empirically, we do not observe strong evidence for the drift effect, we assume the driftless condition:
\begin{equation}
	\label{DriftlessAssump}
	\lambda^{1}(z)-\lambda^{2}(z)=\lambda^{6}(z)-\lambda^{5}(z)=\lambda^{4}(z)-\lambda^{3}(z),
\end{equation}
so that $X(t)$ and $Y(t)$ are driftless, in the sense that
\begin{align*}
	dX(t) &=dM^{1}(t)-dM^{2}(t)+dM^{5}(t)-dM^{6}(t)\\
	dY(t) &=dM^{3}(t)-dM^{4}(t)+dM^{6}(t)-dM^{5}(t),
\end{align*}
where for any $1\leq j\leq 6$, $M^{j}(t):=N^{j}(t)-\int_{0}^{t}\lambda^{j}(Z_{s-})ds$
is a martingale. For the high frequency trading, the number of orders is large and the trading frequency is high, 
so we can rescale time and space to get a diffusion approximation to the discrete model. 
Let us define the rescaled process for $1\leq j\leq 6$,
\begin{equation}
	\label{eq:rescaled process}
	X_{n}(t):=\frac{1}{\sqrt{n}}X(nt),\quad
	Y_{n}(t):=\frac{1}{\sqrt{n}}Y(nt),\quad
	M_{n}^{j}(t):=\frac{1}{\sqrt{n}}M^{j}(nt).
\end{equation}

The discrete model \eqref{DiscreteModel} describes the dynamics of the best bid and ask queues at the micro level, but may not
be easy to work with when we are interested to compute the probability of mid-price movement. So, next, let us find a diffusion approximation
to the discrete model \eqref{DiscreteModel}.

Let us assume that
\begin{align}
	\label{SDE}
	& dQ^{b}(t) = \sigma^{b}(Z(t))dW^{b}(t),\quad Q^{b}(0)=x>0\\
	& dQ^{a}(t) = \sigma^{a}(Z(t))dW^{a}(t),\quad Q^{a}(0)=y>0\notag \\
	& Z(t) = \frac{Q^{b}(t)}{Q^{b}(t)+Q^{a}(t)},\notag
\end{align}
where $W^{b}(t)$ and $W^{a}(t)$ are two standard Brownian motions with correlation
$\rho(Z(t))$ at time $t$. 
We assume that $(x,y)\mapsto\sigma^{b}(\frac{x}{x+y})$, $(x,y)\mapsto\sigma^{a}(\frac{x}{x+y})$
are bounded and continuous from $\mathbb{R}^2$ to $\mathbb{R}^{+}$, 
and $(x,y)\mapsto\rho(\frac{x}{x+y})$
is bounded and continuous\footnote{Note that $\frac{x}{x+y}$ can be singular, and $\sigma^{b},\sigma^{a},\rho$ 
are defined as the analytic continuation at the singular points $\pm\infty$.} from $\mathbb{R}^2$ to $[-1,1]$, so that
there exists a unique solution to \eqref{SDE}, see e.g. \cite{StroockVaradhan}. If in addition, we assume that $\sigma^{b},\sigma^{a},\rho$ are Lipschitz,
then the solution is guaranteed to be strong, see e.g. \cite{Oksendal2003}.

Note that the discrete process $(X_{n}(t),Y_{n}(t))$ and $(Q^{b}(t),Q^{a}(t))$ should both
live in the first quadrant. But to avoid the well-definedness after the process hitting the boundary of the first
quadrant, we make the processes well-defined on $\mathbb{R}^{2}$. Since our goal is to compute the probability
of mid-price movement, which is about the first hitting time of the boundary of the first quadrant, the extension from the first quadrant
to $\mathbb{R}^{2}$
will not alter the results, and it is just for the sake of convenience. 

The discrete model \eqref{DiscreteModel} can be approximated by the diffusion model \eqref{SDE} as follows.

\begin{theorem}
	\label{ConvergenceThm}
	Given that $(X(t),Y(t))$ is the discrete model of the best bid and ask queues in (\ref{DiscreteModel}), assume that for $1\leq j\leq 6$, $\lambda^{j}(z):\mathbb{R}\rightarrow\mathbb{R}^{+}$ are continuous and bounded functions and the driftless condition (\ref{DriftlessAssump}) holds.
	Also assume that $(X_{n}(0),Y_{n}(0))=(x,y)\in\mathbb{R}^{+}\times\mathbb{R}^{+}$.
	Then the rescaled process $(X_{n}(t),Y_{n}(t))$ in (\ref{eq:rescaled process}) converges weakly in $D[0,T]$ as $n\to\infty$ to $(Q^{b}(t),Q^{a}(t))$ in 
	(\ref{SDE}), where $D[0,T]$ is the space of c\`{a}dl\`{a}g processes equipped with Skorohod topology.  In addition, the diffusion and correlation coefficients 
	are the explicit functions of the intensities $\lambda^j(z)$:
	\begin{align*}
		\sigma^{b}(z)&=\left[\lambda^{1}(z)+\lambda^{2}(z)+\lambda^{5}(z)+\lambda^{6}(z)\right]^{1/2}\\
		\sigma^{a}(z)&=\left[\lambda^{3}(z)+\lambda^{4}(z)+\lambda^{5}(z)+\lambda^{6}(z)\right]^{1/2}\\
		\rho(z)&=-\frac{\lambda^{5}(z)+\lambda^{6}(z)}{\sigma^{b}(z)\sigma^{a}(z)}.
	\end{align*}
\end{theorem}

The probability of mid-price movement for the diffusion model \eqref{SDE} can be computed in the closed-form as follows.
\begin{theorem}
	\label{ThmI}
	Given the model (\ref{SDE}), $P_{\text{up}}(z)$, the probability of the price moving up, defined in (\ref{eq:prob of price movement}) 
	and (\ref{eq:time to price movement}), is explicitly given by
	\begin{equation}
		P_{\text{up}}(z)=\frac{\int_{0}^{z}e^{-\int_{0}^{y}\frac{\mu(x)}{\nu(x)}dx}dy}
		{\int_{0}^{1}e^{-\int_{0}^{y}\frac{\mu(x)}{\nu(x)}dx}dy},
	\end{equation}
	where $z$ is the imbalance and
	\begin{align}
		\label{eq:mu and nu}
		\mu(z) &= -2(1-z)\sigma^{b}(z)^{2}+2(2z-1)\rho(z)\sigma^{b}(z)\sigma^{a}(z)+2z\sigma^{a}(z)^{2}\\
		\nu(z) &= (1-z)^{2}\sigma^{b}(z)^{2} - 2z(1-z)\rho(z)\sigma^{b}(z)\sigma^{a}(z) + z^{2}\sigma^{a}(z)^{2}.\notag
	\end{align}
\end{theorem}

\begin{remark}
What we are really interested to compute is the probability of mid-price movement for the discrete model, 
and this can be approximated by the probability of mid-price movement for the diffusion model which has closed-form formula, 
that is given in Theorem \ref{ThmI}. For any $n>0$, we have
\begin{align*}
\mathbb{P}(\text{$X(t)$ hits zero before $Y(t)$ does})
&=\mathbb{P}\left(\text{$\frac{1}{\sqrt{n}}X(nt)$ hits zero before $\frac{1}{\sqrt{n}}Y(nt)$ does}\right)
\\
&\simeq\mathbb{P}(\text{$Q^{b}(t)$ hits zero before $Q^{a}(t)$ does}),
\end{align*}
as $n\rightarrow\infty$. Note that the approximation requires that $Q^{b}(0)=\frac{1}{\sqrt{n}}X(0)$ and
$Q^{a}(0)=\frac{1}{\sqrt{n}}Y(0)$ and this is still reasonable since the formula in Theorem \ref{ThmI}
only depends on the ratio of $Q^{b}(0)$ and $Q^{a}(0)$ so we can rescale the initial condition.
\end{remark}

\begin{remark}
	We can recover the results in \cite{Avellaneda2011} from Theorem \ref{ThmI}:
	\begin{enumerate}
		\item When $\sigma^{b}=\sigma^{a}=\sigma$ and $\rho=-1$, we have $\mu(z)\equiv 0$ and thus $P_{\text{up}}(z)=z=\frac{x}{x+y}$ which recovers \eqref{NegCorr}.
		
		\item When $\sigma^{b}=\sigma^{a}=\sigma$ and $\rho=0$, we have
		$\mu(z)=-2(1-2z)\sigma^{2}$ and $\nu(z)=[(1-z)^{2}+z^{2}]\sigma^{2}$. Therefore,
		\begin{align*}
			P_{\text{up}}(z)
			&= \frac{\int_{0}^{z}e^{2\int_{0}^{y}\frac{1-2x}{1-2x+2x^{2}}dx}dy}{\int_{0}^{1}e^{2\int_{0}^{y}\frac{1-2x}{1-2x+2x^{2}}dx}dy}
			= \frac{\int_{0}^{z}e^{-\log(2(y-1)y+1)}dy}{\int_{0}^{1}e^{-\log(2(y-1)y+1)}dy}
			= \frac{\int_{0}^{z}\frac{1}{1-2y+2y^{2}}dy}{\int_{0}^{1}\frac{1}{1-2y+2y^{2}}dy}\\
			&= \frac{\arctan(1-2z)-\frac{\pi}{4}}{-\frac{\pi}{4}-\frac{\pi}{4}}
			= \frac{\frac{\pi}{4}-\arctan(1-2z)}{\frac{\pi}{2}}
			= \frac{\pi}{2}\arctan\left(\frac{z}{1-z}\right)\\
			&=\frac{2}{\pi}\arctan\left(\frac{x}{y}\right),
		\end{align*}
		which recovers \eqref{ZeroCorr}.
		
		\item In the special case that $\sigma^{b}(\cdot)=\sigma^{a}(\cdot)$
		and $\rho(\cdot)\equiv\rho$, by \eqref{ZDrift}, we have
		\begin{table}[H]
			\begin{tabular}{|c|c|c|c|}
				\hline
				$u(x,y)=\theta_0/\alpha$ &                                     $y<\rho x$                                      &             $y=\rho x$             &                                   $y>\rho x$                                    \\ \hline
				        $\rho>0$         & $\frac{\pi-\arctan\left(\lambda\frac{\rho x}{\rho x-y}\right)}{\pi-\arctan\lambda}$ & $\frac{\pi/2}{\pi-\arctan\lambda}$ & $\frac{\arctan\left(\lambda\frac{\rho x}{\rho x-y}\right)}{\pi-\arctan\lambda}$ \\ \hline
				        $\rho=0$         &                                         N/A                                         &                $1$                 &        $\frac{2}{\pi}\arctan\left(\lambda\frac{\rho x}{\rho x-y}\right)$        \\ \hline
				        $\rho<0$         &                                         N/A                                         &  $\frac{\pi/2}{-\arctan\lambda}$   &  $\frac{\arctan\left(\lambda\frac{\rho x}{\rho x-y}\right)}{-\arctan\lambda}$   \\ \hline
			\end{tabular} 
		\end{table}
		where $\lambda=\sqrt{1-\rho^2}/\rho$.
	\end{enumerate}
\end{remark}

\begin{remark}
When $\rho(\cdot)\equiv -1$, we have $\lambda^{1}\equiv\lambda^{2}\equiv\lambda^{3}\equiv\lambda^{4}\equiv 0$
and $\lambda^{5}(\cdot)=\lambda^{6}(\cdot)$. And thus $\sigma^{b}(\cdot)=\sigma^{a}(\cdot)$.
By using the assumption \ref{DriftlessAssump}, we can check that the probability of mid-price moving up 
$u(x,y)=\frac{x}{x+y}$ satisfies \eqref{ContinuousPDE}
and thus this is the probability of mid-price movement for the reduced-form model. Indeed, for the original unscaled discrete model \eqref{DiscreteModel}, 
$u(x,y)$ satisfies the equation 
\begin{equation}
\lambda^{5}\left(\frac{x}{x+y}\right)[u(x+1,y-1)-u(x,y)]
+\lambda^{6}\left(\frac{x}{x+y}\right)[u(x-1,y+1)-u(x,y)]=0,
\end{equation}
for $(x,y)\in\mathbb{Z}_{\geq 0}\times\mathbb{Z}_{\geq 0}$ with boundary condition $u(0,y)=0$ and $u(x,0)=1$.
It is easy to see that $u(x,y)=\frac{x}{x+y}$. Indeed, this result is true for perfectly negatively correlated queues model-free. 
To see this, notice that $X(t)$ and $Y(t)$ are perfectly negatively correlated martingales and we can write $X(t)=x+M(t)$
and $Y(t)=y-M(t)$, where $M(t)$ is a martingale. Therefore, the probability that $X(t)$ hits $0$ before $Y(t)$ does
is the same as the probability that $M(t)$ hits $-x$ before $y$. By optional stopping theorem for martingales, this probability
can be easily computed as $\frac{x}{x+y}$.
\end{remark}

As we can see from Figures \ref{ProbBAC}, \ref{ProbGE}, \ref{ProbGM}, \ref{ProbJPM}, 
the empirical probability of mid-price moving up (dotted lines) is linearly dependent on the imbalance. 
Except the BAC at NASDAQ and the GE at NASDAQ, there is also a strong numerical evidence
of the hidden liquidity, that is, the probability of moving up is bigger than zero when the imbalance
is near zero and less than one when the imbalance is near one. 
To better fit the data, we introduce the hidden liquidity $H\in(0,1)$, so that
the theoretical probability of mid-price moving up is $H$ for imbalance at zero and $1-H$ for imbalance at 
one\footnote{This is slightly different from the definition of hidden liquidity in Avellaneda \textit{et al.} \cite{Avellaneda2011}. 
Our definition of $H$ can be interpreted as a probability, a free parameter that helps get $P_{\text{up}}$
closer to the empirical data in the least-square sense rather than a real liquidity as in Avellaneda \textit{et al.} \cite{Avellaneda2011}. We choose this definition to keep the analytical tractability of the model and
also our definition
	is simpler in the sense that we only need to bucket the data and discretize our analytical formula according
	to the imbalance level rather than the best bid and ask queue sizes as in \cite{Avellaneda2011}}.
That is, $P_{\text{up}}(z)$ satisfies the boundary condition $P_{\text{up}}(0)=H$ and $P_{\text{up}}(1)=1-H$. 
Following the proofs of Theorem \ref{ThmI}, we get the following result (the proof will
be given in Appendix \ref{ProofSection}).

\begin{theorem}
	\label{ThmII}
	Given the model (\ref{SDE}), $P_{\text{up}}(z)$, the probability of the price moving up, defined in (\ref{eq:prob of price movement}) 
	and (\ref{eq:time to price movement}) with the boundary conditions $P_{\text{up}}(0)=H$, $P_{\text{up}}(1)=1-H$, is explicitly given by
	\begin{equation}
		\label{HFormula}
		P_{\text{up}}(z)
		=H+(1-2H)\frac{\int_{0}^{z}e^{-\int_{0}^{y}\frac{\mu(x)}{\nu(x)}dx}dy}{\int_{0}^{1}e^{-\int_{0}^{y}\frac{\mu(x)}{\nu(x)}dx}dy},
	\end{equation}
	where $z=\frac{x}{x+y}$ is the imbalance and $\mu$, $\nu$ are the functions in (\ref{eq:mu and nu}). 
\end{theorem}

To fit the data, we use the empirical data for $\sigma^{b}(\cdot)$, $\sigma^{a}(\cdot)$ and $\rho(\cdot)$
to plug into the formula \eqref{HFormula} and obtain $P_{\text{theoretical}}(z,H)$, 
the theoretical probability of mid-price moving up
at imbalance level $z$ and then use the least square method
\begin{equation}
\min_{H}\sum_{z}(P_{\text{empirical}}(z)-P_{\text{theoretical}}(z,H))^{2},
\end{equation}
to find the best fitting hidden liquidity $H$. The solid lines in Figures \ref{ProbBAC}, \ref{ProbGE}, \ref{ProbGM}, \ref{ProbJPM}
are the model predictions. Remarkably, the solid lines are almost linear in imbalance even though
the correlation is a complicated function of imbalance and far away from being $-1$. 
Hence, we have built up a model with the empirical correlation and volatilities
as the input that can produce an analytical formula for the price movement probability 
that can be used to fit to the empirical probability data
\section{Conclusions}
\label{sec:conclusion}

We did numerical studies of the drift effect, correlation and volatility of the best
bid and ask queues and how they depend on the imbalance of the volumes at the best bid and ask queues
from the level-1 limit order books data from the WRDS. We discovered
that there is little evidence for the drift except when the imbalance is small or large. 
The correlation as a function of the imbalance
exhibits universal behaviors, which is either a $U$-shaped or a $W$-shaped curve, and it is almost always
negative though far away from $-1$. The volatility is much more noisy and in general lacks a clear pattern though very often
exhibit skewed $U$-shapes. All the empirical results are highly stock and also exchange dependent, which suggests
that the dynamics of the limit order books are very sensitive to their particular stock and also exchanges. 
Based on our empirical discoveries, we built up a discrete model for the dynamics of the best
bid and ask queues and showed that it can be approximated by a reduced-form diffusion model 
with functional dependence of the drift, correlation and volatility
on the imbalance, which therefore generalizes the correlated Brownian motion model that is commonly
used in the limit order books literature. Our reduced-form model still keeps analytical tractability, 
and it is self-consistent when it is fit to the data of both the empirical probability of mid-price movement
and the empirical correlation/volatility.

\section*{Acknowledgments}

The authors are very grateful to two anonymous referees and the editor whose suggestions have greatly
improved the presentation and the quality of the paper. 
The authors would also like to thank Jean-Philippe Bouchaud, Yuanda Chen, Arash Fahim, Xuefeng Gao and Alec Kercheval for their invaluable comments.
Lingjiong Zhu is partially supported by NSF Grant DMS-1613164.

\appendix

\section{Proofs}\label{ProofSection}

\begin{proof}[Proof of Theorem \ref{ConvergenceThm}]
Notice that $M^{j}(t)$ are martingales with predictable quadratic variation $\int_{0}^{t}\lambda^{j}(Z_{s-})ds$,
where $\lambda^{j}$ is bounded, i.e. $\Vert\lambda^{j}\Vert_{\infty}<\infty$.
For any $t\in[0,T-\delta]$, $\delta>0$, we have
\begin{equation}
\mathbb{E}\left[\Vert(X_{n}(t+\delta),Y_{n}(t+\delta))-(X_{n}(t),Y_{n}(t))\Vert^{4}\right]
\leq C\sum_{j=1}^{6}\mathbb{E}\left[(M^{j}_{n}(t+\delta)-M^{j}_{n}(t))^{4}\right],
\end{equation}
for some constant $C>0$.
By Burkholder-Davis-Gundy inequality, for any $1\leq j\leq 6$, 
\begin{equation}
\mathbb{E}\left[(M^{j}_{n}(t+\delta)-M^{j}_{n}(t))^{4}\right]
\leq\frac{C}{n^{2}}\mathbb{E}\left[\left(\int_{tn}^{(t+\delta)n}\lambda^{j}(Z_{s})ds\right)^{2}\right]
\leq C\Vert\lambda^{j}\Vert_{\infty}\delta^{2},
\end{equation}
for some constant $C>0$.
Therefore, by applying Kolmogorov's tightness criterion, we can show that $(X_{n}(t),Y_{n}(t))$ is tight.

The infinitesimal generator for the rescaled process $(X_{n}(t),Y_{n}(t))$ is given by
\begin{align*}
	\mathcal{L}_{n}f(x,y) &:= n\lambda^{1}(z)\left[f\left(x+\frac{1}{\sqrt{n}},y\right)-f(x,y)\right]\\
	&\quad + n\lambda^{2}(z)\left[f\left(x-\frac{1}{\sqrt{n}},y\right)-f(x,y)\right]\\
	&\quad + n\lambda^{3}(z)\left[f\left(x,y+\frac{1}{\sqrt{n}}\right)-f(x,y)\right]\\
	&\quad + n\lambda^{4}(z)\left[f\left(x,y-\frac{1}{\sqrt{n}}\right)-f(x,y)\right]\\
	&\quad + n\lambda^{5}(z)\left[f\left(x+\frac{1}{\sqrt{n}},y-\frac{1}{\sqrt{n}}\right)-f(x,y)\right]\\
	&\quad + n\lambda^{6}(z)\left[f\left(x-\frac{1}{\sqrt{n}},y+\frac{1}{\sqrt{n}}\right)-f(x,y)\right],
\end{align*}
where $z=\frac{x}{x+y}$ and $f$ is a twice continuously differentiable test function.

By using the driftless assumption \eqref{DriftlessAssump},
we get
\begin{align*}
	&\mathcal{L}_{n}f(x,y) := n\lambda^{1}(z)\left[\frac{1}{\sqrt{n}}\frac{\partial f}{\partial x}+\frac{1}{2n}\frac{\partial^{2}f}{\partial x^{2}}+O(n^{-3/2})\right]\\
	&\quad + n\lambda^{2}(z)\left[-\frac{1}{\sqrt{n}}\frac{\partial f}{\partial x} + \frac{1}{2n}\frac{\partial^{2}f}{\partial x^{2}}+O(n^{-3/2})\right]\\
	&\quad + n\lambda^{3}(z)\left[\frac{1}{\sqrt{n}}\frac{\partial f}{\partial y} + \frac{1}{2n}\frac{\partial^{2}f}{\partial y^{2}}+O(n^{-3/2})\right]\\
	&\quad + n\lambda^{4}(z)\left[-\frac{1}{\sqrt{n}}\frac{\partial f}{\partial y} + \frac{1}{2n}\frac{\partial^{2}f}{\partial y^{2}}+O(n^{-3/2})\right]\\
	&\quad + n\lambda^{5}(z)\left[\frac{1}{\sqrt{n}}\frac{\partial f}{\partial x} - \frac{1}{\sqrt{n}}\frac{\partial f}{\partial y}
			+\frac{1}{2n}\frac{\partial^{2}f}{\partial x^{2}}
			+\frac{1}{2n}\frac{\partial^{2}f}{\partial y^{2}}- \frac{1}{n}\frac{\partial^{2}f}{\partial x\partial y}+O(n^{-3/2})\right]\\
	&\quad + n\lambda^{6}(z)\left[-\frac{1}{\sqrt{n}}\frac{\partial f}{\partial x} + \frac{1}{\sqrt{n}}\frac{\partial f}{\partial y}
			+\frac{1}{2n}\frac{\partial^{2}f}{\partial x^{2}}
			+\frac{1}{2n}\frac{\partial^{2}f}{\partial y^{2}}-\frac{1}{n}\frac{\partial^{2}f}{\partial x\partial y}+O(n^{-3/2})\right]\\
	&= \frac{1}{2}\left[\lambda^{1}(z)+\lambda^{2}(z)+\lambda^{5}(z)+\lambda^{6}(z)\right]\frac{\partial^{2}f}{\partial x^{2}}
	+ \frac{1}{2}\left[\lambda^{3}(z)+\lambda^{4}(z)+\lambda^{5}(z)+\lambda^{6}(z)\right]\frac{\partial^{2}f}{\partial y^{2}}\\
	&\quad - \left[\lambda^{5}(z)+\lambda^{6}(z)\right]\frac{\partial^{2}f}{\partial x\partial y}+O(n^{-1/2}).
\end{align*}
As $n\rightarrow\infty$, $\mathcal{L}_{n}f(x,y)\rightarrow\mathcal{L}f(x,y)$, where 
\begin{align*}
	&\mathcal{L}f(x,y) 
	= \frac{1}{2}(\sigma^{b}(z))^{2}\frac{\partial^{2}f}{\partial x^{2}} 
	+ \frac{1}{2}(\sigma^{a}(z))^{2}\frac{\partial^{2}f}{\partial y^{2}}
	+ \sigma^{b}(z)\sigma^{a}(z)\rho(z)\frac{\partial^{2}f}{\partial x\partial y}\\
	&\sigma^{b}(z):=\left[\lambda^{1}(z)+\lambda^{2}(z)+\lambda^{5}(z)+\lambda^{6}(z)\right]^{1/2}\\
	&\sigma^{a}(z):=\left[\lambda^{3}(z)+\lambda^{4}(z)+\lambda^{5}(z)+\lambda^{6}(z)\right]^{1/2}\\
	&\rho(z):=-\frac{\lambda^{5}(z)+\lambda^{6}(z)}{\sigma^{b}(z)\sigma^{a}(z)}.
\end{align*}
Also notice that by our assumption, the initial condition satisfies $(X_{n}(0),Y_{n}(0))=(x,y)\in\mathbb{R}^{+}\times\mathbb{R}^{+}$.
The tightness gives the relative compactness of the sequence and and the convergence of infinitesimal generators gives
convergence in distribution for finite 
fixed time point, which guarantees the weak convergence
on $D[0,T]$, see e.g. Theorem 7.8(b)
of Chapter 3 in Ethier and Kurtz \cite{EK}. Hence $(X_{n}(t),Y_{n}(t))\Rightarrow(Q^{b}(t),Q^{a}(t))$ on $D[0,T]$.
\end{proof}

\begin{proof}[Proof of Theorem \ref{ThmI}]
Recall that the price moves up is:
\begin{equation}
	P_{\text{up}}(x,y) = u(x,y) = \mathbb{P}(\tau^{a}<\tau^{b}|Q^{b}(0)=x,Q^{a}(0)=y). 
\end{equation}
Then, $u(x,y)$ satisfies the PDE:
\begin{equation}\label{ContinuousPDE}
	\sigma^{b}(z)^2\frac{\partial^{2}u}{\partial x^{2}} + 2\rho(z)\sigma^{b}(z)\sigma^{a}(z)\frac{\partial^{2}u}{\partial x\partial y}
	+ \sigma^{a}(z)^{2}\frac{\partial^{2}u}{\partial y^{2}} = 0,
\end{equation}
with the boundary condition: $u(0,y)=0$ and $u(x,0)=1$, where $z=\frac{x}{x+y}$ is the imbalance.

Assuming that $u(x,y)$ is a function of $z$ so that $u(x,y)=u(z)$, by the chain rule,
\begin{align}
	&\frac{\partial u}{\partial x} = \frac{y}{(x+y)^{2}}u'(z),\quad \frac{\partial u}{\partial y}=-\frac{x}{(x+y)^{2}}u'(z)\\
	&\frac{\partial^{2}u}{\partial x^{2}} = -\frac{2y}{(x+y)^{3}}u'(z)+\frac{y^{2}}{(x+y)^{4}}u''(z)\\
	&\frac{\partial^{2}u}{\partial y^{2}} = \frac{2x}{(x+y)^{3}}u'(z)+\frac{x^{2}}{(x+y)^{4}}u''(z)\\
	&\frac{\partial^{2}u}{\partial x\partial y} = \frac{x-y}{(x+y)^{3}}u'(z)-\frac{xy}{(x+y)^{4}}u''(z).
\end{align}
Hence, the PDE reduces to:
\begin{multline*}
	\sigma^{b}(z)^{2}\left[-\frac{2y}{(x+y)^{3}}u'(z)+\frac{y^{2}}{(x+y)^{4}}u''(z)\right]\\
	+2\rho(z)\sigma^{b}(z)\sigma^{a}(z)\left[\frac{x-y}{(x+y)^{3}}u'(z)-\frac{xy}{(x+y)^{4}}u''(z)\right]\\
	+\sigma^{a}(z)^{2}\left[\frac{2x}{(x+y)^{3}}u'(z)+\frac{x^{2}}{(x+y)^{4}}u''(z)\right]=0,
\end{multline*}
which can be further reduced to the ODE:
\begin{multline*}
	\sigma^{b}(z)^{2}\left[-2z^{2}(1-z)u'(z)+(1-z)^{2}z^{2}u''(z)\right]\\
	+2\rho(z)\sigma^{b}(z)\sigma^{a}(z)\left[z^{2}(2z-1)u'(z)-z^{3}(1-z)u''(z)\right]\\
	+\sigma^{a}(z)^{2}\left[2z^{3}u'(z)+z^{4}u''(z)\right]=0,
\end{multline*}
with the boundary condition $u(0)=0$ and $u(1)=1$, which can be rewritten as
\begin{equation}
	\mu(z)f(z)+\nu(z)f'(z)=0,
\end{equation}
where
\begin{align*}
	\mu(z) &= -2(1-z)\sigma^{b}(z)^{2}+2(2z-1)\rho(z)\sigma^{b}(z)\sigma^{a}(z)+2z\sigma^{a}(z)^{2}\\
	\nu(z) &= (1-z)^{2}\sigma^{b}(z)^{2} - 2z(1-z)\rho(z)\sigma^{b}(z)\sigma^{a}(z) + z^{2}\sigma^{a}(z)^{2},\notag
\end{align*}
and $f(z)=u'(z)$. This is a first-order linear equation with solution
\begin{equation}
	f(z)=C_{1}e^{-\int_{0}^{z}\frac{\mu(x)}{\nu(x)}dx},
\end{equation}
and hence
\begin{equation}
	u(z) = C_{1}\int_{0}^{z}e^{-\int_{0}^{y}\frac{\mu(x)}{\nu(x)}dx}dy+C_{2},
\end{equation}
where $C_{1},C_{2}$ are two constants to be determined. 
By using the boundary conditions $u(0)=0$ and $u(1)=1$, we conclude that
\begin{equation}
	P_{\text{up}}(x,y) = P_{\text{up}}(z) = u(z) = 
	\frac{\int_{0}^{z}e^{-\int_{0}^{y}\frac{\mu(x)}{\nu(x)}dx}dy}{\int_{0}^{1}e^{-\int_{0}^{y}\frac{\mu(x)}{\nu(x)}dx}dy}.
\end{equation}
\end{proof}

\begin{proof}[Proof of Theorem \ref{ThmII}]
$u(z)$ as in Theorem \ref{ThmI} satisfies the ODE:
\begin{equation}
\mu(z)u'(z)+\nu(z)u''(z)=0,
\end{equation}
now with the boundary conditions $u(0)=H$ and $u(1)=1-H$, where
\begin{align*}
	\mu(z) &= -2(1-z)\sigma^{b}(z)^{2}+2(2z-1)\rho(z)\sigma^{b}(z)\sigma^{a}(z)+2z\sigma^{a}(z)^{2}\\
	\nu(z) &= (1-z)^{2}\sigma^{b}(z)^{2} - 2z(1-z)\rho(z)\sigma^{b}(z)\sigma^{a}(z) + z^{2}\sigma^{a}(z)^{2}.\notag
\end{align*}
As in the proof of Theorem \ref{ThmI}, this ODE has the solution of the form
\begin{equation}
u(z)=C_{1}\int_{0}^{z}e^{-\int_{0}^{y}\frac{\mu(x)}{\nu(x)}dx}dy+C_{2},
\end{equation}
where $C_{1},C_{2}$ are two constants to be determined. 
By using the boundary conditions $u(0)=H$ and $u(1)=1-H$, we conclude that
\begin{equation}
	P_{\text{up}}(x,y) = P_{\text{up}}(z) = u(z) = 
	H+(1-2H)\frac{\int_{0}^{z}e^{-\int_{0}^{y}\frac{\mu(x)}{\nu(x)}dx}dy}{\int_{0}^{1}e^{-\int_{0}^{y}\frac{\mu(x)}{\nu(x)}dx}dy}.
\end{equation}
\end{proof}
\newpage
\section{Tables}
\label{sec:table}

\begin{table}[H]
	\caption{Summary of Volume Changes (NASDAQ)}
	\centering 
	\begin{tabular}{|c|c|c|c|c|c|c|c|c|}
		\hline 
		Imbalance & BAC b & BAC a & GE b  & GE a  & GM b  & GM a  & JPM b & JPM a \\ \hline\hline
		                                 0.0-0.1                                  & 0.349 & 0.499 & 0.367 & 0.508 & 0.640 & 0.462 & 0.628 & 0.496 \\ \hline
		                                 0.1-0.2                                  & 0.449 & 0.519 & 0.461 & 0.503 & 0.508 & 0.488 & 0.580 & 0.487 \\ \hline
		                                 0.2-0.3                                  & 0.458 & 0.542 & 0.439 & 0.513 & 0.480 & 0.517 & 0.498 & 0.507 \\ \hline
		                                 0.3-0.4                                  & 0.472 & 0.556 & 0.478 & 0.536 & 0.482 & 0.532 & 0.450 & 0.542 \\ \hline
		                                 0.4-0.5                                  & 0.503 & 0.537 & 0.507 & 0.539 & 0.508 & 0.531 & 0.472 & 0.529 \\ \hline
		                                 0.5-0.6                                  & 0.533 & 0.508 & 0.529 & 0.518 & 0.532 & 0.510 & 0.514 & 0.448 \\ \hline
		                                 0.6-0.7                                  & 0.540 & 0.488 & 0.537 & 0.490 & 0.530 & 0.489 & 0.518 & 0.477 \\ \hline
		                                 0.7-0.8                                  & 0.540 & 0.467 & 0.528 & 0.469 & 0.514 & 0.471 & 0.527 & 0.514 \\ \hline
		                                 0.8-0.9                                  & 0.522 & 0.479 & 0.523 & 0.472 & 0.494 & 0.498 & 0.516 & 0.527 \\ \hline
		                                 0.9-1.0                                  & 0.486 & 0.310 & 0.496 & 0.378 & 0.463 & 0.643 & 0.508 & 0.696 \\ \hline
	\end{tabular}
	\label{TTable} 
\end{table}

\begin{table}[H]
	\caption{Summary of Volume Changes (NYSE)}
	\centering 
	\begin{tabular}{|c|c|c|c|c|c|c|c|c|}
		\hline 
		Imbalance & BAC b & BAC a & GE b  & GE a  & GM b  & GM a  & JPM b & JPM a \\ \hline\hline
		                                 0.0-0.1                                  & 0.561 & 0.540 & 0.586 & 0.596 & 0.730 & 0.598 & 0.723 & 0.602 \\ \hline
		                                 0.1-0.2                                  & 0.524 & 0.531 & 0.502 & 0.552 & 0.557 & 0.506 & 0.532 & 0.531 \\ \hline
		                                 0.2-0.3                                  & 0.514 & 0.529 & 0.490 & 0.531 & 0.502 & 0.497 & 0.491 & 0.512 \\ \hline
		                                 0.3-0.4                                  & 0.518 & 0.541 & 0.505 & 0.522 & 0.494 & 0.497 & 0.480 & 0.511 \\ \hline
		                                 0.4-0.5                                  & 0.517 & 0.553 & 0.518 & 0.529 & 0.485 & 0.486 & 0.492 & 0.507 \\ \hline
		                                 0.5-0.6                                  & 0.539 & 0.511 & 0.539 & 0.518 & 0.501 & 0.491 & 0.510 & 0.494 \\ \hline
		                                 0.6-0.7                                  & 0.551 & 0.506 & 0.528 & 0.513 & 0.496 & 0.496 & 0.499 & 0.485 \\ \hline
		                                 0.7-0.8                                  & 0.534 & 0.516 & 0.541 & 0.503 & 0.498 & 0.503 & 0.506 & 0.491 \\ \hline
		                                 0.8-0.9                                  & 0.518 & 0.535 & 0.553 & 0.512 & 0.510 & 0.550 & 0.527 & 0.539 \\ \hline
		                                 0.9-1.0                                  & 0.554 & 0.582 & 0.575 & 0.551 & 0.586 & 0.718 & 0.598 & 0.742 \\ \hline
	\end{tabular}
	\label{NTable} 
\end{table}

\begin{table}[H]
	\caption{Summary of Correlation of the Best Bid and Ask}
	\rotatebox{90}{
	\begin{tabular}{|c|c|c|c|c|c|c|c|c|}
		\hline 
		Imbalance & BAC(T) & BAC(N) & GE(T) & GE(N) & GM(T) & GM(N) & JPM(T) & JPM(N) \\ \hline\hline
		                                0.00-0.05                                 & -0.03  &  0.00  & -0.05 & -0.09 & -0.09 & -0.07 & -0.14  & -0.11  \\ \hline
		                                0.05-0.10                                 & -0.12  &  0.02  & -0.08 & -0.10 & -0.11 & -0.18 & -0.13  & -0.20  \\ \hline
		                                0.10-0.15                                 & -0.09  & -0.02  & -0.09 & -0.10 & -0.14 & -0.25 & -0.17  & -0.29  \\ \hline
		                                0.15-0.20                                 & -0.09  & -0.06  & -0.13 & -0.18 & -0.15 & -0.21 & -0.17  & -0.29  \\ \hline
		                                0.20-0.25                                 & -0.07  & -0.09  & -0.10 & -0.16 & -0.16 & -0.31 & -0.21  & -0.34  \\ \hline
		                                0.25-0.30                                 & -0.05  & -0.13  & -0.09 & -0.17 & -0.17 & -0.30 & -0.22  & -0.34  \\ \hline
		                                0.30-0.35                                 & -0.04  & -0.13  & -0.07 & -0.19 & -0.19 & -0.28 & -0.24  & -0.32  \\ \hline
		                                0.35-0.40                                 & -0.04  & -0.12  & -0.06 & -0.19 & -0.18 & -0.28 & -0.23  & -0.34  \\ \hline
		                                0.40-0.45                                 & -0.04  & -0.15  & -0.09 & -0.21 & -0.20 & -0.29 & -0.24  & -0.33  \\ \hline
		                                0.45-0.50                                 & -0.05  & -0.15  & -0.07 & -0.22 & -0.20 & -0.28 & -0.25  & -0.33  \\ \hline
		                                0.50-0.55                                 & -0.04  & -0.12  & -0.06 & -0.20 & -0.20 & -0.27 & -0.25  & -0.34  \\ \hline
		                                0.55-0.60                                 & -0.04  & -0.12  & -0.07 & -0.20 & -0.21 & -0.28 & -0.27  & -0.26  \\ \hline
		                                0.60-0.65                                 & -0.05  & -0.12  & -0.09 & -0.17 & -0.20 & -0.29 & -0.24  & -0.29  \\ \hline
		                                0.65-0.70                                 & -0.05  & -0.11  & -0.08 & -0.18 & -0.19 & -0.28 & -0.24  & -0.28  \\ \hline
		                                0.70-0.75                                 & -0.04  & -0.10  & -0.07 & -0.15 & -0.19 & -0.28 & -0.22  & -0.30  \\ \hline
		                                0.75-0.80                                 & -0.07  & -0.09  & -0.11 & -0.16 & -0.16 & -0.27 & -0.20  & -0.29  \\ \hline
		                                0.80-0.85                                 & -0.11  & -0.07  & -0.12 & -0.16 & -0.16 & -0.25 & -0.22  & -0.26  \\ \hline
		                                0.85-0.90                                 & -0.12  & -0.04  & -0.10 & -0.15 & -0.12 & -0.22 & -0.20  & -0.18  \\ \hline
		                                0.90-0.95                                 & -0.10  & -0.03  & -0.09 & -0.07 & -0.08 & -0.19 & -0.16  & -0.17  \\ \hline
		                                0.95-1.00                                 & -0.06  &  0.00  & -0.04 & -0.10 & -0.03 & -0.08 & -0.06  & -0.06  \\ \hline
	\end{tabular}
	}
	\label{corrTable} 
\end{table}

\begin{table}[H]
	\caption{Summary of Empirical Probability (E) and Model Prediction (P) for Stocks Traded on NASDAQ}
	\centering 
	\rotatebox{90}{
	\begin{tabular}{|c|c|c|c|c|c|c|c|c|}
		\hline 
		Imbalance & BAC(E) & BAC(P) & GE(E) & GE(P) & GM(E) & GM(P) & JPM(E) & JPM(P) \\ \hline\hline
		                                0.00-0.05                                 & 0.044  & 0.072  & 0.053 & 0.104 & 0.102 & 0.176 & 0.077  & 0.158  \\ \hline
		                                0.05-0.10                                 & 0.139  & 0.131  & 0.127 & 0.144 & 0.197 & 0.211 & 0.137  & 0.196  \\ \hline
		                                0.10-0.15                                 & 0.136  & 0.181  & 0.168 & 0.186 & 0.245 & 0.247 & 0.225  & 0.235  \\ \hline
		                                0.15-0.20                                 & 0.159  & 0.224  & 0.186 & 0.229 & 0.267 & 0.285 & 0.254  & 0.275  \\ \hline
		                                0.20-0.25                                 & 0.215  & 0.263  & 0.221 & 0.273 & 0.306 & 0.322 & 0.285  & 0.315  \\ \hline
		                                0.25-0.30                                 & 0.270  & 0.301  & 0.231 & 0.319 & 0.342 & 0.359 & 0.336  & 0.356  \\ \hline
		                                0.30-0.35                                 & 0.307  & 0.337  & 0.308 & 0.364 & 0.367 & 0.394 & 0.378  & 0.396  \\ \hline
		                                0.35-0.40                                 & 0.351  & 0.374  & 0.364 & 0.409 & 0.412 & 0.429 & 0.406  & 0.436  \\ \hline
		                                0.40-0.45                                 & 0.426  & 0.410  & 0.397 & 0.454 & 0.463 & 0.463 & 0.460  & 0.475  \\ \hline
		                                0.45-0.50                                 & 0.441  & 0.447  & 0.441 & 0.498 & 0.482 & 0.497 & 0.484  & 0.514  \\ \hline
		                                0.50-0.55                                 & 0.492  & 0.484  & 0.504 & 0.542 & 0.503 & 0.531 & 0.540  & 0.553  \\ \hline
		                                0.55-0.60                                 & 0.539  & 0.523  & 0.545 & 0.586 & 0.533 & 0.566 & 0.591  & 0.591  \\ \hline
		                                0.60-0.65                                 & 0.628  & 0.562  & 0.600 & 0.631 & 0.570 & 0.602 & 0.607  & 0.630  \\ \hline
		                                0.65-0.70                                 & 0.658  & 0.603  & 0.653 & 0.676 & 0.604 & 0.639 & 0.639  & 0.668  \\ \hline
		                                0.70-0.75                                 & 0.704  & 0.646  & 0.669 & 0.722 & 0.644 & 0.677 & 0.645  & 0.706  \\ \hline
		                                0.75-0.80                                 & 0.737  & 0.693  & 0.749 & 0.766 & 0.679 & 0.716 & 0.668  & 0.744  \\ \hline
		                                0.80-0.85                                 & 0.783  & 0.750  & 0.762 & 0.811 & 0.698 & 0.754 & 0.722  & 0.780  \\ \hline
		                                0.85-0.90                                 & 0.844  & 0.825  & 0.799 & 0.855 & 0.740 & 0.791 & 0.791  & 0.815  \\ \hline
		                                0.90-0.95                                 & 0.885  & 0.912  & 0.846 & 0.896 & 0.810 & 0.825 & 0.799  & 0.848  \\ \hline
		                                0.95-1.00                                 & 0.960  & 1.000  & 0.930 & 0.933 & 0.902 & 0.856 & 0.918  & 0.878  \\ \hline
	\end{tabular}
	}
	\label{probTableNASDAQ} 
\end{table}

\begin{table}[H]
	\caption{Summary of Empirical Probability (E) and Model Prediction (P) for Stocks Traded on NYSE }
	\centering 
	\rotatebox{90}{
	\begin{tabular}{|c|c|c|c|c|c|c|c|c|}
		\hline 
		Imbalance & BAC(E) & BAC(P) & GE(E) & GE(P) & GM(E) & GM(P) & JPM(E) & JPM(P) \\ \hline\hline
		                                0.00-0.05                                 & 0.128  & 0.191  & 0.138 & 0.251 & 0.157 & 0.217 & 0.125  & 0.179  \\ \hline
		                                0.05-0.10                                 & 0.194  & 0.216  & 0.237 & 0.271 & 0.221 & 0.244 & 0.174  & 0.209  \\ \hline
		                                0.10-0.15                                 & 0.249  & 0.243  & 0.332 & 0.292 & 0.277 & 0.272 & 0.219  & 0.241  \\ \hline
		                                0.15-0.20                                 & 0.249  & 0.272  & 0.346 & 0.315 & 0.310 & 0.302 & 0.272  & 0.275  \\ \hline
		                                0.20-0.25                                 & 0.276  & 0.303  & 0.347 & 0.340 & 0.341 & 0.333 & 0.313  & 0.310  \\ \hline
		                                0.25-0.30                                 & 0.335  & 0.338  & 0.367 & 0.367 & 0.373 & 0.365 & 0.348  & 0.347  \\ \hline
		                                0.30-0.35                                 & 0.377  & 0.375  & 0.384 & 0.396 & 0.393 & 0.398 & 0.428  & 0.385  \\ \hline
		                                0.35-0.40                                 & 0.405  & 0.415  & 0.395 & 0.428 & 0.422 & 0.432 & 0.411  & 0.425  \\ \hline
		                                0.40-0.45                                 & 0.420  & 0.458  & 0.435 & 0.461 & 0.443 & 0.466 & 0.439  & 0.465  \\ \hline
		                                0.45-0.50                                 & 0.442  & 0.502  & 0.435 & 0.495 & 0.454 & 0.501 & 0.474  & 0.506  \\ \hline
		                                0.50-0.55                                 & 0.468  & 0.546  & 0.476 & 0.529 & 0.516 & 0.536 & 0.522  & 0.545  \\ \hline
		                                0.55-0.60                                 & 0.490  & 0.588  & 0.520 & 0.562 & 0.554 & 0.570 & 0.569  & 0.584  \\ \hline
		                                0.60-0.65                                 & 0.530  & 0.628  & 0.540 & 0.593 & 0.572 & 0.603 & 0.608  & 0.623  \\ \hline
		                                0.65-0.70                                 & 0.590  & 0.665  & 0.580 & 0.622 & 0.601 & 0.636 & 0.637  & 0.661  \\ \hline
		                                0.70-0.75                                 & 0.648  & 0.700  & 0.603 & 0.649 & 0.638 & 0.667 & 0.664  & 0.696  \\ \hline
		                                0.75-0.80                                 & 0.649  & 0.731  & 0.619 & 0.675 & 0.659 & 0.698 & 0.701  & 0.730  \\ \hline
		                                0.80-0.85                                 & 0.723  & 0.760  & 0.647 & 0.699 & 0.701 & 0.727 & 0.730  & 0.763  \\ \hline
		                                0.85-0.90                                 & 0.793  & 0.787  & 0.660 & 0.723 & 0.729 & 0.756 & 0.771  & 0.794  \\ \hline
		                                0.90-0.95                                 & 0.832  & 0.811  & 0.763 & 0.745 & 0.785 & 0.783 & 0.812  & 0.823  \\ \hline
		                                0.95-1.00                                 & 0.924  & 0.833  & 0.860 & 0.766 & 0.856 & 0.808 & 0.875  & 0.849  \\ \hline
	\end{tabular}
	}
	\label{probTableNYSE} 
\end{table}

\bibliographystyle{plain}
\bibliography{reference.bib}

\begin{thebibliography}{10}

\bibitem{ABERGEL2013}
Fr\'{e}d\'{e}ric Abergel and Aymen Jedidi.
\newblock A mathematical approach to order book modeling.
\newblock {\em International Journal of Theoretical and Applied Finance},
  16(05):1350025, 2013.

\bibitem{Avellaneda2011}
Marco Avellaneda, Josh Reed, and Sasha Stoikov.
\newblock Forecasting prices from level-{I} quotes in the presence of hidden
  liquidity.
\newblock {\em Algorithmic Finance}, 1(1), 2011.

\bibitem{Cont2012}
Rama Cont and Adrien De~Larrard.
\newblock Order book dynamics in liquid markets: limit theorems and diffusion
  approximations.
\newblock {\em Available at SSRN 1757861}, 2012.

\bibitem{Cont2013}
Rama Cont and Adrien de~Larrard.
\newblock Price dynamics in a {M}arkovian limit order market.
\newblock {\em SIAM Journal on Financial Mathematics}, 4(1):1--25, 2013.

\bibitem{Cont2010}
Rama Cont, Sasha Stoikov, and Rishi Talreja.
\newblock A stochastic model for order book dynamics.
\newblock {\em Operations Research}, 58(3):549--563, 2010.

\bibitem{EK}
S.~N. Ethier and T.~G. Kurtz.
\newblock {\em Markov Processes: Characterization and Convergence}.
\newblock Wiley-Interscience, New Jersey, second edition, 2005.

\bibitem{Gareche2013}
A.~Gar\`eche, G.~Disdier, J.~Kockelkoren, and J.-P. Bouchaud.
\newblock {F}okker-{P}lanck description for the queue dynamics of large tick
  stocks.
\newblock {\em Phys. Rev. E}, 88:032809, Sep 2013.

\bibitem{Guo2015}
Xin Guo, Zhao Ruan, and Lingjiong Zhu.
\newblock Dynamics of order positions and related queues in a limit order book.
\newblock {\em arXiv preprint arXiv:1505.04810}, 2015.

\bibitem{Huang2012}
He~Huang and Alec~N. Kercheval.
\newblock A generalized birth-death stochastic model for high-frequency order
  book dynamics.
\newblock {\em Quantitative Finance}, 12(4):547--557, 2012.

\bibitem{HLR}
Weibing Huang, Charles-Albert Lehalle, and Mathieu Rosenbaum.
\newblock Simulating and analyzing order book data: The queue-reactive model.
\newblock {\em Journal of the American Statistical Association}, 110:107--122,
  2015.

\bibitem{Iyengar1985}
Satish Iyengar.
\newblock Hitting lines with two-dimensional {B}rownian motion.
\newblock {\em SIAM Journal on Applied Mathematics}, 45(6):983--989, 1985.

\bibitem{LL}
Charles-Albert Lehalle and Sophie Laruelle.
\newblock {\em Market Microstructure in Practice}.
\newblock World Scientific, Singapore, first edition, 2013.

\bibitem{Macey2015}
Jonathan Macey and David Swensen.
\newblock The cure for stock-market fragmentation: {M}ore exchanges.
\newblock {\em The Wall Street Journal}, May 31 2015.

\bibitem{Metzler2010}
Adam Metzler.
\newblock On the first passage problem for correlated {B}rownian motion.
\newblock {\em Statistics \& Probability Letters}, 80(5-6):277--284, 2010.

\bibitem{Oksendal2003}
Bernt {\O}ksendal.
\newblock {\em Stochastic differential equations}.
\newblock Universitext. Springer-Verlag, Berlin, sixth edition, 2003.
\newblock An introduction with applications.

\bibitem{StroockVaradhan}
D.~W. Stroock and S.~R.~S. Varadhan.
\newblock {\em Multidimensional diffusion processes}.
\newblock Springer-Verlag, Berlin, 1979.

\end{thebibliography}

\end{document}